\documentclass[11pt, draftclsnofoot, peerreview, onecolumn]{IEEEtran}

\usepackage{bbm}
\usepackage{lineno}
\usepackage{amssymb}
\usepackage{amsmath}
\usepackage{amsthm}
\usepackage{epsfig}
\usepackage{graphicx}
\usepackage{graphics}
\usepackage{float}
\usepackage{subfigure}
\usepackage{multirow}
\usepackage{color}

\usepackage{makeidx}
\usepackage{xspace}
\usepackage{wrapfig}
\usepackage{lipsum}
\usepackage{mathtools}
\usepackage{cuted}
\usepackage{cite}
\usepackage{amsfonts}
\usepackage{textcomp}
\usepackage{enumerate}

\usepackage{textcomp}
\usepackage{enumerate}
\usepackage{algpseudocode}
\usepackage{algorithm}

\usepackage[dvipsnames]{xcolor}

\makeindex
\theoremstyle{plain}

\newtheorem{Theorem}{Theorem}
\newtheorem{Proposition}{Proposition}
\theoremstyle{definition}
\newtheorem{Definition}{Definition}
\newtheorem{corollary}{Corollary}
\newtheorem{Lemma}{Lemma}

\newtheorem{remark}{Remark}
\newtheorem{Proof of Lemma}{Proof of Lemma}

\makeatletter

\renewcommand\@endtheorem{\vvv@endmarker\endtrivlist\@endpefalse}
\newcommand\vvv@endmarker{%
  {\unskip\nobreak\hfil\penalty50
  \hskip2em\vadjust{}\nobreak\hfil\openbox
  \parfillskip=0pt \finalhyphendemerits=0 \par
  \penalty 10000 \parskip=0pt\noindent}\ignorespaces}
\makeatother

\definecolor{darkred}{rgb}{1, 0.1, 0.3}
\definecolor{darkblue}{rgb}{0.1, 0.1, 1}
\definecolor{darkgreen}{rgb}{0,0.6,0.5}

\def \T {\mathcal{T}}
\def \X {\mathcal{X}}

\def \W {\mathcal{W}}

\def \I {\mathcal{I}}
\def \J {\mathcal{J}}
\def \F {\mathbb{F}}

\def \C {\mathcal{C}}

\def \x {\mathbf{x}}
\def \y {\mathbf{y}}

\allowdisplaybreaks

\def\BibTeX{{\rm B\kern-.05em{\sc i\kern-.025em b}\kern-.08em
    T\kern-.1667em\lower.7ex\hbox{E}\kern-.125emX}}

\cleardoublepage
\begin{document}
    \title{Multi-User Linearly-Separable\\ Distributed Computing\\
      \thanks{%
     This work was supported by the European Research Council (ERC)
through the EU Horizon 2020 Research and Innovation Program under Grant
725929 (Project DUALITY).}
     \thanks{%
     The authors are with the Communication Systems Department at
EURECOM, 450 Route des Chappes, 06410 Sophia Antipolis, France (email:
ali.khalesi@eurecom.fr; elia@eurecom.fr).}
    }

    \author{\IEEEauthorblockN{Ali Khalesi, Petros Elia}}
    \maketitle

\maketitle
\begin{abstract}
    In this work, we explore the problem of multi-user linearly-separable distributed computation, where $N$ servers help compute the desired functions (jobs) of $K$ users, and where each desired function can be written as a linear combination of up to $L$ (generally non-linear) subtasks (or sub-functions). Each server computes some of the subtasks, communicates a function of its computed outputs to some of the users, and then each user collects its received data to recover its desired function.  We explore the computation and communication relationship between how many servers compute each subtask vs. how much data each user receives.

    For a matrix $\mathbf{F}$ representing the linearly-separable form of the set of requested functions, our problem becomes equivalent to the open problem of sparse matrix factorization $\mathbf{F} = \mathbf{D}\mathbf{E}$ over finite fields, where a sparse decoding matrix $\mathbf{D}$ and encoding matrix $\mathbf{E}$ imply reduced communication and computation costs respectively. This paper establishes a novel relationship between our distributed computing problem, matrix factorization, syndrome decoding and covering codes. To reduce the computation cost, the above $\mathbf{D}$ is drawn from covering codes or from a here-introduced class of so-called `partial covering' codes, whose study here yields computation cost results that we present. To then reduce the communication cost, these coding-theoretic properties are explored in the regime of codes that have low-density parity check matrices. The work reveals --- first for the commonly used one-shot scenario --- that in the limit of large $N$, the optimal normalized computation cost $\gamma  \in (0,1)$ is in the range $\gamma  \in (H_{q}^{-1}(\frac{\log_q(L)}{N}), H_q^{-1}(K/N))$ --- where $H_{q}$ is the $q$-ary entropy function --- and that this can be achieved with normalized communication cost that vanishes as $\sqrt{\log_q(N)}/N$. The above reveals an unbounded coding gain over the uncoded scenario, as well as reveals the role of a certain \emph{functional rate} $\log_q(L)/N$ and \emph{functional capacity} $H_{q}(\gamma)$ of the system. In the end, we also explore the multi-shot scenario, for which we derive bounds on the computation cost.
\end{abstract}

\begin{IEEEkeywords}
\textbf{Distributed computing, Linearly separable
functions, Coding theory, Sparse matrix factorization, Covering codes, Straggler mitigation, Distributed gradient coding.}
\end{IEEEkeywords}

\section{Introduction}
Distributed computing plays an ever-increasing role in speeding up non-linear and computationally-hard computing tasks. As the complexity of these tasks increases, research seeks novel parallel processing techniques to efficiently offload computations to groups of distributed servers, under various frameworks such as MapReduce \cite{dean2008mapreduce} and Spark \cite{zaharia2010spark}. Distributed computing naturally entails several challenges that involve accuracy~\cite{jahani2021codedsketch,wang2021price,ozfatura2021coded}, scalability~\cite{li2017scalable,haddadpour2019trading,yang2020coded,charalambides2021numerically,soleymani2021analog}, privacy and security~\cite{sun2018capacity,soleymani2020distributed,khalesi2021capacity,soleymani2020privacy,soleymani2021list,bitar2022adaptive,hofmeister2021secure,akbari2021secure,jia2021capacity,chen2021gcsa,yang2021coded,yu2020coded,xhemrishi2022distributed}, as well as latency and straggler mitigation~\cite{raviv2020gradient,lee2017speeding,egger2022efficient,kai1,yu2020straggler,yu2017polynomial,jia2021cross,behrouzi2020efficient}.  For a detailed survey of related research works, the reader is referred to~\cite{ng2020survey,CIT-103}. \nocite{8051074,8437333}

This aforementioned effort to efficiently distribute computation load across multiple servers, is intimately intertwined with the concept of communication complexity which refers to the amount of communication required to solve a computation problem when the desired task is distributed among two or more parties~\cite{yao2009communication}. This celebrated computation-vs-communication relationship has been studied in a variety of different forms and scenarios~\cite{verbraeken2020survey,ulukus2022private,wang2018fundamental,li2017fundamental,wan2022cache,yu2017polynomial2,dutta2019optimal,reisizadeh2021codedreduce,woolsey2021new,woolsey2021coded,woolsey2021practical,egger2022efficient,chen2021distributed} for various types of problems.

 \paragraph{Preliminary description of setting}
This same relationship between computation and communication costs, is the topic of interest in our work here for the very broad and practical setting of multi-user, multi-server computation of linearly-separable functions.
Such functions appear in several classes of problems such as for example in training large-scale machine learning algorithms and deep neural networks with massive data~\cite{verbraeken2020survey}, where indeed both computation and communication costs are crucial~\cite{zinkevich2010parallelized,chilimbi2014project}.

In particular, our setting here considers a master node that manages $N$ server nodes that must contribute in a distributed manner to the computation of the desired functions of $K$ different users. Under the linearly-separable assumption (cf.~\cite{kai1}), we consider that user $k \in \{1,2,\hdots, K\}$ demands a function $F_{k}(D_1,D_2,\hdots,D_L)$ that takes as input $L$ datasets $D_1,D_2,\hdots,D_L$, and that each such requested function takes the basic form
\begin{align} \label{linearlySep1}
    F_{k}(D_1,\hdots,D_L) = \sum^{L}_{\ell=1} f_{k,\ell} f_{\ell}(D_\ell) = \sum^{L}_{\ell=1} f_{k,\ell} W_\ell
\end{align}
where in the above, $W_\ell = f_\ell(D_\ell)$ denotes the computed output of a subfunction when the input is $D_\ell$, and where $f_{k,\ell}$ are the combining coefficients which belong, together with the entries of $W_\ell$, in some finite field\footnote{The setting nicely includes the case where each $F_k$ itself is a linear combination of some linearly separable functions, i.e., where $F_k$ can itself be written as $ F_{k}(D_1,\hdots,D_L) = \sum^{L}_{\ell=1}\sum^{M}_{i=1} f_{k,\ell,i} f_{\ell,i}(D_\ell) = \sum^{L}_{\ell=1}\sum^{M}_{i=1} f_{k,\ell,i}  W_{\ell,i}$, corresponding to some set of basis subfunctions $f_{\ell,i}(D_\ell)$. For simplicity we will henceforth refer to the model in~\eqref{linearlySep1}.}.
Upon notification of the users' requests --- where these requests are jointly described by the $K\times L$ matrix $\mathbf{F}$ that contains the different coefficients $f_{k,\ell}$ --- the master instructs the servers to compute some of the subfunctions $f_\ell(D_\ell)$. Each server may naturally compute a different number of functions. Upon completing its computations, each server communicates linear combinations of its locally computed outputs (files) to carefully selected subsets of users. Each user can then only linearly combine what it receives from the servers that have transmitted to it, and the goal is for each user to recover its desired function. The problem is completed when every user $k$ retrieves its desired $F_{k}(D_1,\hdots,D_L)$.

We note that there is a clear differentiation between the server nodes that are asked to compute hard (generally non-linear) component functions (subfunctions), and the users that can only linearly combine their received outputs. Generating the so-called output files $W_\ell =  f_\ell(D_\ell), \ell \in \{1,2,\dots,L\}$, can be the result of a computationally intensive task that may for example relate to training a deep learning model on a dataset, or it can relate to the distributed gradient coding problem \cite{tandon2017gradient, ye2018communication,raviv2020gradient,halbawi2018improving}, the distributed linear-transform computation problem~\cite{dutta2016short,wang2018fundamental}, or even the distributed matrix multiplication and
the distributed multivariate polynomial computation problems~\cite{ramamoorthy2019universally,das2019distributed,haddadpour2018codes,lee2017speeding,yu2017polynomial,wang2018coded,yu2020straggler,dutta2019optimal,ramamoorthy2020straggler,jia2021cross}.

\paragraph{Brief summary of the basic ingredients of the problem}

Our setting brings to the fore the following crucial questions.
\begin{itemize}
    \item How many and which servers must compute each subfunction $f_\ell(D_\ell)$?
    \begin{itemize}
    \item
    This decision defines the computation cost: the more the servers that compute a subfunction, the higher the computation cost. The extreme centralized scenario where each active server would compute all $L$ sub-functions, would imply a maximal computation cost, but a minimal communication cost, equal to (as we can see) one transmission received per user. The other extreme scenario (for the case of $L=N$) would imply a minimal computation cost of $1$ subfunction per server, but a maximal communication cost of $N$ shots received per user.
    \end{itemize}
    \item What linear combinations of its computed outputs must each server generate?
        \begin{itemize}
            \item These linear combination coefficients in question, define an $N\times L$ matrix $\mathbf{E}$ that describes which servers compute each subfunction, and how each server combines its computed outputs in order to transmit them. This matrix must be designed in consideration of the requested functions, which are themselves described by the aforementioned $K\times L$ matrix $\mathbf{F}$.
            \item The number of non-zero elements in $\mathbf{E}$ reflects the computation cost on the collective of servers.
        \end{itemize}
    \item What fraction of the servers must each user get data from, and from which servers?
    \begin{itemize}
        \item This defines the communication cost. The more data each user gets, the higher the cost.
    \end{itemize}
    \item How must each user combine (linearly decode) the computed outputs arriving from the servers?
    \begin{itemize}
        \item This step is determined by a $K\times N$ decoding matrix $\mathbf{D}$ that must be carefully designed. The number of non-zero elements of $\mathbf{D}$ reflects our communication cost. Having a non-sparse $\mathbf{D}$, implies the need to activate a substantial fraction of the existing communication links.
    \end{itemize}
    \item How sparse can $\mathbf{D}$ and $\mathbf{E}$ be so that each user recovers their desired function?
    \begin{itemize}
        \item This defines the overall costs in computation and communication. As one might expect, the larger the number $L$ of possible subtasks/datasets, the higher the worst-case costs. Having a larger $L$ allows the computing service to provide more refined computations, conceivably though at a higher cost.
    \end{itemize}
\end{itemize}

To answer these questions, we take a novel approach that employs coding theory. The general idea behind our approach is described as follows.
\paragraph{Brief summary of the new connection to sparse matrix factorization and covering codes}
\begin{itemize}
\item \emph{Connection to the problem of matrix factorization into sparse components:} First, when exploring our distributed computing problem, one can see that the feasibility conditions that ensure that each user recovers its desired function, constitute in fact a (preferably sparse) matrix factorization problem of the form
 \begin{align}
     \mathbf{D} \mathbf{E} = \mathbf{F}\label{main-intro}
 \end{align}
 where the problem is over some $q$-sized finite field $\mathbb{F}$, and where any potential sparsity of $\mathbf{D}$ and $\mathbf{E}$ translates to savings in communication and computation costs respectively.

\item \emph{Connection to coding theory and syndrome decoding:} To then resolve this problem in a manner that yields non-trivial sparse factors, we notice that --- if for example, we were to fix the above matrix $\mathbf{D}$, and associate this to the parity-check matrix of some linear code --- then for each column $\mathbf{E}_\ell$ of $\mathbf{E}$ and associated column $\mathbf{F}_\ell$ of $\mathbf{F}$, the corresponding equation $\mathbf{D}\cdot \mathbf{E}_\ell=\mathbf{F}_\ell$ would tells us that the desired sparse $\mathbf{E}_\ell$ can be the lowest-weight coset leader whose syndrome is equal to $\mathbf{F}_\ell$. Hence, under this analogy, the columns of $\mathbf{E}$ are associated to error vectors, the columns of $\mathbf{F}$ to the corresponding syndromes, and $\mathbf{D}$ is assigned the role of a parity check matrix, and the question is of which code?
\item \emph{Connection to covering codes and the new class of partial covering codes:} The above connection with syndromes, in turn brings about the concept of covering codes that refer to codes with good covering properties, which in turn entail low weight $\mathbf{E}_\ell$, which is what we need. In (error-control) coding theory though --- which generally considers that any error vector is possible --- such covering codes consider a full space of possible syndromes, i.e., consider the case where any appropriately-dimensioned vector can indeed be a syndrome. To account for the fact  that $\mathbf{F}$ corresponds to a \emph{restricted} set of syndromes (only those that correspond to the columns of our $\mathbf{F}$), we here consider a new class of \emph{partial covering codes}, the analysis of which is part of this work.
\item \emph{Connection with codes having low-density parity-check matrices:} The above effort yields a sparse $\mathbf{E}$. Our effort is concluded when the aforementioned exploration of covering codes and partial covering codes (which yielded a sparse $\mathbf{E}$), is extended to involve analysis of codes with a sparse $\mathbf{D}$ as well.
\item \emph{Extending the one-shot scenario:} Our framework allows us to address but also extend the one-shot scenario which is the scenario of choice in various works (see for example~\cite{kai1}) and which, in our case, asks that each server can send only one linear combination to one set of users. We extend this model to the practical and realistic scenario where, for a fixed subset of subfunctions/files $\{f(D_\ell)\}$ computed locally at each server, the server can
communicate linear combinations to various sets of users.
\end{itemize}

 \paragraph{Highlights of contributions}
Our focus is on establishing the normalized computation\footnote{Both communication and computation costs will be defined in more detail later on. Also, in the following, $\omega(\cdot )$ represents the well known Hamming weight of the argument vector or matrix.} cost $\gamma = \frac{1}{N}\underset{l \in \{1,\dots,L\}}{\max} \:\omega(\mathbf{E}(:,l))$, and the normalized communication cost $\delta = \omega(\mathbf{D})/KN$. In our setting, $\gamma \in(0,1]$ represents the maximum fraction of all servers that must compute any one subfunction, while $\delta \in(0,1]$ represents the average fraction of servers that each user gets data from, which in turn simply implies an average number of $\Delta = \delta N$  `symbols' received by each user.

We first consider the one-shot case. We proceed to highlight some of the derived results, whose exact statement can be found in the following sections.

\begin{itemize}
\item Theorem~\ref{Bridge} makes the connection between coding theory and our distributed computing problem, by showing that {\color{black}a $(\gamma,\delta)$-feasible distributed computing scheme exists if and only if the decoding matrix $\mathbf{D}$ has degree of sparsity $\delta$ and is the parity check matrix of an $N$-length code $\mathcal{C}\subset \F^{N}$ over a field $\F$ where this code has minimum normalized distance from each vector $\{\mathbf{x} \in \F^{N}| \mathbf{D} \mathbf{x} = \mathbf{F}(:,\ell), \ell \in \{1,\dots,L\}\}$ that is at most $\gamma N$. }This brings to the fore the concept of covering and partial covering codes, where covering codes are codes that guarantee a minimum distance to each vector of the entire vector space, while partial covering codes must guarantee a minimum distance to only a specific subset of the entire space. Establishing the properties of such codes is key to our problem.
     \item Theorem~\ref{Achievability} shows that in the limit of large $N$, the optimal computation cost per server is in the range $\gamma \in (H_{q}^{-1}(\frac{\log_q(L)}{N}), H_q^{-1}(K/N))$, where $H_{q}$ is the entropy function over our field of size $q$. This theorem reveals the role of what one might refer to as the \emph{functional rate} $R_f  = \log_q(L)/N$. The higher this rate, the more `involved' is the space of functions we can compute. In this sense --- given that, from the above, $\frac{log_q(L)}{N} \leq H_{q}(\gamma)$ --- the expression $H_{q}(\gamma)$ plays the role of an upper bound on what one might call the \emph{functional capacity} of the system.
    \item {\color{black} Extending the famous covering codes theorem of Blinovskii from}~\cite{blinovskii1987lower}, we established  our bounds on partial covering codes to the setting of codes with low density parity check matrices, revealing that any aforementioned achievable computation cost $\gamma$, can be achieved with normalized communication cost that vanishes\footnote{We will henceforth use $\doteq$ to denote asymptotic optimality. This will be clarified later on.} as $\delta \doteq \sqrt{\log_q(N)}/N$. This latter cost will be unboundedly lower than in the uncoded approach of resource-sharing between the two extreme regimes discussed previously in the introduction (See Figure~\ref{Comparison-fig} in Section~\ref{Discussion}). As a consequence, we can talk of an unbounded coding gain in our distributed computing problem.
     \item We also consider the multi-shot scenario where, for the same fixed subset of subtasks/files $\{f_\ell(D_\ell)\}$ computed locally at each server, now the server can communicate different linear combinations to different sets of users. This ability offers a certain degree of refinement that the single-shot scenario may lack. This is exploited, and Theorem~\ref{Achievability-multi} reveals a range of parameters for which the multi-shot approach provides computation savings over the single-shot scenario. Interestingly, these computational savings are shown to be unbounded.
 \end{itemize}

\subsection{Paper Organization}
  The rest of the paper is organized as follows. Section~\ref{System-Model} introduces the model for multi-user distributed computing of linearly separable functions. Section~\ref{Formulating} formulates our problem, focusing on the single-shot scenario, for which Section~\ref{Single-Shot} presents the main results. This latter section first makes the connection to coding theory, and then presents the converse and achievability for the computation cost, as well as the achievable bound on the communication cost.  Section~\ref{Discussion} offers some insights including a discussion on the gains due to coding. Subsequently, in Section~\ref{Multi-Shot}, we present our proposed achievable multi-shot scheme and the corresponding results, and finally we conclude in Section~\ref{conclusion}. The appendices are in the subsequent sections.

\noindent{\bf Notations:}
We define $[n] \triangleq \{1,2,\hdots , n\}$.
For matrices $\mathbf{A}$ and $\mathbf{B}$, $[\mathbf{A},\mathbf{B}]$ indicates the horizontal concatenation of the two matrices.
For any matrix $\mathbf{X} \in \F^{m \times n}$, then $\mathbf{X}(i,j),\: i \in [m],\: j \in [n]$, represents the entry in the $i$th row and $j$th column, while $\mathbf{X}(i,:),\: i \in [m]$, represents the $i$th row, and $\mathbf{X}(:,j),\: j \in [n]$ represents the $j$th  column of $\mathbf{X}$. For two index sets $\I\subset [m], \J\in[n]$, then $\mathbf{X}(\I,\J)$ represents the sub-matrix comprised of the rows in $\I$ and columns in $\J$. We will use $\omega(\mathbf{X})$ to represent the number of nonzero elements of some matrix (or vector) $\mathbf{X}$. We denote the finite field $\mathbf{GF}{(q)}$ as $\mathbb{F}$. For any code $\mathcal{C} \subseteq \F^{n}$ and any vector $\mathbf{x} \in \F^{n}$, we use $d(\mathbf{x},\mathcal{C})$ to represent the Hamming distance of $\mathbf{x}$ to the nearest codeword in $\mathcal{C}$. We will dedicate the use of the letter $\rho$ when referring to normalized covering radii, and we will often use $\rho(\mathcal{C})$ to indicate that this is the normalized covering radius of a specific code $\mathcal{C} \subset \F^{n}$. We will often use the notation $\mathcal{C}_{\mathbf{H}}$ to refer to a code whose parity check matrix is $\mathbf{H}$, and similarly, we will use $\mathbf{H}_{\mathcal{C}}$ to refer to a matrix that serves as the parity-check matrix of a specific linear code $\mathcal{C}$. For some $k\leq n,\: k,n \in \mathbb{N}$, we will also often use the notation $\mathcal{C}(k,n)$ to emphasise that a linear code has message length $k$ and codeword length $n$.
For any two codes $\mathcal{C}_1$  and $\mathcal{C}_2$, we will use $[\mathcal{C}_1, \mathcal{C}_2]$ to represent the code resulting from their direct product. For some vector $\mathbf{x} \in \F^{n}$, we will use $\mathcal{C}_2 = <\mathbf{x},\mathcal{C}_1>$ to represents a code whose basis is the union of $\mathbf{x}$ with the basis of $\mathcal{C}_1$. Furthermore $V_q(n,\rho)$ will represent the volume of a Hamming ball in $\mathbb{F}^{n}$ of radius $\rho n$. For $0\leq x\leq 1 -\frac{1}{q}, x \in \mathbb{R}$, we will use $H_q(x) \triangleq  x \log_q(q-1) - x \log_q(x) - (1-x) \log_q(1-x)$ to represent the $q$-ary entropy function, while when $q=2$ we will use the simplified notation $H(x)$. We will use $\sup(\mathbf{x}^{\intercal})$ to represent the support of some vector $\mathbf{x}^{\intercal} \in \F^{n}$, describing the set of indices of non-zero elements. We will also use the notation $\epsilon(n)$ to represent an expression which, in the limit of large $n$, vanishes to zero.

\section{System Model}\label{System-Model}
We consider the multi-user linearly-separable distributed computation setting (cf.~Fig.~\ref{Fig: System Model}), which consists of $K$ users/clients, $N$ active (non-idle) servers, and a master node that coordinates servers and users.  A main characteristics of this setting is that the tasks performed at the servers, substantially outweigh in computation cost the linear operations performed at the different users. Another defining characteristic is that the cost of having the servers communicate to the users is indeed non-trivial.
We consider the setting where each server can use ${T}$ consecutive time slots to communicate different messages to different subsets of users, where in particular, during time-slot (shot) $t\in[T]$, server $n$ communicates to some arbitrary user-set $\T_{n,t} \subset [K]$, via a dedicated broadcast channel.

In our setting, each user asks for a (generally non-linear) function from a space of linearly separable functions, where each such function takes several datasets as input. Each desired function can be decomposed into a different linear combination of individual (again generally non-linear, and computationally hard) sub-functions $f_{\ell}(D_\ell)$ that each take a single dataset $D_\ell$ as input.
Consequently the demanded function $F_k(D_1,\dots,D_L)$ of each user $k \in[K]$, is a function of $L$ independent datasets $D_\ell,\: \ell \in[L]$, and it takes the general linearly-separable form
\begin{align}
    F_k(D_1,D_2,\hdots,D_L)&\triangleq f_{k,1}f_{1}(D_1) + f_{k,2}f_{2}(D_2) +\hdots +f_{k,L}f_{L}(D_\ell),\:\: k \in [K]\\
    &=f_{k,1}W_1 + f_{k,2}W_2 +\hdots +f_{k,L}W_L,\:\: k \in [K]\label{DefinitionOfLSFunctions}
\end{align}
where, as previously discussed, $W_\ell = f_{\ell}(D_\ell) \in \F ,\: \ell \in [L]$ is a so-called `file' output, and  $f_{k,\ell} \in \F ,\: k \in [K], \ell \in [L]$ are the linear combination coefficients.  As also mentioned before, $F_k$ itself can be a linear combination of some linearly separable functions. 

\begin{subsection}{Phases}
The model involves three phases, with the first being the \emph{demand phase}, then the \emph{assignment and computation phase} and then the \emph{transmission and decoding phase}. In the {demand phase},  each user $k \in [K]$ sends the information of its desired function $F_k(.)$ to the master node, who then deduces the linearly-separable decomposition of this function according to~\eqref{DefinitionOfLSFunctions}.
Then based on these $K$ desired functions, during the assignment and computation phase, the master assigns some of the subfunctions to each server, who then proceeds to compute these and produce the corresponding files $W_\ell = f_{\ell}(D_\ell)$. In particular, each subfunction $f_{\ell}(D_\ell)$ will be assigned to the servers belonging to some carefully chosen server-set $\W_\ell\subset [N]$.

During the transmission phase, each server $n \in [N]$ broadcasts during time slots $t=1,2,\dots,T$, different linear combinations of the locally computed output files, to different subsets of users $\T_{n,t}$. In particular, during time slot $t$, each server $n$ transmits
  \begin{align}
  z_{n,t}\triangleq \sum_{\ell \in [L]} e_{n,\ell,t} W_\ell,\:\: n\in [N] , t \in  [T] \label{EncodedFiles}
  \end{align}
  where the so-called encoding coefficients $e_{n,\ell,t}\in \F$ are determined by the master.
 Finally during the decoding part, each user $k$ linearly combines the received signals as follows
\begin{align}
    F'_{k} \triangleq \sum_{n \in [N], t \in [T]} d_{k,n,t} z_{n,t} \label{DecedFiles}
\end{align}
 for some decoding coefficients $d_{k,n,t} \in \F, n\in [N], t\in [T]$, determined again by the master node. Naturally $d_{k,n,t} =0,\forall k \notin \mathcal{T}_{n,t}$.  Decoding is successful when $F'_k = F_k$ for all $k\in[K]$.

  \begin{figure}
      \centering
      \includegraphics[scale=0.7]{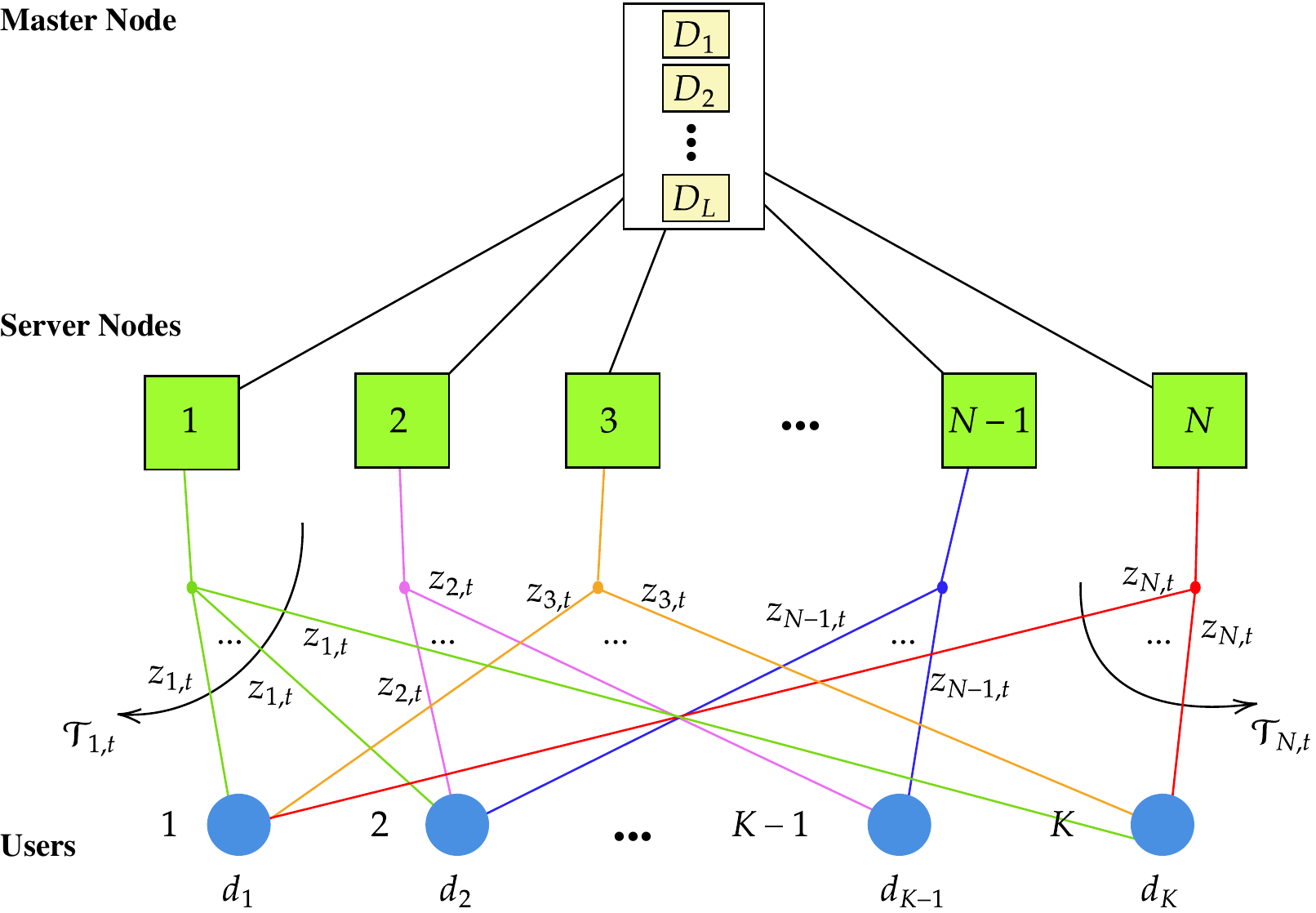}
      \caption{The $K$-user, $N$-server, linearly separable computation setting. After each user informs the master of its desired function $F_k(.)$, each component subfunction $W_\ell = f_{\ell}(D_\ell)$ is computed at each server in $\W_\ell\subset [N]$. During slot $t$, each server $n$ broadcasts a linear combination ${z}_{n,t}$ (of the locally available computed files) to all users in $\T_{n,t}$. This combination is defined by the coefficients $e_{n,\ell,t}$. Finally, to decode, each user $k \in [K]$ linearly combines (based on decoding vectors $\mathbf{d}_k$) all the received signals from all the slots and servers it has received from. Decoding must produce for each user its desired function $F_k(D_1,\dots,D_L)$.}
      \label{Fig: System Model}
  \end{figure}

  \end{subsection}

 \begin{subsection}{Computation and Communication Costs}

Remembering that $|\W_{\ell}|$ indicates the number of servers that compute a subfunction $W_{\ell} = f_\ell(D_{\ell}),\: \ell \in [L]$, our \emph{normalized computation cost} metric takes the form
  \begin{align}
      \gamma \triangleq \frac{\underset{\ell \in [L]}{\max} |\W_{\ell}|}{N}\label{ComputationCost}
  \end{align}
and represents the maximum fraction of all servers that must compute any subfunction.

We also formally define the \emph{normalized communication cost} as
  \begin{align}
     \delta  \triangleq \frac{\sum^{T}_{t=1}\sum^{N}_{n=1}|\mathcal{T}_{n,t}|}{KN}\label{comunication-cost2}
  \end{align}
to represent the average fraction of servers that each user gets data from\footnote{We here clarify that our setting implies that any link can be exploited, and our metric simply captures how many of these links are engaged when communicating. Reducing the communication cost implies activating fewer of these links, leaving the rest to be used for other responsibilities of the computing network.},\footnote{The observant reader may notice the computational cost being a worst-case cost, unlike the communication cost which refers to the average case. This choice is essential in making the connection to coding theory. This same choice though has an advantage; it allows us to better capture the effect of having some subfunctions that are much harder to compute than others.}. Hence in our setting,
  \begin{align}
     \Delta \triangleq \delta N\label{comunication-cost}
  \end{align}
represents the average number of transmitted `symbols' received by each user.
   \end{subsection}
   We wish to provide schemes that correctly compute the desired functions, at reduced computation and communication costs.

\section{Problem Formulation: One-Shot Setting}\label{Formulating}
In this single-shot setting of $T=1$, we will remove the use of the index $t$. Thus the transmitted value from~\eqref{EncodedFiles} will take the form \begin{align}
z_{n} = \sum_{\ell \in [L]} e_{n,\ell} W_\ell,\:\: n\in [N] \label{transmissionSingleSHOT}
\end{align}  where $e_{n,\ell}\in \F$ will denote the corresponding encoding coefficients, and where each such transmitted value at server $n$ will now be destined for the users in set $\T_n$. Similarly, the decoding value at each user $k$ (cf.~\eqref{DecedFiles}) will take the form $F'_{k} \triangleq \sum_{n \in [N]} d_{k,n} z_{n}$, where now $d_{k,n}, n\in [N]$, are the decoding coefficients. The desired functions $F_k(.)$ (cf.~\eqref{DefinitionOfLSFunctions}), their linear decomposition coefficients $f_{k,\ell}$ (cf.~\eqref{DefinitionOfLSFunctions}), and the decoded functions $F'_k(.)$ in~\eqref{DecedFiles}, remain as previously described.
With the above in place, we will use
\begin{align}
      \mathbf{f}&\triangleq [F_1,F_2,\hdots,F_K]^{\intercal} \\
      \mathbf{f}_k &\triangleq [f_{k,1},f_{k,2},\hdots,f_{k,L}]^{\intercal},\: k \in [K]\label{function-vectors-1}\\
    \mathbf{w}&\triangleq [W_{1},W_{2},\hdots,W_{L}]^{\intercal}\label{message-vectors-1}
\end{align}
where $\mathbf{f}$ represents the vector of the output demanded functions (cf.~\eqref{DefinitionOfLSFunctions}), $\mathbf{f}_k$ the vector of function coefficients for user $k$ (cf.~\eqref{DefinitionOfLSFunctions}), and $\mathbf{w}$ the vector of output files.
We also have
\begin{align}
\mathbf{e}_{n} &\triangleq [e_{n,1},e_{n,2},\hdots, e_{n,L}]^\intercal,\: n \in [N] \label{encoding-vectors-per-shot}\\
     \mathbf{z}& \triangleq [z_{1}, z_{2},\hdots, z_{N}]^ \intercal\label{linear-combinations}
\end{align}
respectively representing the encoding vector at server $n$, and the overall transmitted vector across all the servers (cf.~\eqref{transmissionSingleSHOT}).
Furthermore, we have
\begin{align}
     \mathbf{d}_{k} &\triangleq [d_{k,1},d_{k,2},\hdots, d_{k,N}]^ \intercal,\: k \in [K] \label{decoding-vectors-per-shot-1}\\
      \mathbf{f}'&\triangleq [F'_1,F'_2,\hdots,F'_K]^{\intercal}
\end{align}
respectively representing the decoding vector at user $k$, and the vector of the decoded functions across all the users.
In addition, we have
\begin{align}
      \mathbf{F}& \triangleq [\mathbf{f}_1,\mathbf{f}_2,\hdots,\mathbf{f}_K]^{\intercal} \in \F^{K \times L} \label{Demand-Matrix-1}\\
        \mathbf{E} &\triangleq [\mathbf{e}_{1},\mathbf{e}_{2},\hdots, \mathbf{e}_{N}]^\intercal \in \F^{N \times L}\label{encoding-vectors-per-shot-1}\\
    \mathbf{D} &\triangleq [\mathbf{d}_1,\mathbf{d}_2, \hdots , \mathbf{d}_K]^{\intercal} \in \F^{K \times N} \label{DecodingMatrix-1}
\end{align}
where $\mathbf{F}$ represents the $K\times L$ matrix of all function coefficients across all the users, where $\mathbf{E}$ represents the $N\times L$ \emph{encoding matrix} across all the servers, and where $\mathbf{D}$ represents the $K\times N$ \emph{decoding matrix} across all the users.

Directly from~\eqref{DefinitionOfLSFunctions}, we have that
\begin{align}
    \mathbf{f} =[\mathbf{f}_1,\mathbf{f}_2,\hdots,\mathbf{f}_K]^{\intercal} \mathbf{w}\label{Functions-one}
\end{align}
and from \eqref{EncodedFiles} we have the overall transmitted vector taking the form
\begin{align}
    \mathbf{z} =[\mathbf{e}_{1}, \mathbf{e}_{2}, \hdots,\mathbf{e}_{N}]^\intercal \mathbf{w} = \mathbf{E} \mathbf{w} . \label{EncodedCashedData-1}
\end{align}
Furthermore, directly from~\eqref{DecedFiles} we have that
\begin{align}
    F'_k= \mathbf{d}_{k}^{T} \mathbf{z}
\end{align}
and thus we have
\begin{align}
    \mathbf{f}' = [\mathbf{d}_1,\mathbf{d}_2,\hdots,\mathbf{d}_K]^{\intercal} \mathbf{z} = \mathbf{D}\mathbf{z}.\label{DecodedData-one}
\end{align}
Recall that we must guarantee that
\begin{align}
    \mathbf{f}'=\mathbf{f}\label{feasibility-one}.
\end{align}
After substituting \eqref{Functions-one}, \eqref{EncodedCashedData-1} and \eqref{DecodedData-one} into \eqref{feasibility-one}, we see that the above feasibility condition in~\eqref{feasibility-one} is satisfied iff
\begin{align}
    \mathbf{D}\mathbf{E}\mathbf{w} = \mathbf{F}\mathbf{w}.\label{MainEquationWithW}
\end{align}
For this to hold for any $\mathbf{w}$, we must thus guarantee
\begin{align}
    \mathbf{D}\mathbf{E} = \mathbf{F}.\label{MainEquation}
\end{align}
At this point, since $\W_{\ell} = \sup(\mathbf{E}(:, \{\ell\})^{\intercal})$, and since $|\W_{\ell}| = \omega(\mathbf{E}(:, \{\ell\}))$, we have that
\begin{align}
    \underset{\ell \in [L]}{\max}\: \omega(\mathbf{E}(:,\ell)) = \max_{\ell \in [L]} |\W_{\ell}| \label{bound-on-the-commp-cost-one}
\end{align}
which simply tells us that our computation cost $\gamma$ from \eqref{ComputationCost} takes the form
\begin{align}
    \gamma = \frac{1}{N} \underset{\ell \in [L]}{\max}\: \omega(\mathbf{E}(:,\ell)).\label{CommpCost-E}
\end{align}
Similarly, directly from~\eqref{DecedFiles} and \eqref{comunication-cost}, we see that
  \begin{align}
     \delta  = \frac{\omega(\mathbf{D})}{KN}\label{Com-Cost-relation-2}
\end{align}
which simply says (cf.~\eqref{comunication-cost}) that
\begin{align}
     \Delta  = \frac{\omega(\mathbf{D})}{K}\label{Com-Cost-relation-1}.
\end{align}
It is now clear that decomposing $\mathbf{F}$ into the product of two relatively sparse matrices $\mathbf{D}$ and $\mathbf{E}$, implies reduced communication and computation costs respectively.

We here provide a simple example to help clarify the setting and the notation.

\subsection{Simple Example}
\label{Motivating-Example}
As illustrated in Figure~\ref{example-1}, we consider the example of a system with a master node, $N=8$ servers, $K=4$ users, $L=6$ subfunctions/datasets, and a field of size $q=7$.
  \begin{figure}
      \centering
      \includegraphics[scale=1]{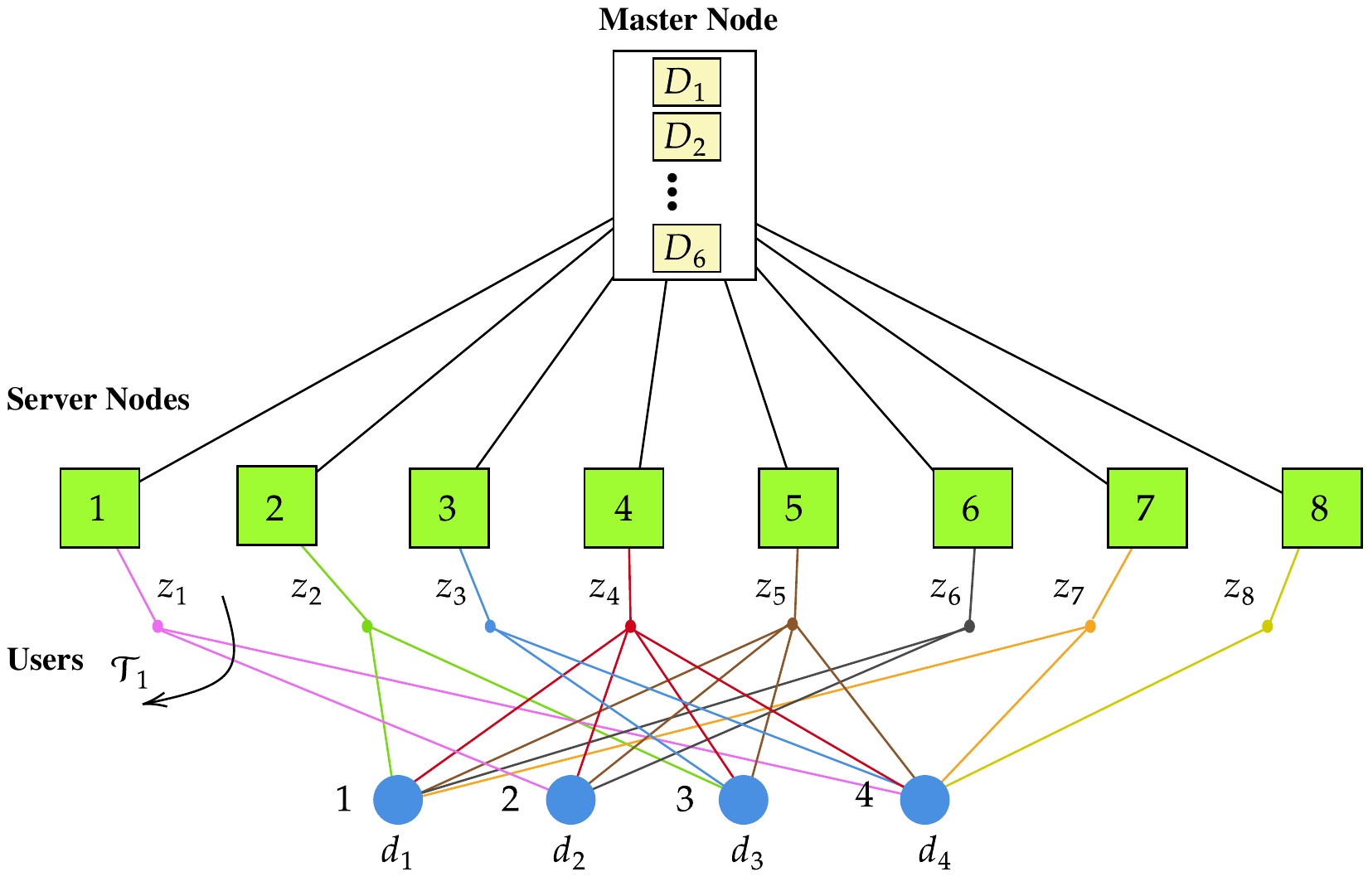}
      \caption{Multi-user distributed computing setting with 8 servers, 4 users, and 6 datasets/subfunctions.}
      \label{example-1}
  \end{figure}

Let us assume that the users ask for the following functions:
\begin{align}
    F_1 &= 2f_1(D_1) + 4 f_2(D_2) + 4 f_3(D_3) + 5 f_4(D_4) + 5 f_5(D_5) = \mathbf{f}_1^{\intercal} \mathbf{w},\\
        F_2 &= 3f_1(D_1) + 4 f_2(D_2) + 5 f_3(D_3) + 2 f_4(D_4) + 6 f_5(D_5) + 6 f_6(D_6) = \mathbf{f}_2^{\intercal} \mathbf{w},\\
        F_3 &= 2 f_1(D_1) + 4 f_2(D_2) +  6f_3(D_3) + 5 f_4(D_4) + 2f_5(D_5) = \mathbf{f}_3^{\intercal} \mathbf{w},\\
        F_4 &= 3 f_1(D_1) + 5 f_2(D_2) + 2f_4(D_4) + 3 f_5(D_5) + f_6(D_6) = \mathbf{f}_4^{\intercal} \mathbf{w}
\end{align}
where $F_k,\mathbf{f}_k,\: \: k \in [4]$, and $\mathbf{w}$, are respectively defined in~\eqref{DefinitionOfLSFunctions}, \eqref{message-vectors-1} and \eqref{function-vectors-1}. Consequently from~\eqref{Demand-Matrix-1}, our function matrix takes the form
\begin{align}
    \mathbf{F} =
\left[ \begin{array}{cccccc}
 2& 4 & 4 & 5 & 5 & 0 \\[-5pt]
3 & 4 & 5 & 2 & 6 & 6\\[-5pt]
  2&  4& 6 &  5 & 2 & 0\\[-5pt]
 3&5 & 0  & 2 & 3 & 1
\end{array}\right].
\end{align}

In the assignment phase, the master allocates the computation of $f_1(D_1),f_2(D_2),\hdots,f_6(D_6)$ to the $8$ servers according to
\begin{align}
    \mathcal{W}_1 &=\{ 1,2,3,5,8\},\: \mathcal{W}_2 = \{1,2,3,4,6,7\},\: \mathcal{W}_3 = \{1,2,3\},\: \mathcal{W}_4 = \{1,4,5,7\}\\
    \mathcal{W}_5 &= \{1,2,4,5,6,8\}, \: \mathcal{W}_6 =\{3,4,5,6,7,8\}
\end{align}
so that for example subfunction $f_3(D_3)$ is assigned to servers $\{1,2,3\}$, while we can also see that for example server 2 has to compute $W_1 = f(D_1),W_2 = f(D_2),W_3 = f(D_3)$, and $W_5 = f(D_5)$.
A quick inspection shows that the normalized  computation cost (cf.~\eqref{ComputationCost}) is equal to
\begin{align}
    \gamma = \frac{\underset{\ell \in [6]}{\max} |\W_{\ell}|}{8}= 6/8  \label{Computed-CC}.
\end{align}
After computing their designated output files, each server $n$ transmits $z_n$ as follows
\begin{align}
    z_{1} &= 2 W_1 + 6 W_2 + 3 W_3 + W_4 + 2 W_5, \ \ \   z_{2} = 4 W_1 + 5 W_2 + 2 W_3 + 3 W_5,\\ z_{3} &= W_1 + 2 W_2 + W_3 + 2 W_6, \ \ \     z_{4} = W_2 + 2 W_4 + 4 W_5 + W_6, \\ z_{5} &= 2W_1 + W_4 + 3 W_5 +2 W_6, \ \ \    z_{6} = 2 W_2 + 5 W_5 + 3 W_6 \\
       z_{7} &= W_2 + 2 W_4 + 4 W_6,    \ \ \   z_{8} = 2 W_1 + 4 W_5 + 5 W_6
\end{align}
corresponding to an encoding matrix (cf.~\eqref{EncodedCashedData-1}) of the form
\begin{align}
    \mathbf{E} =
    \left[ \begin{array}{cccccc}
  2& 6 & 3 &1 & 2& 0\\[-7pt]
 4& 5 & 2 &0 & 3& 0\\[-7pt]
 1& 2 & 1 & 0 & 0 & 2\\[-7pt]
 0& 1 & 0 & 2 & 4 & 1\\[-7pt]
 2& 0 & 0 & 1 & 3 & 2\\[-7pt]
0& 2 & 0 & 0 & 5 & 3 \\[-7pt]
0 & 1 & 0 & 2 & 0 & 4\\[-7pt]
2 & 0 & 0 &0 & 4 & 5
\end{array}\right].
\end{align}
We can quickly verify (cf.~\eqref{Computed-CC}) that indeed $\underset{\ell \in [6]}{\max}\: \omega(\mathbf{E}(:,\ell))/8= 6/8=\gamma$.

Subsequently, the master asks each server $n$ to send its generated $z_{n}$ to its designated receiving users, where for each server, these user-sets are:
\begin{align}
    \T_{1} = \{2,4\}, \: \T_{2} = \{1,3\},\: \T_{3} = \{3\}, \: \T_{4} =\{1,2,3,4\},\:\\ \T_{5} =\{1,2,3,4\},\: \T_{6} = \{1,2\}, \T_{7} = \{1,4\}, \T_{8} = \{ 4\}
\end{align}
so now, for example, server 2 will broadcast $z_{2}$ to users $1$ and $3$. A quick inspection also shows that users $1$ and $4$ will receive $5$ different symbols, whereas users $2$ and $3$ will receive $4$ symbols each. The above corresponds to a normalized communication cost (cf.~\eqref{comunication-cost}) equal to
\begin{align}
    \delta = \frac{\sum^{8}_{n=1} |\T_n|}{4\cdot 8} = (5+4+4+6)/32 = 19/32
\end{align}
corresponding to an average of $\Delta = \frac{19}{4}$ symbols received per user.

To decode, each user $k \in [4]$ computes the linear combination $F_k'$ as
\begin{align}
\begin{array}{cc}
    F_1' = 2 z_{2} + 3  z_{4} + 4 z_{5} + 2 z_{6} +z_{7}, &  F_2' = 4 z_{1} + 2  z_{4} +  z_{5} + 3 z_{6},\\
    F_3' = 4 z_{2} + 5 z_{3} + 2 z_{4} + z_{5}, & F_4' = 4 z_{1} + 2 z_{3} + z_{4} + 2 z_{5} +4 z_{7} + 5 z_{8}
\end{array}
\end{align}
adhering to a decoding matrix of the form
\begin{align}
    \mathbf{D}  =
    \left[ \begin{array}{cccccccccc}
  0& 2 & 0 &3 & 4& 2 & 1 & 0\\
4 & 0 & 0 & 2 & 1 & 3 & 0 & 0 \\
0 & 4 & 5 & 2 & 1& 0 & 0 &0 \\
4 & 0 & 2 & 1 & 2 & 0 & 4 &5
\end{array}\right].
\end{align}
A quick verification\footnote{Let us recall that each decoded symbol takes the form $F_k' = \mathbf{d}^{\intercal}_k \mathbf{z}$ where $\mathbf{d}^{\intercal}_k$ is the $k$th row of $\mathbf{D}$, and where $\mathbf{z} = [z_1 \ z_2 \ \cdots \ z_N]^T$.} reveals the correctness of decoding, and that indeed $F_k' = F_k$ for all $k=1,2,3,4$. For example, for the first user, we see that
$F_1' = 2 z_2 + 3  z_4 + 4 z_5 + 2 z_6 + z_7= 2 (4 W_1 + 5 W_2 + 2 W_3 + 3 W_5 ) + 3 (W_2 + 2 W_4 + 4 W_5 + W_6) + 4 ( 2W_1 + W_4 + 3 W_5 +2 W_6) + 2 (2 W_2 + 5 W_5 + 3 W_6) +  (W_2 + 2 W_4 + 4 W_6)=2 W_1 + 4 W_2 + 4 W_3 + 5 W_4 + 5 W_5 + 0 W_6$ which indeed matches $F_1$. In this example, each user recovers their desired function, with a corresponding normalized computation cost $\gamma = 3/4$ and a normalized communication cost $\delta = 19/32$. This has just been an example to illustrate the setting. The effort to find a solution with reduced computation and communication costs, follows in the subsequent section.

\section{Computation and Communication Costs for the Single-Shot Setting} \label{Single-Shot}
In this section we present the results for the one-shot setting. We first rigorously establish the bridge between our problem, coding theory, covering and partial covering codes. The main results --- focusing first on the computational aspects --- are presented in Section~\ref{computationally bounded} which derives bounds on the optimal computation cost in the large $N$ setting. With these results in place, the subsequent Section~\ref{Comuunication-Computation-Optimal} extends our consideration to the communication cost as well. Finally, Section~\ref{Discussion} offers some intuition on the results of this current section.

We briefly recall (cf.~\cite{cohen1997covering}) that an $n$-length code $\mathcal{C} \subset \F^{n} $ is called a $\rho$-covering code if it satisfies
\begin{align}
    d(\mathbf{x},\mathcal{C}) \leq \rho n,\:\: \forall \mathbf{x} \in \F^{n} \label{Covering-Codes-Definition}
\end{align}
for some $\rho\in(0,1)$ which is referred to as the normalized covering radius.

\subsection{Establishing a Relationship to Covering Codes and Partial Covering Codes}
\label{Relationship-Coding}
We will first seek to decompose $\mathbf{F}$ into $\mathbf{F} = \mathbf{D}\mathbf{E}$ under a constrained computation cost $\gamma$ which will generally imply a sparsity constraint on $\mathbf{E}$. For $\mathbf{E}_\ell \triangleq  \mathbf{E}(:,\ell)$ and $\mathbf{F}_\ell \triangleq  \mathbf{F}(:,\ell)$ denoting the $\ell$th column of $\mathbf{E}$ and $\mathbf{F}$ respectively, we can rewrite our decomposition as
\begin{align}
    \mathbf{D}\mathbf{E}_\ell=\mathbf{F}_\ell, \  \ \forall \ell \in [L].\label{Eq:non-Agnostic}
\end{align}
As suggested before, if we viewed $\mathbf{D}\in \F^{K\times N}$ as a parity check matrix $\mathbf{H}_{\mathcal{C}} = \mathbf{D}$ of a code $\mathcal{C}\subset \F^N$, then we could view $\mathbf{E}_\ell\in \F^N$ as an arbitrary error pattern, and $\mathbf{F}_\ell\in \F^K$ as the corresponding syndrome. Since we wish to sparsify $\mathbf{E}_\ell$, we are interested in having $\mathbf{E}_\ell$ be the minimum-weight coset leader whose syndrome is $\mathbf{F}_\ell$. This is simply the output of the minimum-distance syndrome decoder\footnote{Naturally our viewing $\mathbf{D}$ as a parity check matrix, does not limit the scope of options in choosing $\mathbf{D}$. Similarly, associating $\mathbf{E}_\ell$ the role of an error pattern, or a minimum-weight coset leader, is again not a limiting association.}. To get a first handle on the weights of $\mathbf{E}_\ell$, we can refer to the theory of covering codes which bounds the weights of coset leaders, where these weights are bounded by the code's covering radius $\rho(\mathcal{C}) N$, for some normalized radius $\rho(\mathcal{C})\in(0,1)$. Since the covering radius $\rho  N$ upper bounds the weights of the coset leaders\footnote{Let us recall (cf.~\cite{roth2006introduction}) that the preferred coset leaders are the minimum-weight vectors in each row of the standard array.}, it upper bounds our computation cost. A covering radius $\gamma N$ would reflect our computation constraint $\gamma$.

To capture some of the coding-theoretic properties, we will transition to the traditional coding-theoretic notation which speaks of an $n$-length code $\mathcal{C}\subset \F^{n} $ of rate $k/n$, where for us $n=N$ and $k = N-K$. The parity check matrix $\mathbf{H}_{\mathcal{C}}\in \F^{(n-k)\times n}$ will generally be associated to our decoding matrix $\mathbf{D} \in \F^{K\times N}$, the received (or error) vectors $\mathbf{x} \in \F^{n}$ will be associated to the encoding vectors $\mathbf{E}_\ell\in \F^{N}$, and its syndrome $\mathbf{s}_{\mathbf{x}} \in \F^{n-k}$ (or just $\mathbf{s}$, depending on the occasion) will be associated to $\mathbf{F}_\ell\in \F^{K}$. Please recall that when we write $\mathcal{C}_{\mathbf{D}}$ (or $\mathcal{C}_{\mathbf{H}}$), we will refer to the code whose parity check matrix is $\mathbf{D}$ (or $\mathbf{H}$).

As a first step, we extend the concept of covering codes to the following class.
\begin{Definition}\label{Partial-Covering-Code}
For some $ \rho\in (0, 1]$, we say that a set $\mathcal{X} \subseteq \F^{n}$ is $\rho$-covered by a code $\mathcal{C}\subseteq \F^{n}$ iff
\begin{align}
    d(\mathbf{x},\mathcal{C}) \leq \rho n,\:\: \forall \mathbf{x} \in \mathcal{X}
\end{align}
in which case we say that $\mathcal{C}$ is a $(\rho,\mathcal{X})$-partial covering code.
\end{Definition}
Naturally when $\mathcal{X} = \F^{n}$, such a $(\rho,\mathcal{X})$-partial covering code is simply the traditional covering code.
We are now able to link partial covering codes to our distributed computing problem.

\begin{Theorem}\label{Bridge}
For the setting of distributed-computing with $K$ users, $N$ servers and $L$ subfunctions, a solution to the linearly separable function computation problem $\mathbf{D}\mathbf{E} = \mathbf{F}$ with normalized computation cost $\gamma$
exists if and only if $\mathbf{D}$ is the parity check matrix to a $(\gamma ,\mathcal{X})$-partial covering code $\mathcal{C}_{\mathbf{D}}$ for some existing set
\begin{align}
    \mathcal{X}\supset \mathcal{X}_{\mathbf{F},\mathbf{D}}\triangleq  \{\mathbf{x} \in \F^{N}| \mathbf{D} \mathbf{x}  =\mathbf{F}(:,\ell), \ \text{for some} \ \ell \in [L]\}. \label{x_F_definiton}
\end{align}
 With such $\mathbf{D}$ in place, each $\mathbf{E}(:,\ell)$ is the output of the minimum-distance syndrome decoder of $\mathcal{C}_{\mathbf{D}}$ for syndrome $\mathbf{F}(:,\ell)$.
\end{Theorem}

\begin{proof}
To first prove that the computation constraint $\gamma=\rho$ indeed requires $\mathbf{D}$ to correspond to a partial covering code that covers $\mathcal{X}$, let us assume that $\mathbf{D}$ does not have this property, and that there exists an $\mathbf{x} \in \mathcal{X}_{}$ such that $d(\mathbf{x},\mathcal{C}_{\mathbf{D}}) > \rho n$. Let $\mathbf{c}_{\min}$ be the closest codeword to $\mathbf{x}$ in the sense that $d(\mathbf{x},\mathbf{c}_{\min}) = d(\mathbf{x},\mathcal{C}_{\mathbf{D}})$. Now let $\mathbf{e}_{\min} = \mathbf{x} - \mathbf{c}_{\min}$, and note, directly from the above assumption, that $ \omega(\mathbf{e}_{\min}) > \rho n$. Naturally $\mathbf{D} \mathbf{x} = \mathbf{D} (\mathbf{e}_{\min} + \mathbf{c}_{min}) = \mathbf{D} \mathbf{e}_{\min}$ by  virtue of the fact that $\mathbf{D}$ is the parity check matrix of $\mathcal{C}_{\mathbf{D}}$. Since $\mathbf{x} \in \mathcal{X}$, we know that $\exists \: \ell \in [L]$ such that $\mathbf{D} \mathbf{x} = \mathbf{F}(:,\ell)$, which directly means that $\exists \: \ell \in [L]$ such that $\mathbf{D} \mathbf{e}_{\min} = \mathbf{F}(:,\ell).$ This $\mathbf{e}_{\min}$ is the coset leader associated to syndrome $\mathbf{F}(:,\ell)$.

Since though $\mathbf{D}\mathbf{E} = \mathbf{F}$, we also have that $\mathbf{D} \mathbf{E}(:,\ell ) = \mathbf{F}(:,\ell)$. Since $\mathbf{E}(:,\ell )$ and $\mathbf{e}_{\min}$ are in the same coset (of the same syndrome $\mathbf{F}(:,\ell)$), and since $\mathbf{e}_{\min}$ is the minimum-weight coset leader, we can conclude that $\omega(\mathbf{E}(:,\ell)) \geq  \mathbf{e}_{\min}$. Thus the assumption that $ \omega(\mathbf{e}_{\min}) > \rho n$ implies that $\omega(\mathbf{E}(:,\ell ))  > \rho n$ which contradicts the computation-cost requirement that $\omega(\mathbf{E}(:,\ell)) \leq \rho n$ from~\eqref{CommpCost-E}. Thus if $\mathbf{D}$ does not correspond to a partial covering code (with $\rho = \gamma$) that covers $\mathcal{X}_{\mathbf{F}, \mathbf{D}}$, the complexity constraint is violated.

On the other hand, recalling that $\mathcal{C}_\mathbf{D}$ is a partial covering code for $\mathcal{X}$, we get that for any $\mathbf{x} \in \mathcal{X}$ then $d(\mathbf{x},\mathcal{C}_{\mathbf{D}}) \leq \rho n$. For the same $\mathbf{x} \in \mathcal{X}$, let $\mathbf{c}_{\min}$ be again its closest codeword, and let $\mathbf{e}_{\min} = \mathbf{x} - \mathbf{c}_{\min}$, where again by definition of the partial covering code, $ \omega(\mathbf{e}_{\min}) \leq \rho n$. Since, like before, $\mathbf{D} \mathbf{e}_{\min} = \mathbf{F}(:,\ell)$ for some $\ell\in[L]$, then we simply set $\mathbf{E}(:,\ell) = \mathbf{e}_{\min}$ whose weight is indeed sufficiently low to guarantee the computation constraint. We recall that for each $\mathbf{F}(:,\ell)$, this coset leader $\mathbf{E}(:,\ell) = \mathbf{e}_{\min}$ can be found by using the minimum-distance syndrome decoder.
\end{proof}
Intuitively, a smaller $\mathcal{X}$ could potentially --- depending on $\mathcal{X}$ and the code --- be covered in the presence of a smaller covering radius.  Now that we have established the connection with partial covering codes, we proceed to present computation bounds. The following result, as well as all subsequent results, assumes large $N$.
\subsection{Bounds on the Optimal Computation Cost}
\label{computationally bounded}
The following theorem bounds the optimal computation cost of the multi-user linearly-separable computation setting.

\begin{Theorem}\label{Achievability}
For the setting of distributed-computing of linearly-separable functions, with $K$ users, $N$ servers and any number of $L$ subfunctions, the optimal computation cost is bounded as
  \begin{align}
      \gamma &\in (H_q^{-1}(\frac{\log_q(L)}{N}), H_q^{-1}(\frac{K}{N} )).\label{Achievable-interval}
  \end{align}

\end{Theorem}

\begin{proof}
The proof of the converse (lower bound in~\eqref{Achievable-interval}) employs sphere-covering arguments, and can be found in Appendix \ref{A_0}.
The proof of achievability follows from covering- and partial covering-code arguments, and can be found in Appendix~\ref{B}.
\end{proof}
\begin{remark}
The two bounds meet when $L = q^K$.
\end{remark}
Theorem \ref{Achievability} suggests a range of computation costs. In the next corollary, we will describe the conditions under which a reduced normalized computation cost, strictly inside this range, can be achieved. This reduced cost will relate to (our ability to choose) a set $\mathcal{X}\subset \F^N$. As we will see, a smaller $\mathcal{X}$ will imply a smaller $\gamma$. To understand the connection between our problem and this set $\mathcal{X}$, and thus to better understand the following theorem whose proof will be presented in Appendix~\ref{A_1}, we provide the following sketch of some crucial elements in the proof of Theorem~\ref{Achievability}. In particular, we will here sketch an algorithm that iterates in order to converge to the aforementioned $\mathcal{X}$, and then to the corresponding decoding matrix $\mathbf{D}$, that will eventually provide reduced normalized complexity $\gamma$.
Before describing the algorithm, it is worth noting that a crucial ingredient can be found in Lemma~\ref{CohenLemma} (see Appendix \ref{A_2}), which modifies the approach in \cite{cohen1985good} in order for us to design --- for any set $\mathcal{X}' \in \F^{N}$ --- a $(\rho,\mathcal{X}')$-partial covering code for some $\rho = H_q^{-1}(\frac{K}{N} - (1 - \frac{ \log_q(|\X'|)}{N})$.

With this in place, the algorithm starts by picking an initial set $\mathcal{X}_{0} \in \F^{N},|\mathcal{X}_{0}| = L q^{N-K}$, and then applies Lemma \ref{CohenLemma} to construct a $(\rho_0,\mathcal{X}_{0})$-partial covering code, $\mathcal{C}_{0}$, where $\rho_0 = H_q^{-1}(\frac{K}{N} - (1 - \frac{ \log_q(|\X_{0}|)}{N})$. With this code $\mathcal{C}_{0}$ in place, we create --- as a function of $\mathcal{C}_{0}$ --- the set $\mathcal{X}_{\mathbf{F},\mathbf{D},0}$ as defined in~\eqref{x_F_definiton} where $\mathbf{D}=\mathbf{H}_{{\C}_{0}}$, and then we check if $\mathcal{X}_{0} \supseteq \mathcal{X}_{\mathbf{F},\mathbf{D},0}$. If so, then the algorithm terminates, else it goes to the next iteration which starts by picking a new larger set $\mathcal{X}_{1} \in \F^{N},|\mathcal{X}_1| = L q^{N-K}+1$, then uses Lemma \ref{CohenLemma} to create a new $(\rho_1,\mathcal{X}_{1})$-partial covering code for
$\rho_1 = H_q^{-1}(\frac{K}{N} - (1 - \frac{ \log_q(|\X_{1}|)}{N})$, and then compares if $\mathcal{X}_{1} \supseteq \mathcal{X}_{\mathbf{F},\mathbf{D},1}$. This procedure terminates during some round $m$ where this terminating round is the first round for which the chosen set $\mathcal{X}_m$ (now of cardinality $|\mathcal{X}_m| = L q^{N-K}+m$) and the corresponding $(\rho_m,\mathcal{X}_{m})$-partial covering code with $\rho_m = H_q^{-1}(\frac{K}{N} - (1 - \frac{ \log_q(|\X_{m}|)}{N})$, yield $\mathcal{X}_{m} \supseteq \mathcal{X}_{\mathbf{F},\mathbf{D},m}$.

In the following corollary, the mentioned $\mathcal{X}$ refers to the terminating\footnote{Note that in the worst case this termination will happen when $\mathcal{X}_m  = \F^{N},$ in which case the output code will be a covering code.} $\mathcal{X}_m$, and the decoding matrix $\mathbf{D}$ will be the parity-check matrix of the aforementioned $(\rho_m,\mathcal{X}_{m})$-partial covering code that covers the terminating $\mathcal{X}=\mathcal{X}_{m}$, while the normalized computation cost in the theorem will take the form $\gamma = \rho = \rho_m$.

With the above in place, the following speaks of a set $\mathcal{X}$ that is $\rho N$-covered by a code $\mathcal{C}_{\mathbf{D}}$ that generates --- as described in~\eqref{x_F_definiton} --- its set $\mathcal{X}_{\mathbf{F},\mathbf{D}}$.

\begin{corollary}\label{Better-Achievability}
In the multi-user linearly separable computing problem $\mathbf{D}\mathbf{E} = \mathbf{F},$
if there exists a set
\begin{align*}
    \mathcal{X}\supset \mathcal{X}_{\mathbf{F},\mathbf{D}}\triangleq  \{\mathbf{x} \in \mathbb{F}^{N}| \mathbf{D} \mathbf{x}  =\mathbf{F}(:,\ell), \ \text{for some} \ \ell \in [L]\} \label{x_F_definiton}
\end{align*}
that is $\rho N$-covered by a code $\mathcal{C}_{\mathbf{D}}$ for $ \rho =  H_q^{-1}(\frac{K}{N} - (1 - \frac{ \log_q(|\mathcal{X}|)}{N}))$, then the computation cost \[ \gamma =  H_q^{-1}(\frac{K}{N} - (1 - \frac{ \log_q(|\mathcal{X}|)}{N}))\] is achievable. If $\mathcal{X} = \mathcal{X}_{\mathbf{F}, \mathbf{D}}$, then $\gamma = H_q^{-1}(\frac{\log_q(L)}{N})$ is achievable and optimal.
\end{corollary}

\begin{proof}
The proof can be found in Appendix \ref{A_1}.
\end{proof}
As suggested before, the above reflects that covering a smaller $\mathcal{X}$ could entail a smaller covering radius and thus a smaller computation cost.

\subsection{Jointly Considering Computation and Communication Costs} \label{Comuunication-Computation-Optimal}
The following theorem combines computation and communication considerations. Theorem~\ref{Sparse-Theorem} builds on Theorem~\ref{Bridge}, where now we recall that any chosen decoding matrix $\mathbf{D}$ will automatically yield a normalized communication cost $\delta = \frac{\omega{(\mathbf{D}})}{KN}$ corresponding to $\Delta = \delta N = \frac{\omega{(\mathbf{D}})}{K}$. The following bounds this communication cost.

\begin{Theorem}\label{Sparse-Theorem}
For the setting of distributed-computing of linearly-separable functions, with $K$ users, $N$ servers and $L$ subfunctions, the optimal computation cost is bounded as
  \begin{align}
      \gamma &\in (H_q^{-1}(\frac{\log_q(L)}{N}), H_q^{-1}(\frac{K}{N} ))\label{}
  \end{align}
and for any achievable computation cost $\gamma \leq \min\{\frac{\sqrt{5}-1}{2}, 1 -\frac{1}{q}\}$, then the corresponding achievable communication cost takes the form
\begin{align}
    \delta \doteq \frac{\sqrt{\log_q(N)}}{N}.
\end{align}

\end{Theorem}
\begin{proof}
The proof can be found in Appendix~\ref{E_1}.
\end{proof}

We here offer a quick sketch of the proof of the above theorem. The proof first employs a modified version of the famous result by Blinovskii in~\cite{blinovskii1987lower} which proved that, as $n$ goes to infinity, almost all random linear codes $\mathcal{C}(k,n)$ are covering codes, as long as the normalized covering radius satisfies $\rho  \geq H_q^{-1}(\frac{n-k}{n})$. This modification of Blinovskii's theorem is presented in Theorem \ref{partial-covering-all}, whose proof is found in Appendix~\ref{E_1}. With this modification in place, we prove that almost all $(k,n)$ random linear codes with
\begin{align}
    \rho=H_q^{-1}(\frac{\log_q(|\mathcal{X}|) - k}{n}) \label{radii-relation}
\end{align}
are $(\rho,\mathcal{X})$-partial covering codes, each for their own set $\mathcal{X} \in \F^{n}$. This is again in Theorem~\ref{partial-covering-all}. With this theorem in place, we then employ a concatenation argument (which can be found in the proof of Theorem~\ref{Sparse-Theorem} in Appendix~\ref{E_1}), to build a sparse parity-check matrix $\mathbf{H}$ of a partial covering code, which --- by virtue of the connection made in Theorem~\ref{Bridge} --- allows us to complete the proof of Theorem~\ref{Sparse-Theorem}.

To show that sparse parity check codes can indeed offer reduced computation costs, we had to show that sparse codes can indeed offer good partial covering properties. To do that, we followed some of the steps described below.
In particular, we designed an algorithm that begins with constructing a sparse parity check code that can cover, for a given radius $\rho_0$, a minimum necessary cardinality set $\mathcal{X}_0$, where this minimum cardinality of  $|\mathcal{X}_{0}| = L q^{N-K}$ is imposed on us by $\mathbf{F}$. The parity-check matrix of this first code is $\mathbf{H}_{0}$. Then following the steps in the proof of Theorem~\ref{Bridge}, we set $\mathbf{D}=\mathbf{H}_0$ and check if $\mathcal{X}_0 \supseteq \mathcal{X}_\mathbf{\mathbf{F},\mathbf{D}}$ holds. If it indeed holds, the algorithm outputs $\mathbf{D}$ and $\mathcal{X}_{0}$, and the corresponding cost is $\gamma = \rho_0$, where this $\rho$ value is derived from \eqref{radii-relation} by setting $\mathcal{X}= \mathcal{X}_0$. Otherwise the algorithm constructs another sparse partial covering code with a new parity check matrix $\mathbf{H}_{1}$, now covering a set $\mathcal{X}_{1}$ with cardinality $|\mathcal{X}_1| = L q^{N-K} +1$, and then checks again the same inclusion condition as above. The procedure continues until it terminates, with some covered set $\mathcal{X}_{m}$ of cardinality $|\mathcal{X}_{m}| = Lq^{N-K}+m$. As before, reaching $\mathcal{X}_{m}= \F^{N}$ will terminate the algorithm (if it has not terminated before that).
In the proposition below, the set $\mathcal{X}$ is exactly our terminating set $\mathcal{X}_{m}$ we referred to above.

\begin{Proposition}\label{Better-Achievability-Sparse}
After adopting the achievable scheme proposed in Theorem~\ref{Sparse-Theorem} together with adopting the corresponding conditions on $\rho,\delta$ and its corresponding $\mathbf{D}$ that was designed as a function of $\mathbf{F}$, then if there exists a subset $\mathcal{X} \supseteq \mathcal{X}_{\mathbf{F}, \mathbf{D}}, \mathcal{X} \subseteq \F^{N}$, that is $\rho N $-covered by $\mathcal{C}_{\mathbf{D}}$ for some $ \rho = H_q^{-1}(\frac{K}{N} - (1 - \frac{ \log_q(|\X|)}{N}))$, we can conclude that the computation cost $ \gamma = H_q^{-1}(\frac{K}{N} - (1 - \frac{ \log_q(|\X|)}{N}))$ is achievable. If $\mathcal{X} = \mathcal{X}_{\mathbf{F}, \mathbf{D}}$, then the computation cost converges to the optimal $H_q^{-1}(\frac{\log_q(L)}{N})$. The above remains in place for any $\mathbf{D}$ which yields communication cost no less than $ \Delta = O(\sqrt{\log_q(N)})$.
\end{Proposition}
\begin{proof}
The proof can be found in Appendix~\ref{F}.
\end{proof}

\begin{subsection}{Discussing the Results of the Current Section} \label{Discussion}

Theorem~\ref{Sparse-Theorem} reveals that the optimal computation cost lies in the region $\gamma \in (H_q^{-1}(\frac{\log_q(L)}{N}), H_q^{-1}(\frac{K}{N} ))$, and that this cost can be achieved with communication cost that vanishes as $\delta \doteq \sqrt{\log_q(N)}/N$. To get a better sense of the improvements that come from our coded approach, let us compare this to the uncoded case. Looking at Figure~\ref{Trivial_Example_1}, this uncoded performance is described by (the line connecting) point 1 and point 2.
Point 1, located at ($\gamma = 1/N,\delta = 1$), corresponds to the fully decentralized scenario where each server must compute just one subfunction\footnote{Due to the single-shot assumption, this corresponds to having $N(q-1)=L$. This matches the converse --- in our large $N$ setting --- because after writing $L =  N (q-1) = \binom{N}{1}(q-1) \simeq q^{N H_q^{}(1/N)} = q^{N H_q^{}(\gamma)}$, we see that $H_q^{-1}(\frac{\log_q(L)}{N}) = \gamma = \frac{1}{N}$.}, but which in turn implies that each server must communicate to all $K$ users. This scenario corresponds to the decomposition $\mathbf{D} \mathbf{I} = \mathbf{F}$ where we maximally\footnote{Note that the stated $\delta = 1$ accounts for the worst-case scenario where $\mathbf{F}$ contains no zero elements.} sparsify $\mathbf{E}$ by setting it equal to $\mathbf{E} = \mathbf{I}_{N \times N}$.

On the other hand, point 2, located at ($\gamma = 1,\delta = 1/N$), corresponds to the fully centralized scenario where each of the $K$ activated servers\footnote{This number of activated users is again a consequence of the single-shot assumption.} is asked to compute all $L$ subfunctions, but where now each server need only transmit to a single user.
Point 5 is a trivial converse.

\begin{figure}
      \centering
      \[\includegraphics[scale=0.7]{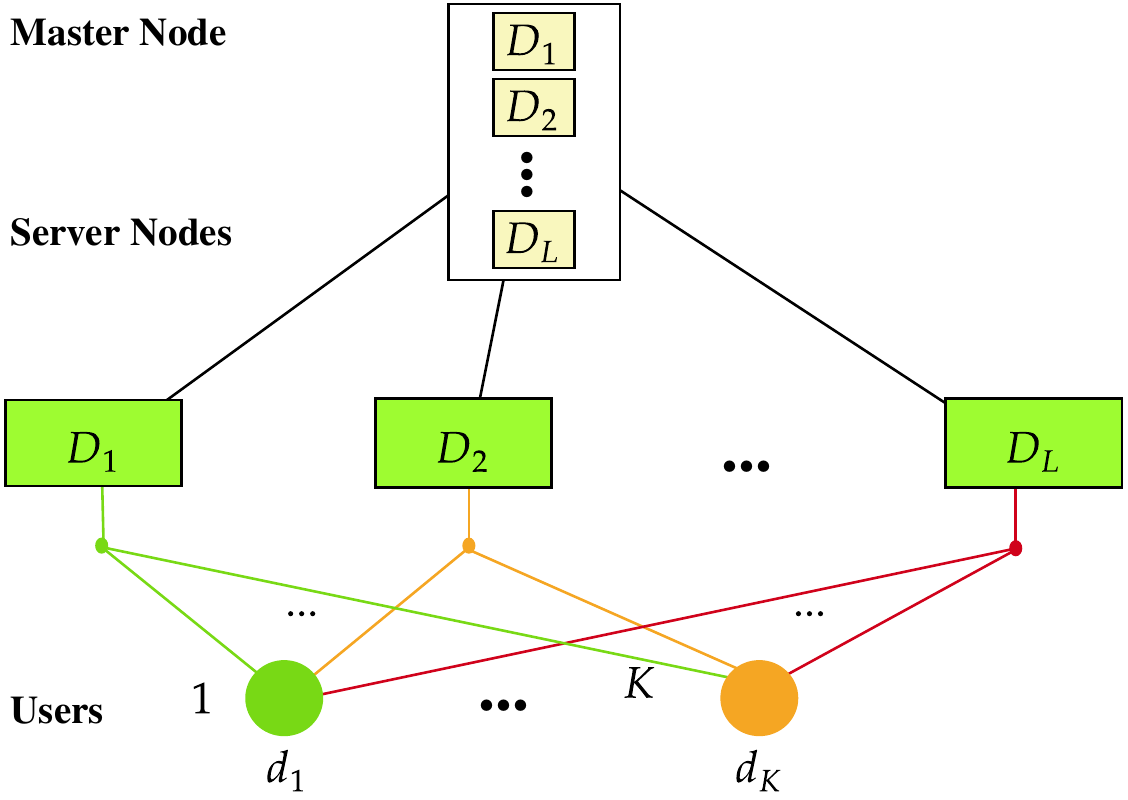}  \ \ \ \  \includegraphics[scale=0.7]{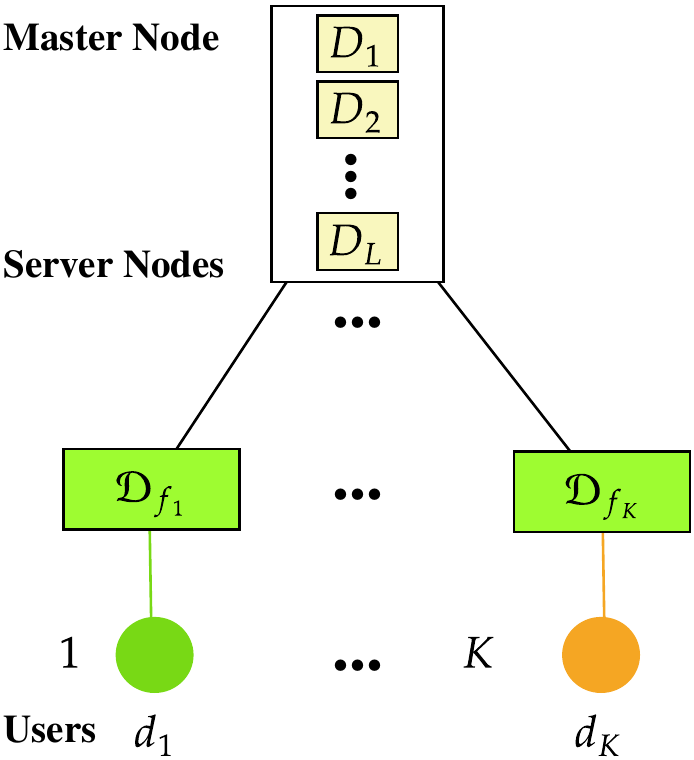}\]
      \caption{(Left. Fully decentralized): Uncoded scheme for point 1 corresponding to ($\gamma = 1/N,\delta = 1$). Each of the $N(q-1)=L$ servers, computes one subfunction, but must send to all $K$ users. (Right. Fully centralized):  Uncoded scheme for point 2 corresponding to ($\gamma = 1,\delta = 1/K$). $K$ activated servers, each computing $L$ subfunctions, and each transmitting to a single user.
    }
      \label{Trivial_Example_1}
  \end{figure}

\begin{figure}
      \centering
      \includegraphics[scale=0.8]{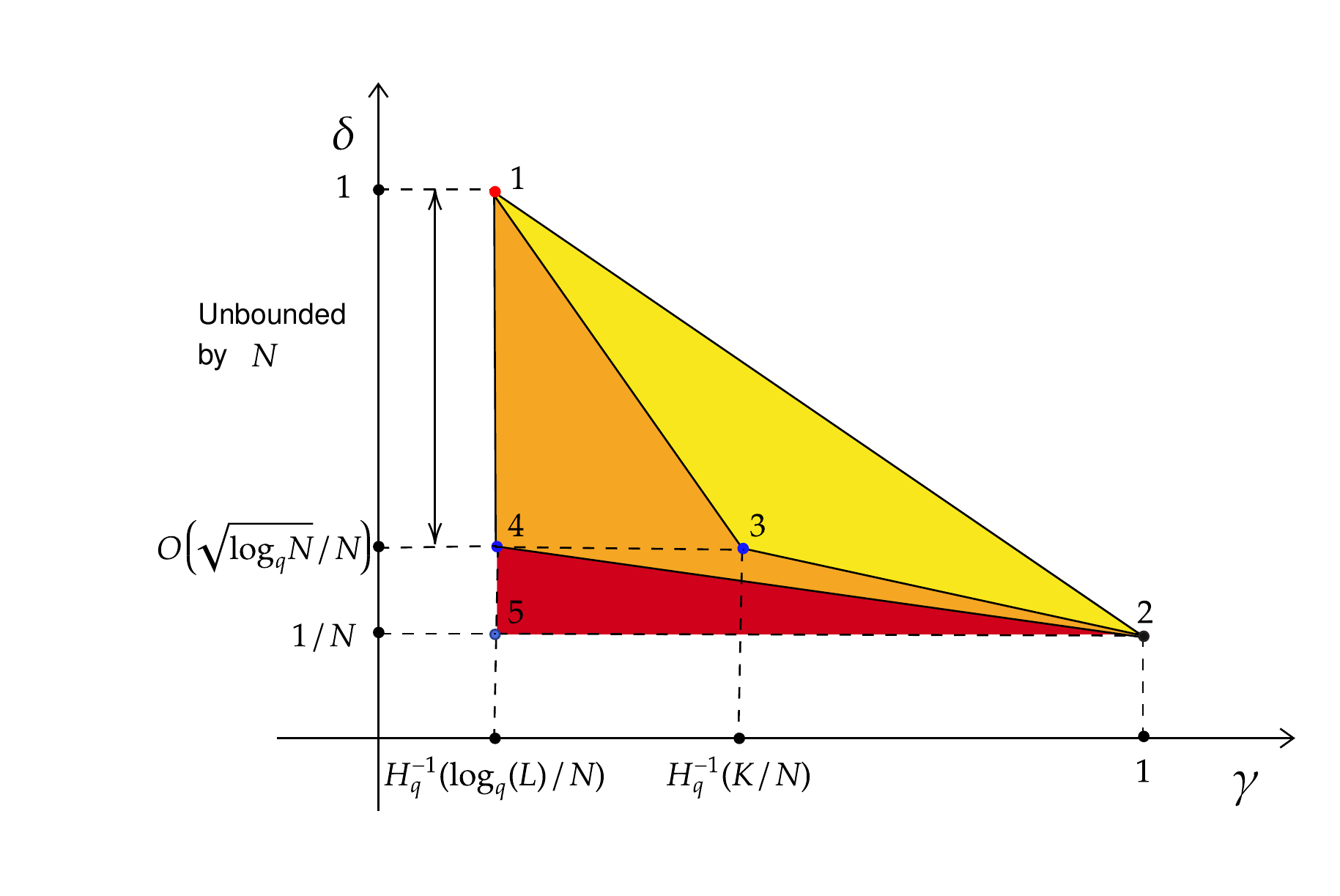}
      \caption{The figure summarizes the results of Theorem~\ref{Sparse-Theorem}. Recall that while $N$ is asymptotically large, both $K/N$ and $\log_q(L)/N$ are fixed.}
      \label{Comparison-fig}
  \end{figure}

From Theorem~\ref{Sparse-Theorem}, we now know that point $3$ at $(\gamma = H_q^{-1}(\frac{K}{N} ), \delta \doteq\sqrt{\log_q(N)})/N)$ is a guaranteed achievable point, and so is any point inside the triangle defined by points $1,2,3$. Any point inside the region defined by points $1,4,2,3$, is conditionally achievable in accordance to Theorem~\ref{Better-Achievability-Sparse}, and in particular in accordance to Corollary~\ref{Better-Achievability}. The converse also tells us that no point to the left of point $4$, i.e., no point with $\gamma < H^{-1}_q(\frac{\log_q(L)}{N})$, can be achieved.  Finally, the points inside the triangle defined by corner points $5,2,4$ could conceivably be achievable under additional techniques --- not covered in this paper --- that manage to further increase the sparsity of $\mathbf{D}$.

\end{subsection}

\section{Distributed Computing of Linearly-Separable Functions with Multi-Shot Communications ($T>1$)}\label{Multi-Shot}
In this section we present our results for the multi-shot setting where each server is able to broadcast $T$ consecutive transmissions to $T$ potentially different subsets of users. This is mainly motivated by the fact that having $T>1$, naturally allows us to employ fewer servers, but it is also motivated --- as we will discuss later on --- by an additional coding flexibility and refinement that multiple transmissions can provide. We briefly note that we assume as before that $K$ and $N$ are sufficiently large. 

\subsection{Problem Formulation}\label{Formulating-Multi}
The notation of the parameters that characterize the system will now generally follow directly from Section \ref{Formulating}, sometimes after clarifying the corresponding time-slot $t$ of interest.
For example, as before we will have 
\begin{align}
\mathbf{f}&\triangleq [F_1,F_2,\hdots,F_K]^{\intercal}, \\
\mathbf{f}_k &\triangleq [f_{k,1},f_{k,2},\hdots,f_{k,L}]^{\intercal},\: k \in [K],\label{function-vectors}\\
    \mathbf{w}&\triangleq [W_{1},W_{2},\hdots,W_{L}]^{\intercal},\label{message-vectors}\\
          \mathbf{f}'&\triangleq [F'_1,F'_2,\hdots,F'_K]^{\intercal},\\
             \mathbf{F}& \triangleq [\mathbf{f}_1,\mathbf{f}_2,\hdots,\mathbf{f}_K]^{\intercal}. \label{Demand-Matrix}    \end{align}
On the other hand, the notation for the encoding coefficients and the corresponding transmitted symbols during slot $t$, will now take the slightly modified form
    \begin{align}
    \mathbf{e}_{n,t} &\triangleq [e_{n,1,t},e_{n,2,t},\hdots, e_{n,L,t}]^\intercal,\: n \in [N],\: t \in [T], \label{encoding-vectors-per-shot-T}\\
       \mathbf{z}_{t} &\triangleq [z_{1,t},z_{2,t},\hdots, z_{N,t}]^ \intercal,\: t \in [T],\\
      \mathbf{z} &\triangleq [\mathbf{z}_{1}^{\intercal},\mathbf{z}^{\intercal}_{2},\hdots, \mathbf{z}^{\intercal}_{T}]^ \intercal\:
      \end{align}
with the corresponding modified
    \begin{align}
        \mathbf{E}_{t} &\triangleq [\mathbf{e}_{1,t},\mathbf{e}_{2,t},\hdots, \mathbf{e}_{N,t}]^\intercal,\: t \in [T] \label{encoding-vectors-per-shot-T-1}
\end{align}
while the corresponding decoding coefficients will now take the form
     \begin{align}
       \mathbf{d}_{k,t} &\triangleq [d_{k,1,t},d_{k,2,t},\hdots, d_{k,N,t}]^ \intercal,\: k \in [K], t \in [T], \label{decoding-vectors-per-shot}\\
        \mathbf{d}_k &\triangleq [\mathbf{d}_{k,1}^{\intercal},\mathbf{d}_{k,2}^{\intercal},\hdots, \mathbf{d}_{k,T}^{\intercal}]^ \intercal,\: k \in [K].\label{decoding-vectors}
\end{align}
We note that the decoding coefficients are decided as a function of all received signals throughout all $T$ transmissions.

As before, (cf.~\eqref{DefinitionOfLSFunctions}), we have that
\begin{align}
    \mathbf{f} =[\mathbf{f}_1,\mathbf{f}_2,\hdots,\mathbf{f}_K]^{\intercal} \mathbf{w}\label{Functions}
\end{align}
and now we use
\begin{align}
    \mathbf{z}_t = \mathbf{E}_t \mathbf{w} =[\mathbf{e}_{1,t}, \mathbf{e}_{2,t}, \hdots,\mathbf{e}_{N,t}]^\intercal \mathbf{w} \label{EncodedCashedData}
\end{align}
to denote the $t$-th slot transmission vector across all servers. The set of all transmissions now takes the form
\begin{align}
    \mathbf{z} = \mathbf{E} \mathbf{w}
\end{align}
where now the encoding matrix takes the form
\begin{align}
    \mathbf{E} \triangleq[\mathbf{E}^{\intercal}_1,\mathbf{E}^{\intercal}_2, \hdots, \mathbf{E}^{\intercal}_{T}]^{\intercal} \in \F^{NT \times L}. \label{EncodingMatrix}
\end{align}
Upon decoding, each user $k$ generates
\begin{align}
    F'_k= \mathbf{d}_{k}^{T} \mathbf{z}
\end{align}
and the cumulative set of all decoded elements across the users takes the form
\begin{align}
    \mathbf{f}' = [\mathbf{d}_1,\mathbf{d}_2,\hdots,\mathbf{d}_K]^{\intercal} \mathbf{z}.\label{DecodedData}
\end{align}
Now we note that our decoding matrix
\begin{align}
    \mathbf{D} \triangleq [\mathbf{d}_1,\mathbf{d}_2, \hdots , \mathbf{d}_K]^{\intercal} \in \F^{K \times NT} \label{DecodingMatrix}
\end{align}
is of dimension ${K \times NT}$. Naturally, correct decoding requires
\begin{align}
    \mathbf{f}=\mathbf{f}'\label{feasibility}
\end{align}
and after substituting \eqref{Functions}, \eqref{EncodedCashedData}, \eqref{DecodedData} into \eqref{feasibility}, we can conclude as before that computing succeeds if and only if
\begin{align}
    \mathbf{D}\mathbf{E} = \mathbf{F}.\label{MainEquation-Multi}
\end{align}
The problem remains similar to the one in the single-shot scenario, except that now our decoding matrix $\mathbf{D} \in \F ^ {K \times NT}$ and encoding matrix $\mathbf{E} \in \F^{NT \times L}$ are bigger\footnote{The size of $\mathbf{F} \in \F^{K \times L}$ remains the same, and thus again we have $L \leq q^{K}$.} and can have a certain restrictive structure.

Again similar to before, each server $n \in [N]$ is asked to compute all the subfunctions in $\cup^{T}_{t=1} \sup(\mathbf{e}_{n,t})$, and thus equivalently the set of servers $\W_{\ell}$ that must compute subfunction $f_\ell(D_\ell)$, takes the form
\begin{align}
   \W_{\ell} = \cup^{T}_{t=1} \sup(\mathbf{E}([(t-1)N+1 : tN], \{\ell \})^{\intercal}), \forall \ell \in [L],\: \forall t \in [T].\label{New-Computaton-Cost-2}
\end{align}

The following theorem provides an achievable upper bound on the computation cost $\gamma$ of our distributed computing setting for the multi-shot scenario.

\begin{Theorem}\label{Achievability-multi}
For the setting of distributed-computing of linearly-separable functions, with $K$ users, $N$ servers, $L$ sub-functions and $T$ shots, the optimal computation cost $\gamma$ is upper bounded by
  \begin{align}
      \gamma &\leq T H_q^{-1}(\frac{K}{NT} ).\label{Achievable-interval-T}
  \end{align}
\end{Theorem}
\begin{proof}
We first note that, directly from~\eqref{New-Computaton-Cost-2} and the union bound, we have that \begin{align}
    \underset{\ell \in [L]}{\max}\: \omega(\mathbf{E}(:,\ell)) \geq \max_{\ell \in [L]} |\W_{\ell}| \label{bound-on-the-commp-cost}
\end{align}
and thus our normalized computation cost will be upper bounded as \[\gamma \leq {  \underset{\ell \in [L]}{\max}\: \omega(\mathbf{E}(:,\ell))}/{N}.\]
To bound $\underset{\ell \in [L]}{\max}\: \omega(\mathbf{E}(:,\ell))$, we apply covering code arguments as in the single-shot case, after though accounting for the dimensionality change from having larger matrices. In particular, this means that now the corresponding covering code $\mathcal{C}(n,k)$ will have $n=NT$ and again $K =n-k$ (now we only ask that $NT \geq K$). To account for this increase in $n$, we note that while the computation cost must still be normalized by the same number of servers $N$, when considering our covering code\footnote{Let us quickly recall that in the previous single-shot scenario, a covering code with covering radius $\rho n = \rho N$ implied a computation cost of $\gamma N= \rho N$ and thus a normalized computation cost of $\gamma= \rho$.}, we must consider a radius $\frac{\gamma}{T} n = \frac{\gamma}{T} NT  = \gamma N$ to guarantee our computation constraint. In other words, the $\rho$-covering codes that will guarantee the computation constraint, will be for $\rho = \gamma /T$. Consequently, combined with the aforementioned union bound, we now see that $\rho T$ serves as an achievable upper bound on $\gamma$. The rest follows directly from the proof of the corresponding theorem in the single-shot scenario.
\end{proof}

The following two propositions help us make sense of the computational effect of having $T>1$.
\begin{Proposition}\label{Proposition-T-1}
\label{T-inf}
In the distributed computing setting of interest in the limit of large $T$, the normalized computation cost $\gamma$ vanishes to zero.
\end{Proposition}
\begin{proof}
The proof is direct once we prove that for any fixed $c$, then
\begin{align}
  \lim_{T \rightarrow{ \infty}}
  T H_q^{-1}(c/T) =0.
\end{align}
This property will be proved in Appendix~\ref{N}.
\end{proof}

In a system with an unchanged number of users and servers, the above reveals the notable (unbounded) computational advantage of allowing a large number $T$ of distinct transmissions per server. This advantage of the multi-shot approach must be seen in light of the fact that in the single-shot approach, the computation cost $\gamma$ was always bounded below by a fixed $\gamma \geq H_q^{-1}(\frac{\log_q(L)}{N})$, irrespective of the communication cost. Consequently we can deduce that the computational gains that we see in the regime of larger $T$, are --- at least partly --- a result of the increased refinement in transmission that a larger $T$ allows, and it should not be solely attributed to an increased communication cost.

The following proposition discusses the non-asymptotic computational effect of increasing $T$ beyond $1$. Recall that our results hold for sufficiently large $K$ and $N$.
\begin{Proposition}\label{Derivative}
For $q=2$, then $\gamma$ monotonically decreases in $T$, while for $q>2$ then $\gamma$ monotonically decreases in $T$ after any $T\geq \lceil\frac{K}{NH^{-1}_q(1/q)}\rceil$.
\end{Proposition}
\begin{proof}
The proof is based on the fact that the derivative of $f=T H_q^{-1}(c/T), 0\leq c/T \leq 1-1/q$, with respect to $T$, satisfies
\begin{align}
{\frac{\partial f}{\partial T}} = \frac{H_q(f/T)}{\log_q\left( \frac{f/T}{1-f/T}  (q-1)\right)} + f/T.\label{the-t-emtropy-function}
\end{align}
This is proved in Appendix~\ref{O}. From the above, and after observing that ${\frac{\partial f}{\partial T}} \leq 0$ where $0\leq H_q^{-1}(K/NT) = f/T\leq 1/q$, we can conclude that since $0\leq H_q^{-1}(K/NT) = f/T \leq 1/2=1-1/q$, then for $q=2$, increasing $T$ always strictly reduces $\gamma$.  On the other hand, when $q>2$, this reduction happens --- as we see above --- when $T\geq \lceil T_0 \rceil $ for some real $T_0$ for which $H_q(K/N{T_{0}}) = 1/q$.
\end{proof}

\section{Conclusions}\label{conclusion}
In this work we have introduced a new multi-user distributed-computation setting for computing from the broad class of linearly-separable functions.

Our work revealed the link between distributed computing and the problem of factorizing a `functions' matrix $\mathbf{F}$ into a product of two preferably sparse matrices, these being the encoding matrix $\mathbf{E}$ and the decoding matrix $\mathbf{D}$. The work then made the new connection to the area of covering codes, revealing for the first time the importance of these codes in distributed computing problems, as well as in sparse matrix factorization over finite fields. Furthermore, this work here brought to the fore the concept of partial covering codes, and the need for codes that cover well smaller subsets of the ambient vector space. For this new class of codes --- which constitute a generalization of covering codes --- we have provided some extensions and generalizations of well-studied results in the literature.

Our two metrics --- $\gamma$, representing the maximum fraction of all servers that must compute any subfunction, and $\delta$, representing the average fraction of servers that each user gets data from --- capture the computation and communication costs, which are often at the very core of distributed computing problems. The observant reader might notice that the creation of $\mathbf{E}$ entails a complexity equal to that of syndrome decoding. Our results hold unchanged when we consider --- as suggested before --- that the computational cost of evaluating the various subfunctions, far exceeds all other costs. What the results reveal is that in the large $N$ regime, the optimal computation cost lies in the region $\gamma \in (H_q^{-1}(\frac{\log_q(L)}{N}), H_q^{-1}(\frac{K}{N} ))$, and that this entails the use of a vanishingly small fraction $\delta \doteq \sqrt{\log_q(N)}/N$ of all communication resources. What we show is that our coded approach yields unbounded gains over the uncoded scenario, in the sense that the ratio $\frac{\delta_{un}(\gamma)}{\delta(\gamma)}$ between the uncoded and coded communication costs, is unbounded. 

We have also studied the multi-shot setting, where we have explored the gains over the single-shot approach. What we now know is that the gains from increasing $T$, are unbounded (and strictly increasing) in the regime of large $T$, whereas in the regime of finite $T$, the gains are strictly increasing after some threshold value of $T$. We are thus able to conclude, as suggested before, that computation reductions due to larger $T$, are --- at least partly --- a result of the increased refinement in transmission that a larger $T$ allows, and that these gains should not be interpreted as being purely the result of an increased communication load.

Our work naturally relates partly to the recent results in~\cite{kai1} that considered the single-user linearly-separable distributed computing scenario, where a single user may request multiple linearly-separable functions. In this setting in~\cite{kai1}, as well as in the extended works in~\cite{kai2} and \cite{kai3}, a key ingredient is the presence of struggling servers, while another key ingredient is that the subfunction-assignment is fixed and oblivious of the actual functions requested by the user.  In this context, the coefficients of the functions are assumed to be distributed uniformly and \emph{i.i.d}, and the decodability is probabilistic. There is also an interesting connection (cf.~\cite{das2013finite,dimakisCompressed}) between compressed sensing and coding theory. Naturally this connection entails no link to covering codes, as the problem of compressed sensing relates to decodability and is very different from the existence problem that we are faced with. 

As suggested above, our setting can apply to a broad range of `well-behaved' functions, and thus can enjoy several use cases, some of which are suggested in our introduction (see also~\cite{kai1} for additional motivation of the linearly separable function computation problem). When considering problems over the real numbers, we may consider a very large $q$. 
An additional new scenario that our work can extend is the so-called hierarchical or tree-like scenario introduced in  \cite{reisizadeh2021codedreduce,reisizadeh2019tree} whose purpose is to ameliorate bandwidth limitations and straggler effects in distributed gradient coding~\cite{lee2017speeding}. In this hierarchical setting, each user\footnote{In \cite{reisizadeh2021codedreduce}, these users are referred to as master nodes.} is connected to a group of servers in a hierarchical manner\footnote{In particular, each user computes a linearly separable function based on its locally available data, and then sends this to the `Aggregator' that finally computes the gradient.} that allows for a hierarchical aggregation of the sub-gradients. Our approach can extend the hierarchical model by allowing the users to connect to any subset of servers, as well as by allowing them to deviate from the single-shot assumption. Finally as one might expect, our analysis also applies to the transposed computing problem corresponding to $\mathbf{E}^{\intercal }\mathbf{D}^{\intercal} = \mathbf{F}^{\intercal}$. Additional considerations that involve stragglers, channel unevenness or computational heterogeneity, are all interesting research directions.

\newpage
\appendices

 \begin{figure}
      \centering
      \includegraphics[scale=0.7]{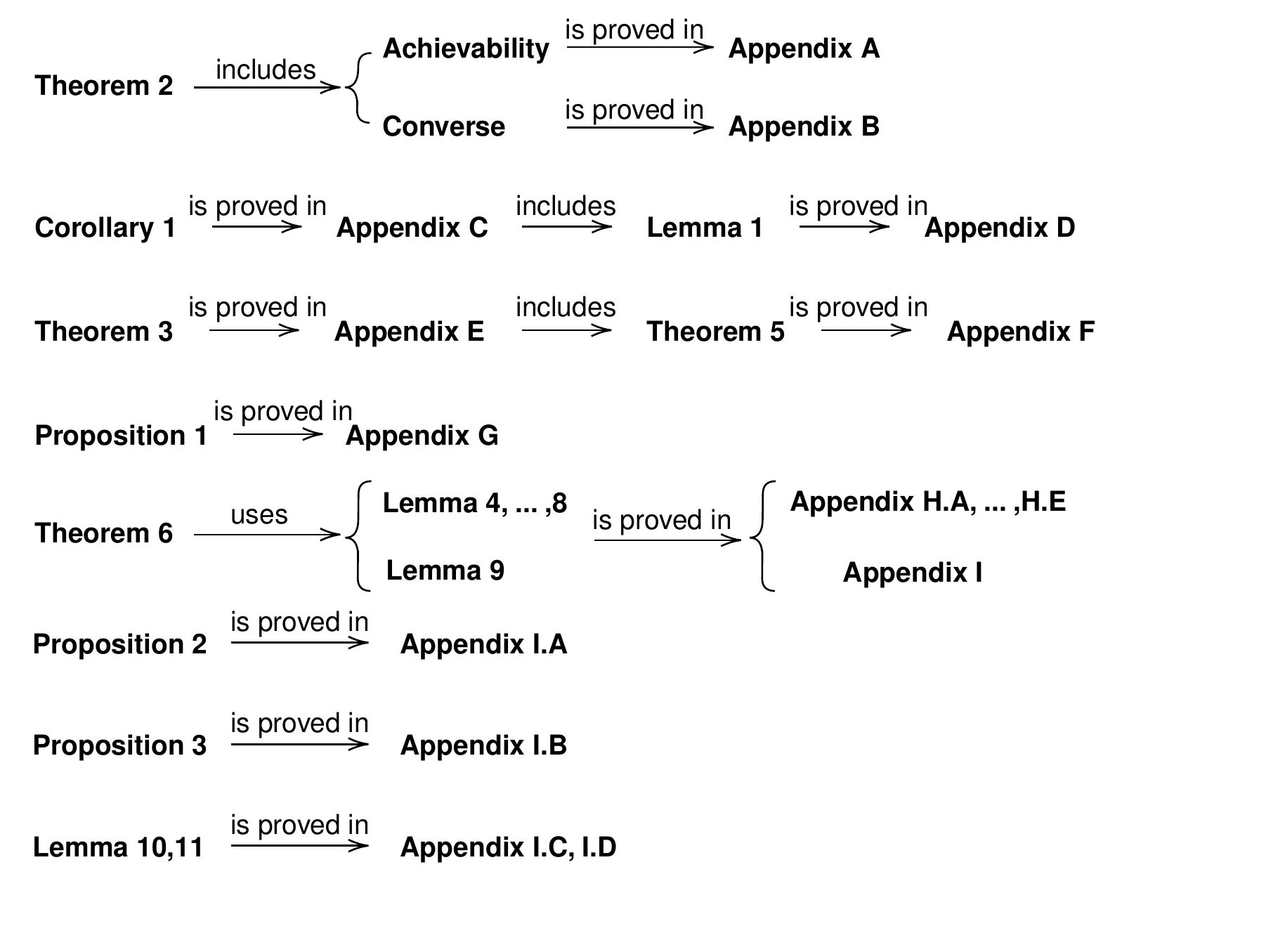}
      \caption{Map of lemmas, theorems and appendices.}
      \label{Tree-3}
  \end{figure}

\section{Proof of Converse in Theorem~\ref{Achievability}}\label{A_0}
To prove the converse in~\eqref{Achievable-interval}, we modify the sphere-covering bound for the case of partial covering codes. We wish to show that for a set $\mathcal{X}$ that satisfies $ \mathcal{X} \subseteq\F^n_q,\:|\mathcal{X}| = q^{k}L,\: k \in \mathbb{N}$,  a $(\rho,\mathcal{X})$-partial covering code $\mathcal{C}(k,n)$ has to  satisfy
\begin{align}
   \log_q(L) \leq \log_q(V_q(n,\rho)). \label{Partial-Covering-code-converse}
\end{align}
This is easy to show because having $q^{k}$ codewords directly means that the maximum number of points they can jointly $\rho n$-cover is equal to $q^{k} V_q(n,\rho)$. This in turn implies that
\begin{align}
    L q^{k} \leq V_q(n,\rho) q^{k}
\end{align}
which yields \eqref{Partial-Covering-code-converse} after taking the logarithm on both sides of the inequality.

Now letting the above $\mathcal{X}$ be the $\mathcal{X}$ found in Theorem~\ref{Bridge}, we note that if $|\mathcal{X}|=Lq^{k}$ then $\mathcal{X}=\mathcal{X}_F$. Then by substituting $N=n,K=n-k$, we see that $\log_q(L) \leq \log_q(V_q(N,\rho))$. Since $q^{NH_q(\rho)-o(N)} \leq V_q(N,\rho) \leq q^{NH_q(\rho)}$, we can conclude that $\log_q(L) \leq {NH_q(\rho)}$ and thus that $ H_q^{-1}(\frac{\log_q(L)}{N}) \leq \rho$, which concludes the proof.
\qed

\section{Proof of Achievability in Theorem \ref{Achievability}}\label{B}
Directly from~\cite{cohen1985good}, we know that there exists at least one $\rho$-covering code $\mathcal{C}_{\mathcal{X}}(k,n)$ that satisfies
\begin{align}
    n-k \geq \log_q(V_q(n,\rho)) -  2 \log_2(n)+ \log_q(n) - O(1).\label{First-Achievability}
\end{align}
Then applying Theorem~\ref{Bridge} with $\mathbf{D}=\mathbf{H}_{\mathcal{C}}, N=n,K=n-k$ and $\mathcal{X}=\F^{n}$, allows us to conclude that there exists a feasible scheme for the distributed computing problem, with computation cost $\gamma=\rho$, that satisfies
\begin{align}
    K/N \geq \log_q(V_q(N,\rho))/N - 2 \log_2(N)/N+\log_q(N)/N -O(1)/N.
\end{align}
Combining this with the fact that $q^{NH_q(\rho)-o(N)} \leq V_q(N,\rho) \leq q^{NH_q(\rho)}$, yields
\begin{align}
    K/N \geq H_q(\rho) -\epsilon(N)
\end{align}
which tells us that $\rho \leq H_q^{-1}(K/N+\epsilon(N))$, which in turn proves the result in the limit of large $N$. \qed

\section{Proof of Corollary \ref{Better-Achievability}}\label{A_1}
We first start with the following lemma which proves the existence of a $(\rho,\mathcal{X})$-partial covering linear code $\mathcal{C}$, for a properly-sized set $\mathcal{X} \subseteq \F^n$ that encloses $\mathcal{B}_q(0,\rho)$. Before proceeding with the lemma, we note that the lemma is an outcome of involving a linear greedy algorithm. Let us also briefly recall from Theorem~\ref{Achievability} and its proof in Appendix~\ref{A_0}, that $\log_q(L) \leq \log_q(V_q(n,\rho))$.

\begin{Lemma} \label{CohenLemma}
Let $\mathcal{X} \subseteq\F^n_q$ be a set of size $|\mathcal{X}| = L'q^{k}$ that satisfies $\mathcal{X} \supseteq \mathcal{B}_q(0,\rho)$. Then as long as
\begin{align}
   \log_q(L') \geq  \log_q(V_q(n,\rho)) - 2 \log_2(n) +\log_q(n) - O(1)\label{Partial-Covering-code}
\end{align}
there exists a $(\rho,\mathcal{X})$-partial covering code.
\end{Lemma}
\begin{proof}
The proof is found in Appendix \ref{A_2}.
\end{proof}
With this lemma in place, let us define $\mathcal{A}_{m} \triangleq \{\mathcal{X} \subseteq \F^{n} \ | \ |\mathcal{X}| = m,  \mathcal{X} \supseteq \mathcal{B}_q(0,\rho)\}$ to be the family of all subsets of $\F^{n}$ which have cardinality $m$ and which enclose $\mathcal{B}_q(0,\rho)$. Consider the following algorithm.
\begin{enumerate}
    \item Assign $m=Lq^{N-K}$.
    \item For each $\mathcal{X}$ in $\mathcal{A}_{m}$, find a $(\rho,\mathcal{X})$-partial covering code $\mathcal{C}_{\mathcal{X}}$ via (the algorithm corresponding to) Lemma \ref{CohenLemma}.
    \item For each $\mathcal{X}$ in $\mathcal{A}_{m}$, set $\mathbf{D}=\mathbf{H}_{\mathcal{C}_{\mathcal{X}}}$, and create $\mathcal{X}_{\mathbf{F}, \mathbf{D}} = \{\mathbf{x} \in \F^{N}| \mathbf{D} \mathbf{x} = \mathbf{F}(:,\ell), \ \text{for some} \ \ell \in [L]\}$.
    \item If there exists an $\mathcal{X}$ in $\mathcal{A}_{m}$, for which $\mathcal{X} \subseteq \mathcal{X}_{\mathbf{F}, \mathbf{D}}$, then output this $\mathcal{X}$ and its corresponding $\mathbf{D}=\mathbf{H}_{\mathcal{C}_{\mathcal{X}}}$ from the above step.
    \item If there exists no $\mathcal{X}$ in $\mathcal{A}_{m}$ for which $\mathcal{X} \subseteq \mathcal{X}_{\mathbf{F}, \mathbf{D}}$, then increase $m$ by one and go back to step~$2$.
\end{enumerate}

Let us continue now by supposing that the scheme terminates, outputting $\mathbf{D}$ and $\mathcal{X}$ at the fourth step, before $m$ reaches $m = q^{N}$. Lemma~\ref{CohenLemma}, which guarantees (cf.~\eqref{Partial-Covering-code}) that
\begin{align}\label{partial-cov2}
      \log_q(|\mathcal{X}|q^{-k}) \geq  \log_q(V_q(n,\rho)) - 2 \log_2(n) +\log_q(n) - O(1)
\end{align}
also guarantees that
\begin{align}\label{partial-cov3}
  \frac{\log_q(|\mathcal{X}|){-(N-K)}}{N} \geq  \log_q(V_q(N,\rho))/N - 2 \log_2(N)/N +\log_q(N)/N - O(1)/N
\end{align}
{\color{black}where this last inequality holds after setting $N=n,K=n-k$, and after we divide both sides of~\eqref{partial-cov2} by $N$, and then apply Theorem~\ref{Bridge} after recalling that $\mathcal{X}$ is indeed $\rho n$-covered by $\mathcal{C}_{\mathbf{D}}$.}
Applying that $q^{NH_q(\rho)-o(N)} \leq V_q(N,\rho) \leq q^{NH_q(\rho)}$ into \eqref{partial-cov3}, gives
\begin{align}
    (\frac{K}{N}-1+ \frac{\log_q(|\mathcal{X}|)}{N}) \geq H_q(\rho) - \epsilon(N)
\end{align}
telling us that the algorithm yields a scheme with computation cost\footnote{Here it is worth elaborating on a fine point regarding our metric. As the reader may recall, $\gamma$ describes the fraction of \emph{active} (non-idle) servers that compute any subfunction. Then the observant reader may wonder if our proposed scheme indeed activates all existing servers. This corresponds to having a scheme with an $\mathbf{E}$ matrix that has no all-zero rows. In the (rare) degenerate scenario where a row of $\mathbf{E}$ may contain only zeros, then our derived computation cost $\gamma$ would --- by definition --- have to be recalculated (to account for having idle servers) and would be higher than stated here. To account for this degenerate case, we add a small step in our algorithm which reduces the recorded computation cost by guaranteeing that all servers are active. This step simply says that if a row in $\mathbf{E}$ contains only zeros, then this row is substituted by an arbitrary non-zero row (let's say, the first row) of $\mathbf{E}$, except that, if that (first) row contains a non-zero element in the position $\ell_{\max} = \arg\max \:\omega(\mathbf{E(:,\ell)})$, then this element is substituted by a zero. Then the two servers (the first server and the previously idle server) will split their communication load, except that the server corresponding to the originally all-zero row, will not send any linear combination that involves $w_{\ell_{\max}}$. This small modification guarantees that whatever $\gamma$ we declare here as being achievable, is indeed achievable even in degenerate scenarios. Finally, this degenerate scenario does not affect the algebraic converse.}
\begin{align}
    \gamma = \rho \leq H_q^{-1}(K/N-1+\log_q(|\mathcal{X}|/N)+\epsilon(N))\label{approximate-formula}
\end{align}
which matches the stated result in the regime of large $N$. Note that when $\mathcal{X} = \mathcal{X}_{\mathbf{F}, \mathbf{D}}$, then naturally $|\mathcal{X}|= L q^{N-K}$ which, directly from~\eqref{approximate-formula}, yields $\gamma = \rho = H_q^{-1}(\log_q(L)/N+\epsilon(N))$. At the other extreme, when the algorithm terminates at the very end when $\mathcal{X}=\F^{N}$, then the corresponding code will be the standard $\rho$-covering code (see Appendix~\ref{B}), {\color{black}and the computation cost will correspond to $ \gamma = H_q^{-1}(K/N)$.}
\qed
\section{Proof of Lemma~\ref{CohenLemma}} \label{A_2}
We here start by employing the recursive construction approach of Cohen and Frankl in~\cite{cohen1985good}. This recursive approach builds an $(n,j+1)$ code $\mathcal{C}_{j+1}$ from a previous $(n,j)$ code $\mathcal{C}_j$, by carefully adding a vector $\mathbf{x}$ on the basis of $\mathcal{C}_j$, so that now the new basis span is bigger. Our aim will be to recursively construct ever bigger codes that cover an ever increasing portion of our set $\mathcal{X}$.

Let us start by setting $\mathcal{C}_{0} =\{\mathbf{0}\}$. Let us then make the assumption that the aforementioned integer $L'$ in Lemma~\ref{CohenLemma}, takes the form
\begin{align}
    L' = q^{n-k'} \label{k'-to-L}
\end{align}
for some real $k'\geq k$. Let $Q(\mathcal{C})$ denote the set of points in $\mathcal{X}$ that are not $\rho n$-covered by $\mathcal{C}$, and let
\begin{align}
    q(\mathcal{C}) \triangleq \frac{|Q(\mathcal{C})|}{q^{n +k -k'}}
\end{align}
where naturally
\begin{align}
   | {Q}{(\mathcal{C}_{0})}| = q^{n+k-k'} - V_q(n,\rho)
\end{align}
and
\begin{align}
   q(\mathcal{C}_{0}) = 1 - V_q(n,\rho) q^{-(n+k-k')}. \label{initialization}
\end{align}
To proceed, we need the following lemma from~\cite{cohen1985good}.
\begin{Lemma}[\cite{cohen1985good}]\label{linear-greedy}
Let $\mathcal{Y}\subseteq \F^{n},\mathcal{Z}\subset \F^{n}$, and consider $\mathcal{Y}+\mathbf{x} = \{\mathbf{y}+ \mathbf{x}: \mathbf{y} \in \mathcal{Y}\}$ for some $\mathbf{x} \in \F.$  Then
\begin{align}
\mathbb{E}(|(\mathcal{Y} + \mathbf{x}) \cap \mathcal{Z}|)=q^{-n} |\mathcal{Y}||\mathcal{Z}|
\end{align} where the average is taken, with uniform probability, over all $\mathbf{x} \in \F^{n}$.
\end{Lemma}

Now, we develop the proof in two parts.
\begin{enumerate}
    \item \textbf{Binary Case:}
    The proof for $q=2$ where $k=k'$ (corresponding to the singular case of maximal $L=2^K$) has been presented in \cite{cohen1983nonconstructive} and \cite{cohen1985good} in two different ways. We will modify the latter approach to establish our claim for any $k'\geq k$ (which will allow us to also handle $L$ values that are smaller than $2^K$).
    First let us easily deduce from Lemma~\ref{linear-greedy} that there exists an $\mathbf{x} \in \F^n$ for which $|(\mathcal{Y}+\mathbf{x}) \cap \mathcal{Z}| \leq \frac{|\mathcal{Y}||\mathcal{Z}|}{q^n}$.
Now let us set $\mathcal{Y} = \mathcal{Z} = Q(\mathcal{C}_{j})$, and let us append a vector $\mathbf{x}$ to the generator matrix of $\mathcal{C}_{j}$ to create $\mathcal{C}_{j+1}$, where $\mathbf{x}$ is chosen to minimize $|\mathcal{Q}(\mathcal{C}_{j+1})|$. Now we can directly verify that
    \begin{align}
       |\mathcal{Q}(\mathcal{C}_{j+1})|  = |\mathcal{Q}(\mathcal{C}_{j}) \cap  \mathcal{Q}(\mathcal{C}_{j}+\mathbf{x})| =|\mathcal{Q}(\mathcal{C}_{j}) \cap  (\mathcal{Q}(\mathcal{C}_{j})+\mathbf{x})| \leq |\mathcal{Q}(\mathcal{C}_{j})|^2/2^{n}
    \end{align}
    which implies that
    \begin{align}
    q(\mathcal{C}_{j+1}) \leq q(\mathcal{C}_{j})^2 2^{k-k'} \leq q(\mathcal{C}_{j})^2 \label{Decent}
    \end{align}
    where the latter inequality holds because $k'\geq k$.
    Combining~\eqref{initialization} and~\eqref{Decent}, gives
    \begin{align}
        q(\mathcal{C}_{k}) &\leq q(\mathcal{C}_{0})^{2^k} \leq (1 - V_q(n,\rho) 2^{-(n-k'+k)})^{2^k}
    \end{align}
    where the latter inequality again holds due to the fact that $k'\geq k$. Now let us continue this recursion until $k$ is such that
    \begin{align}
        2^k = \lceil(n-k'+k) 2^{(n-k'+k)} \ln(2) / V_2(n,\rho)  \rceil\label{Choice-of-k-binary-case}
    \end{align}
  in which case
  --- given that $(1 - \frac{1}{x})^{x} \leq e^{-1}, \: \forall x \geq 1$ --- we get that
  \begin{align}
        q(\mathcal{C}_k) < 2^{-(n+k-k')}
    \end{align}
    which automatically yields that $Q(\mathcal{C}_k) =0$. This, again with the choice of $k$ in~\eqref{Choice-of-k-binary-case}, tells us that for a set $\mathcal{X}$ that satisfies $\mathcal{B}_q(0,\rho) \subseteq  \mathcal{X} \subseteq\F^n_q,\:|\mathcal{X}| = L q^{k}$, then indeed there exists a $(\rho,\mathcal{X})$-partial covering code $\mathcal{C}(n,k)$ satisfying
 \begin{align}
     &0 \leq  \log_q(L / V_q(n,\rho)) + 2 \log_2(\log_q(|\X|)) -\log_q(\log_q(|\mathcal{X}|)) + O(1).\
 \end{align}
This conclusion can be considered as a tighter version of Lemma~\ref{CohenLemma}. After a few very basic algebraic manipulations we get the proof of Lemma~\ref{CohenLemma}, for the binary case of $q=2$.

    \item \textbf{Non-Binary Case:}
    Considering first an arbitrary $\mathcal{Z}\subset \F^{n}$,
we have that
    \begin{align}
        \mathbb{E}(1 - (q^{-n + k' - k}|(\mathcal{Z}+ \mathbf{x}) \cup \mathcal{Z}|)) &=
           \mathbb{E}(1 - q^{-n + k' - k}((|(\mathcal{Z}+ \mathbf{x})|+|\mathcal{Z}|) - |(\mathcal{Z}+ \mathbf{x}) \cap \mathcal{Z}|))\label{bounding-union-begin}\\ &=
           \mathbb{E}(1 - 2 q^{-n  + k' - k} |\mathcal{Z}| + q^{-n  + k' - k}|(\mathcal{Z}+ \mathbf{x}) \cap \mathcal{Z}|) \\&\overset{(a)}{=} 1- 2 q^{-n + k' -k} |\mathcal{Z}| +  q^{-2n  + k' - k} |\mathcal{Z}|^2 \\& \overset{(b)}{\leq}
            1 - 2 q^{-(n -k'+k)} |\mathcal{Z}| +  q^{-2(n  - k' + k)} |\mathcal{Z}|^2
             \\&= (1 - \frac{|\mathcal{Z}|^{}}{q^{(n - k'+k)}})^{2} \label{bounding-union}
    \end{align}
    where (a) is directly from Lemma~\ref{linear-greedy}, and where (b) holds since $k'\geq k$.
Similarly to the binary case, we begin with $\mathcal{C}_{0} =\{\mathbf{0}\}$, and again recursively extend as
\begin{align}
  \mathcal{C}_{j+1} = <\mathcal{C}_{j};\mathbf{x}>
\end{align}
where $\mathbf{x}$ is chosen so that $|\mathcal{Z}|$ is maximized. We do so, after again setting $\mathcal{Z} = Q(\mathcal{C}_{j})$.

At this point, from~\eqref{bounding-union} we have that
\begin{align}
     q(\mathcal{C}_{j+1}) \leq q(\mathcal{C}_{j})^2. \label{desending}
\end{align}
We now consider the following lemma from \cite{cohen1985good}.
\begin{Lemma}(\cite[Lemma 2]{cohen1985good})\label{Lemma-2of5}
For any fixed $\mathcal{Z} \subseteq \mathcal{X}\subset \mathbb{F}^n$ where $|\mathcal{Z}|q^{-({n-k'+k})} = \epsilon < (q(n-k'+k))^{-1}$, then
\begin{align}
    \mathbb{E}_{\mathbf{x}\in\mathbb{F}^n}(1 -q^{-(n-k'+k)} |\cup_{\alpha \in \F_q} \mathcal{Z} + \alpha \mathbf{x}|) \leq (1 - \epsilon)^{q(1-(2(n-k'+k))^{-1})}.
\end{align}
\end{Lemma}

Continuing from $\mathcal{Z} = \mathcal{X} \cap (\cup_{\mathbf{c}  \in \mathcal{C}_{j}} \mathcal{B}_q(\mathbf{c},\rho)) $, where \[|\mathcal{Z}| < \frac{1}{n}q^{(n-k'+k -1)}, \ \ \ \ q(\mathcal{C}_{j+1}) \leq q(\mathcal{C}_{j})^{q(1-(2(n-k'+k)^{-1}))}\] we have that
\begin{align}
    q(\mathcal{C}_{j+1}) \leq (1 - q^{n-k'+k} V_q(n,\rho))^{(q(1-(2(n-k'+k))^{-1}))^{j}} \leq (1 - q^{n+k-k'} V_q(n,\rho) )^{e^{-0.5} q^{j}}
\end{align}
since $(1 -(2(n-k'+k))^{-1}) \geq (1 - (2(n-k'+k))^{-1})^{n-k'+k-1} \geq e^{-0.5}$.
For
\begin{align} \label{j11}
j_1\triangleq \arg\min_j\{q(\mathcal{C}_j) \leq 1 -(q(n +k - k'))^{-1}\}
\end{align}
we see that
\begin{align} \label{j1}
j_1  \leq n- \log_q(q^{k'-k} V_q(n,\rho)) - \log_q(n+k-k') + O(1)
\end{align}
{\color{black}where the inequality holds by first observing that Lemma~\ref{Lemma-2of5} yields
\begin{align}
    1 - (q(n - k' + k))^{-1} \leq q(\mathcal{C}_{j}) \leq (1 - q^{(n - k' +k) V_q(n,\rho)})^{q^{j_1 -1} e^{-1/2}}\label{lemma-3}
\end{align}
and then by comparing the upper and lower bounds in~\eqref{lemma-3}.}

We now have an $(n,j_1)$ code $\mathcal{C}$ and we have~\eqref{desending}. We are now looking for the minimum number $j_2$ of generators $\mathbf{x}$ that have to be appended to the generator of $\mathcal{C}$ in order to get a $(n,j_1+j_2)$ code with $q(\mathcal{C}_{j_1 + j_2}) \leq q^{-(n-k'+k)}$. We note that $q(\mathcal{C}_{j_1 }) \leq 1 -(q(n-k'+k))^{-1}$, so by~\eqref{lemma-3} we only need to ensure that $(1 - (q(n-k'+k))^{-1})^{2^{j_2}} \leq q^{-(n-k'+k)}$, which can be achieved by using
\begin{align} \label{j2}
    j_2 = 2 \log_2(n-k'+k) + O(1).
\end{align}
Hence for $k =  j_1 + j_2$, there indeed exist $(n,k)$ codes with normalized covering radius no bigger than $\rho$. Applying \eqref{j1}, \eqref{j2}, \eqref{k'-to-L}, and the fact that $|\X| = L q^{k}$, proves~\eqref{Partial-Covering-code} and thus proves Lemma~\ref{CohenLemma}.
\end{enumerate}
\qed

\section{Proof of Theorem \ref{Sparse-Theorem}} \label{E_1}
We quickly note that the converse (lower bound on $\gamma$) holds directly from the converse arguments in Theorem \ref{Achievability}.

Let us start with the following definition.
\begin{Definition}\label{Partial-Covering-Code4}
Let $ \rho\in (0, 1- \frac{1}{q}]$, and let $\tau\in(0,1]$. A code $\mathcal{C}\subseteq \F^{n}$ is said to be a $(\rho,\tau)$-partial covering code if there exists a set $\mathcal{X} \subseteq \F^{n}$, with $\frac{1}{n}\log_q(|\mathcal{X}|) = 1-\tau$, that is $\rho$-covered by $\mathcal{C}$.
\end{Definition}

We now present a theorem that extends the famous Theorem of Blinovskii in~\cite{blinovskii1987lower}, which proved that almost all linear codes satisfy the sphere-covering bound.
We recall that $\mathcal{C}_{k,n}$ denotes the ensemble of all linear codes generated by all possible $k \times n$ matrices in $\F^{k\times n}$.

\begin{Theorem}\label{partial-covering-all}
Let $ \rho\in (0, 1- \frac{1}{q}]$. Then there exists an infinite sequence $k_n$ that satisfies
\begin{align}
    \frac{k_n}{n} \leq 1 - \tau - H_q(\rho) + O(n^{-1} \log_q(n)) \label{partial-covering-almost}
\end{align}
for $\tau\in[0,1-H_q(\rho)-\frac{k}{n}]$ so that the fraction of codes $\mathcal{C}_n \in \mathcal{C}_{k_n,n}$ that are $(\rho,\tau)$-partial covering, tends to 1 as $n$ grows to infinity. Thus in the limit of large $n$, almost all codes of rate less than $1 - \tau - H(\rho)$ will be $(\rho,\tau)$-partial covering.
\end{Theorem}
\begin{proof}
The proof can be found in Appendix~\ref{D}.
\end{proof}

Now let us design such covering codes.
In the following we will consider the set of codes in $\mathcal{C}_{k_n,n}$ that are $(\rho,\tau)$-partial covering, for the claimed sequence $k_n$ of Theorem \ref{partial-covering-all}, and for some real $\tau$. We will also consider $g(n)$ to be the fraction of such $(\rho,\tau)$-partial covering codes among all codes in $\mathcal{C}_{k_n,n}$. The scheme design is defined by the following steps.

\begin{enumerate}
    \item Assign $m=L$.
    \item Set $\tau = \frac{K-\log_q(m)}{N}$.
\item
Noticing that the value
\begin{align}
    m_n \triangleq g(n)q^{k_n n} \label{no-good-codes}
\end{align}
serves as a lower bound on the number of $(\rho,\tau)$-partial covering codes in the ensemble $\mathcal{C}_{k_n,n}$, we now create $\mathcal{B} \triangleq \{\mathcal{C}_1,\mathcal{C}_2, \hdots, \mathcal{C}_{m_n}\}$ to be the set of the first $m_n$ such codes.

Now let
\begin{align}
    \mathbf{D}_{n} \triangleq
\left[ \begin{array}{cccc}
  \mathbf{H}_{\C_1}& & \\
   & \mathbf{H}_{\C_2} & \\
  & & \ddots& \\
 & & & {\mathbf{H}}_{\C_{m_n}}
\end{array}\right]\label{Parity-check-of-sparse-code}
\end{align}
and accordingly set $K=m_n(n-k_n)$, and $N =m_n n$.

Now design $\mathcal{C}_{\W_{n}} =[\mathcal{C}_1,\mathcal{C}_2,\hdots,\mathcal{C}_{m_n}]$, and then create the set
       \begin{align}
           \mathcal{X}_{\mathbf{F}, \mathbf{D}} \triangleq \{\mathbf{x} \in \F^{N}| \mathbf{D} \mathbf{x} = \mathbf{F}(:,\ell), \text{for some}\:  \ell  \in [L]\}.
       \end{align}
       Then create the set
       \begin{align}
       \mathcal{X} \triangleq \{\mathbf{x} = [\mathbf{x}_1,\mathbf{x}_2,\hdots,\mathbf{x}_{m_n}] \ | \  \mathbf{x}_i \in \mathcal{X}_{i}\}\label{C_D_n_covered_set}
       \end{align}
       where $\mathcal{X}_i, i \in [m_n]$, is the set of all $n$-length vectors that are $\rho n$-covered by $\mathcal{C}_i$. Then note that
       \begin{align}
          |\mathcal{X}_i| \geq q^{n(1-\tau)}, \forall i \in [m_n]\label{Lower_Bound}
       \end{align}
       because of Definition~\ref{Partial-Covering-Code4}. We now note that for any $\mathbf{x} \in \mathcal{X}$, it is the case that
       \begin{align}
          \mathbf{d}(\mathbf{x},\mathcal{C})/N =\sum^{m_n}_{i=1}\mathbf{d}(\mathbf{x}_i,\mathcal{C}_i)/N \leq \sum^{m_n}_{i=1} \frac{\rho n}{m_n n}=\sum^{m_n}_{i=1}{\rho}{\frac{1}{m_n}} = \rho
       \end{align}
which means that $\mathcal{C}_{\mathbf{D}_n}$ is also a $(\rho,\mathcal{X})$-partial covering code. Now if $\mathcal{X} \nsupseteq \mathcal{X}_{\mathbf{F}, \mathbf{D}}$, then $m$ has to be increased by one, and the procedure starts again from Step 2.
\item  Let us define $k'_n \triangleq  n - k_n$. From \eqref{partial-covering-almost}, we know that
 \begin{align}
     \frac{k'_n}{n } \geq \tau + H_q(\rho) - O(n^{-1} \log_q(n)).
 \end{align}
We now see that $R\triangleq \frac{K}{N} =\frac{k'_n}{n}$ since $K= k_n' m_n, N= n m_n$. Thus, directly from the above, we have that \begin{align}
     K/N=R = H_q(\rho) + \tau -\epsilon(N).\label{simple-result}
 \end{align}
We note that as $n$ (and thus $N$) goes to infinity, the term $O(n^{-1} \log_q(n))$ vanishes, and thus from the above we have that
 \begin{align}
 \rho = H_q^{-1}(\frac{\log_q(m)}{N} +\epsilon(N)).\label{Sparse-Computational-Cost}
 \end{align}

We also have that
\begin{align}
   \frac{ \omega(\mathbf{D}_{n})}{K} &
 \overset{(a)}
   {\leq} \frac{m_n n k'_n}{m_n k'_n} \overset{}{=}n \label{bound-sparsity}
\end{align}
where (a) holds since $\omega(\mathbf{D}_n) = m_n k_n n$ is the maximum number of nonzero elements that $\mathbf{D}$ can have, due to the block-diagonal design.

After taking the logarithm on both sides of the above, and since $N=m_n n$ and $k_n=(1-R)n$, and after considering~\eqref{no-good-codes}, we have that
\begin{align}
    \log_q(n) + n^2 (1-R) + \log_q(g(n)) = \log_q(N)\label{N-n-relation}
\end{align}
and thus we have that $n^2 (1-R) \leq \log_q(N)$ and $n \leq \sqrt{\frac{\log_q(N)}{(1-R)}}$. Combining this with~\eqref{bound-sparsity} and Theorem \ref{Bridge}, we have that
\begin{align}
    \Delta \leq \sqrt{\frac{\log_q(N)}{(1-R)}} \label{Sparse-Communication-Cost}
\end{align}
where, as mentioned before, $R$ is constant.
\end{enumerate}
We can also see that the above design terminates, since reaching $m=q^K$ implies that $\tau =0$. Then we will have $\mathcal{X}_i = \F^{n}$ since $|\mathcal{X}_i| = q^{n}$ from Definition~\ref{Partial-Covering-Code4}. Therefore from~\eqref{C_D_n_covered_set}, we will have that $\mathcal{X} = \F^{N}=\F^{m_n n}$, which means that $\mathcal{C}_{\mathcal{D}_n}(N,N-K)$ is a $\rho$-covering code, and that $\mathcal{X} \supseteq \mathcal{X}_{\mathbf{F}, \mathbf{D}}$, and thus the scheme would terminate at Step 4 with $ \gamma=\rho = H_q^{-1}(\frac{K}{N} +\epsilon(N))$ from \eqref{Sparse-Computational-Cost}, and with communication cost as shown in~\eqref{Sparse-Communication-Cost}.
\qed

\section{Proof of Theorem \ref{partial-covering-all}} \label{D}

Before offering the formal proof, we provide a quick sketch of the proof to help the reader place the different steps in context.

First we consider the ensemble\footnote{The details about the choice of $k^{*}$ will be described later on.} of codes $\mathcal{C}_{k^{*},n}$, and we prove that with a consistent enumeration of codewords, each nonzero point in $\F^{n}$ has the same chance to be a codeword of a certain index, as we move across the code ensemble.

Second, we pick a code $\mathcal{C} \in \mathcal{C}_{k^{*},n}$ at random, and fix it. Then, based on this code, and for a specific choice of $\tau$ (to be described later on), we introduce a random so-called `covered set' $\mathcal{X}_{\mathcal{C}}$ of size $2^{n(1-\tau)}$ that includes code $\mathcal{C}$.

Then we will see that every point in $\F^{n} \backslash{\mathcal{B}(\mathbf{0},\rho)}$ has an equal probability --- as we go through the choices of $\mathcal{C} \in \mathcal{C}_{k^{*},n}$ --- of belonging to this subset.
To analyze the $\rho$-coverage of points inside $\mathcal{X}_{\C}$, we derive $\mathbb{P}{(\mathbf{c}_i = \mathbf{x}|\mathbf{x} \in \X_{\C})}$, where $\mathbf{c}_i$ describes the codeword indexed by a fixed $i$, as we move across the codes (and the corresponding generator matrices) in the ensemble.

Toward showing that $\mathcal{X}_{\mathcal{C}}$ is covered by $\mathcal{C}$, we first note that $\mathcal{B}(\mathbf{0},\rho)$ is covered since $\mathbf{0} \in \mathcal{C}$. To prove that the remaining part, $\mathcal{X}_{\C} \backslash{\mathcal{B}(\mathbf{0},\rho)}$, is also covered by $\mathcal{C}$, we prove that if codes in the ensemble are sufficiently large, then there is, for almost all codes $\mathcal{C}$, a large number (polynomial in $n$) of codewords that covers each specific point in $\mathbf{x} \in \X_{\C}$. With this in place, we will be able to conclude that almost all codes come close to being $(\rho,\tau)$-partial covering.

Finally, we utilize a linear greedy algorithm and successive appending of a very small number of $\lfloor \log_q{n(1-\tau)} \rfloor$ carefully selected vectors (cf. Lemma~\ref{linear-greedy}) to each of these almost $(\rho,\tau)$-partial covering codes, to render them fully $(\rho,\tau)$-partial covering codes.

We proceed with the formal proof.

Let $k^{*}\triangleq k - \lceil \log_q(n(1-\tau)) \rceil,\: k,n \in \mathbb{N}, 0 \leq  \tau \leq 1$ and let $\mathcal{C}_{k^{*},n}$ be the ensemble of codes generated by $k^{*} \times n$ generator matrices whose elements are chosen randomly and independently with probability $\frac{1}{q}$ from $\F_q$. Naturally, any fixed non-zero linear combination of rows of the generator matrix, will generate --- as we move across the ensemble of generator matrices --- all possible $q^n$ vectors in $\mathbb{F}^n$. The zero codeword corresponds to the void linear combination of rows, and is present in all generated codes. Also let us assume a consistent enumeration of the codewords, in the sense that a word's index is defined by the linear combination of rows of the generator matrix, that generate that codeword, in each code. For example, the word indexed by $5$, will vary in value across the different codes, but it will always be defined as the output of a specific (the fifth) linear combination of the corresponding generator matrix. The first codeword in all codes will be the zero word.
We proceed with the following lemma.
 \begin{Lemma}\label{uniformity}
 For any fixed $i\in[2:2^{k^{*}}]$, and for any fixed $\mathbf{x} \in \mathbb{F}^n$, then
\begin{align}
    \mathbb{P}(\mathbf{c}_i = \mathbf{x})=q^{-n} \label{claim-1}
\end{align}
where the probability is over all codes $\mathcal{C} \in \C_{k^{*},n}$.
\end{Lemma}
\begin{proof}
The proof is presented in Appendix~\ref{H}.
\end{proof}

Let us set $\tau\in[0,1-H_q(\rho) - k^{*}/n]$, and let us note that for sufficiently\footnote{We quickly remind the reader that our results here will hold for sufficiently large $n$.} large $n$, we can guarantee that 
\begin{align}\label{guaranteeCardinality}
q^{n (1- \tau) } & \geq V_q(n, \rho) + q^{k^{*}}.
\end{align}

Let us now go over the ensemble of codes $\mathcal{C} \in \C_{k^{*},n}$, and for each code, let us create the covered set $\X_{\C}$ such that
\begin{align}
    |\X_{\mathcal{C}}| &= q^{n(1-\tau)}\\
    \mathcal{C} \cup \mathcal{B}(\mathbf{0},\rho) &\subseteq \mathcal{X}_{C} \subseteq \mathbb{F}^{n}.\label{Sufficient-Condition-For-Choice}
\end{align}
{\color{black}We can see that~\eqref{guaranteeCardinality} is a necessary condition for the above to happen.} The procedure for designing $\mathcal{X}_{\C}$, simply starts by taking the union $\mathcal{C} \cup \mathcal{B}(\mathbf{0},\rho)$, and then proceeds by appending on this union, a sufficiently large number of vectors, chosen uniformly and independently at random from $\F^{n} \backslash{{ \mathcal{C}  \cup \mathcal{B}(0,\rho)} }$. 
The following lemma simply says that every point $\mathbf{x}\in \F^{n} \backslash{\mathcal{B}(\mathbf{0},\rho)}$ has an equal probability --- as we go through the choices of $\mathcal{C} \in \mathcal{C}_{k^{*},n}$ --- of belonging to this subset $\X_{\mathcal{C}}$.
 \begin{Lemma}\label{Covered-set-point}
For any fixed $\mathbf{x}\in \F^{n}$, then 
\begin{align}
      \mathbb{P}(\mathbf{x} \in \mathcal{X}_{\mathbf{C}}) =
      \begin{cases}
       1 \:\: \: &  \omega(\mathbf{x})\leq \rho n\\
       \frac{q^{n(1-\tau)} - V(\rho,n)}{q^{n} - V(\rho,n)} \:\:\: &  \omega(\mathbf{x})> \rho n.
\end{cases}\label{Claim-of-Lemma}
\end{align}
\end{Lemma}
\begin{proof}
The proof is presented in Appendix \ref{D_1}.
\end{proof}

With the above lemma in place, we will now calculate the following conditional probability. The following asks us to first pick and fix a vector $\mathbf{x} \in \F^{n}$, and then pick an index $i \in [1:2^{k^*}]$. Recall --- from the above discussion on the consistent enumeration of codewords --- that this index will define a codeword $\mathbf{c}_i$, which changes as we go across all the codes $\mathcal{C}$ in the ensemble $\mathcal{C}_{k^{*},n}$.
The following conditional probability is again calculated over the code ensemble. 

\begin{Lemma}\label{Point-Wise-Conditional_Probability}
Pick any vector $\mathbf{x} \in \F^{n}$ and any index $i \in [1:2^{k^*}]$. Then 
\begin{align}
      \mathbb{P}(\mathbf{c}_i = \mathbf{x}|\mathbf{x} \in \mathcal{X}_{\C}) =
      \begin{cases}
       0 \:\: \: & i=1 , \mathbf{x} \neq 0\\
       1 \:\: \: & i=1 , \mathbf{x}=0\\
    q^{-(n)} \:\: \: & i \in [2: K^{*}] , \mathbf{x}\neq  0,\omega(\mathbf{x}) \leq \rho n\\
        q^{-n(1-\tau)}\zeta(n) \:\:\: &  i \in [2: K^{*}] \text{ and } \omega(\mathbf{x}) > \rho n
\end{cases}
\label{probability-1}
\end{align}
where the term $\zeta(n)$ converges to $1$ as $n$ approaches to infinity.
\end{Lemma}
\begin{proof}
The proof is presented in Appendix~\ref{D_2}.
\end{proof}
Let us now discuss the coverage of any vector $\mathbf{x} \in \X_{\mathcal{C}}$, as we go along the ensemble $\mathcal{C} \in \C_{k^{*},n}$. 
First of all, it is clear that any $\mathbf{x} \in \mathcal{B}(0,\rho)$ is both $\rho$-covered by the code $\mathcal{C}$ (because $\mathbf{0} \in \mathcal{C}$) as well as is included in $\X_{\mathcal{C}}$ (because $\mathbf{x} \in \mathcal{B}(0,\rho) \subset \X_{\mathcal{C}}$). Thus for each code $\mathcal{C} \in \C_{k^{*},n}$, for the purposes of the current proof, we can focus on the set 
\begin{align}
    \mathcal{X}'_{\C} \triangleq \mathcal{X}_{\mathcal{C}} \backslash{ \mathcal{B}(\mathbf{0}, \rho)}.
\end{align}

For every $\mathbf{x} \in \X'_{\C} $, let us define the random variable $\eta_{\mathbf{x},i}$ which takes the value $1$ if $\mathbf{c}_i$ $\rho n$-covers $\mathbf{x}$, and which takes the value $0$ otherwise. Thus 
\begin{align}
    \eta_{\mathbf{x}} \triangleq \sum^{2^{k^{*}}}_{i=1} \eta_{\mathbf{x},i} \label{Definition-eta-x}
\end{align}
describes the number of codewords that cover $\mathbf{x} \in \mathcal{X}'_{\C}$. For any fixed $\mathbf{x}\in \F^n$, the following lemma describes the conditional average $\eta_{\mathbf{x}}$, where again the average is taken over the code ensemble.
\begin{Lemma}\label{Average-eta-x}
For any fixed $\mathbf{x}\in \F^n$, then 
\begin{align}
      \mathbb{E}{(\eta_{\mathbf{x}}|\mathbf{x} \in \mathcal{X}'_{\C})} &=|\{0\} \cap \mathcal{B}(\mathbf{x},\rho)|\times 1\\&+(q^{k^{*}}-1)[|(\mathcal{B}(0,\rho)\backslash{\{0\}}) \cap \mathcal{B}(\mathbf{x},\rho)| q^{-n}\\&+|\mathcal{B}(\mathbf{x},\rho) \backslash{ \mathcal{B}(0,\rho)}| q^{-n(1-\tau)} \zeta(n)] \label{Average-number-covering}
\end{align}
where again $\zeta(n)\rightarrow 1$ as $n$ increases.
\end{Lemma}
\begin{proof}
The proof is in Appendix~\ref{D_3}, and it involves an extension of Blinovskii's Theorem~\cite{blinovskii1987lower}, from evaluating $\mathbb{E}{(\eta_{\mathbf{x}})}$ to evaluating the conditional $\mathbb{E}{(\eta_{\mathbf{x}}|\mathbf{x} \in \mathcal{X}'_{\C})}$.
\end{proof}
Before proceeding, we need the following lemma which is an extension of a related lemma found in~\cite{blinovskii1987lower}. The following considers as before the set $\mathcal{X}'_{\C} \triangleq \mathcal{X}_{\mathcal{C}} \backslash{ \mathcal{B}(\mathbf{0}, \rho)}$, and considers again the variance and expectation, over the aforementioned code ensemble. 

\begin{Lemma}\label{Bounding_l_lemma}
For any fixed $\mathbf{x}\in \F^n$, then 
\begin{align}
    \frac{Var(\eta_{\mathbf{x}}|\mathbf{x} \in \mathcal{X}'_{\C})}{\mathbb{E}(\eta_{\mathbf{x}}|\mathbf{x} \in \mathcal{X}'_{\C}) q^2} \leq 1. \label{bounding-variance}
\end{align}
\end{Lemma}

\begin{proof}
The proof is presented in Appendix \ref{D_5}.
\end{proof}
Combining~\eqref{bounding-variance} with Chebyshev's inequality, gives 
\begin{align}
    \mathbb{P}(|\eta_{\mathbf{x}} - \mathbb{E}(\eta_{\mathbf{x}}|\mathbf{x} \in \mathcal{X}'_{\C})| > q^{\epsilon+1} \sqrt{\mathbb{E}(\eta_{\mathbf{x}}|\mathbf{x} \in \mathcal{X}'_{\C})}\Big|\mathbf{x} \in \mathcal{X}'_{\C}) < \frac{Var(\eta_
    {\mathbf{x}}|\mathbf{x} \in \mathcal{X}'_{\C})}{ q^{2 \epsilon +2} \mathbb{E}(\eta_{\mathbf{x}}|
    \mathbf{x} \in \mathcal{X}'_{\C})} \leq q^{-2\epsilon}\label{chebyshev-1}.
\end{align}
{\color{black}Our aim is to show that, for any $\mathbf{x} \in \mathcal{X}'_{\C}$,  --- under some conditions on $k^*$ and $\epsilon$ --- $\eta_{\mathbf{x}}$ will be, with high probability, bigger than $0$ which in turn implies that any $\mathbf{x} \in \mathcal{X}'_{\mathcal{C}}$ will, with high probability, be covered by $\mathcal{C}$.} 
To see this, we continue from \eqref{chebyshev-1}. We first see from~\eqref{Average-number-covering} that $\mathbb{E}(\eta_{\mathbf{x}}|\mathbf{x} \in \mathcal{X}'_{\C}) >0 $. Hence whenever we have $\eta_{\mathbf{x}} > \mathbb{E}(\eta_{\mathbf{x}}|\mathbf{x} \in \mathcal{X}'_{\C})$, we also have that $\eta_{\mathbf{x}} > 0$. Let us now focus on the remaining scenario where $\eta_{\mathbf{x}} \leq \mathbb{E}(\eta_{\mathbf{x}}|\mathbf{x} \in \mathcal{X}'_{\C})$. In this case, from~\eqref{chebyshev-1}, we have that
\begin{align}
\mathbb{P}\bigl( \eta_{\mathbf{x}} \geq \mathbb{E}(\eta_{\mathbf{x}}|\mathbf{x} \in \mathcal{X}'_{\C}) - q^{\epsilon+1} \sqrt{\mathbb{E}(\eta_{\mathbf{x}} \ | \ \mathbf{x} \in \mathcal{X}'_{\C})}\bigr |\mathbf{x} \in \mathcal{X}'_{\C})) \geq 1-q^{-2\epsilon}. \label{Reverse-Chebyshev}
\end{align}
We also have that 
\begin{align}
\mathbb{P}\bigl( \eta_{\mathbf{x}} > 0 \ | \ \mathbf{x} \in \mathcal{X}'_{\C}\bigr) \geq 
\mathbb{P}\bigl( \eta_{\mathbf{x}} \geq \mathbb{E}(\eta_{\mathbf{x}}|\mathbf{x} \in \mathcal{X}'_{\C}) - q^{\epsilon+1} \sqrt{\mathbb{E}(\eta_{\mathbf{x}} \ | \ \mathbf{x} \in \mathcal{X}'_{\C})} \ | \ \mathbf{x} \in \mathcal{X}'_{\C}\bigr)\label{Reverse-Chebyshev-2}
\end{align}
{\color{black}under the assumption that 
\begin{align} \label{betaCondition}
\beta(\epsilon)\triangleq  \mathbb{E}(\eta_{\mathbf{x}}|\mathbf{x} \in \mathcal{X}'_{\C}) - q^{\epsilon+1} \sqrt{\mathbb{E}(\eta_{\mathbf{x}}|\mathbf{x} \in \mathcal{X}'_{\C})} >0 .
\end{align}
This assumption will be guaranteed --- as we will see later on --- by a proper choice of $k^*$ and $\epsilon$.}

Now combining \eqref{Reverse-Chebyshev} with  \eqref{Reverse-Chebyshev-2}, we will show that 
\begin{align}
\mathbb{P}\bigl( \eta_{\mathbf{x}}  \geq \mathbb{E}(\eta_{\mathbf{x}}|\mathbf{x} \in \mathcal{X}'_{\C}) - q^{\epsilon+1} \sqrt{\mathbb{E}(\eta_{\mathbf{x}} \ | \ \mathbf{x} \in \mathcal{X}'_{\C})} \ | \ \mathbf{x} \in \mathcal{X}'_{\C}\bigr) \rightarrow 1
\end{align}
as $n$ grows to infinity. 

To guarantee that $\beta(\epsilon) >0$, we must guarantee that 
\begin{align}
    \mathbb{E}(\eta_{\mathbf{x}}| \mathbf{x} \in \X'_\C)>q^{{2\epsilon +2}}. \label{being-positive}
\end{align}
To do this, given Lemma~\ref{Average-eta-x}, we must prove that
\begin{align}
    |\{0\} \cap \mathcal{B}(\mathbf{x},\rho)|\times 1&+(q^{k^{*}}-1)[|(\mathcal{B}(0,\rho)\backslash{\{0\}}) \cap \mathcal{B}(\mathbf{x},\rho)| q^{-n}\\&+|\mathcal{B}(\mathbf{x},\rho) \backslash{ \mathcal{B}(0,\rho)}| q^{-n(1-\tau)} \zeta(n) ]>q^{2 \epsilon +2}
\end{align}
again for some properly chosen $k^{*}$ and $\epsilon$. The following applies toward this effort.
\begin{Lemma} \label{lower-bound-lemma}
{\color{black} For any $\mathcal{C} \in \mathcal{C}_{k^{*},n}$}, any $\mathbf{x} \in \mathcal{X}'_{\C}$, and any $\rho \in (0,\min\{1-1/q,\frac{\sqrt{5}-1}{2}\}]$, then \begin{align}
    |\mathcal{B}(\mathbf{x},\rho) \backslash{ \mathcal{B}(0,\rho)}|> q^{n H_{q}(\rho)-o(n)}.\label{lower-bound-set-minus}
\end{align}
\end{Lemma}
\begin{proof}
The proof is in Appendix~\ref{M}.
\end{proof}
We now combine~\eqref{Average-number-covering} and~\eqref{lower-bound-set-minus} to get 
\begin{align}
   E(\eta_{\mathbf{x}}| \mathbf{x} \in \X'_\C) > (q^{k^{*}} - 1) q^{n H_q(\rho) - o(n) - n(1-\tau) } \zeta(n) \label{lower-bound}
\end{align}
and we also choose $k^{*},\epsilon,$ to guarantee (cf.~\eqref{being-positive}) that the inequality
\begin{align}
     (q^{k^{*}} - 1) q^{n H_q(\rho) - o(n) - n(1-\tau) } {\color{black}\zeta(n)} \geq q^{2 \epsilon + 2}\label{lower-bound-on-E}
\end{align}
holds for large $n$. Thus with \eqref{lower-bound} and \eqref{lower-bound-on-E} in place --- something that will indeed be validated by the end of the proof (cf.~\eqref{Validation-ends}) --- we can guarantee~\eqref{being-positive}. Thus we know that 
\begin{align}
    \mathbb{P}( \eta_{\mathbf{x}} < n^{\alpha} \big | \mathbf{x} \in \mathcal{X}'_{\C})
      < q^{-2 \epsilon}. \label{bound-for-n-alpha}
\end{align}

Following the approach in~\cite{blinovskii1987lower}, we consider points in $\mathcal{X}'_{\C}$ that are called `partial-remote points', which are the points that are $\rho n $-covered by fewer than $n^{\alpha},\alpha>1$ codewords. Now let $Q_{0}(\mathcal{X}'_{\C}) \subset \mathcal{X}'_{\C}$ be the set of partial remote points in $\mathcal{X}'_{\C}$, and let
 \begin{align}
     q_{0}(\mathcal{X}'_{\C}) \triangleq \frac{|\mathcal{Q}_{0}(\mathcal{X}'_{\C})|}{q^{n(1 - \tau)} - V_q(n,\rho)}. \label{def-1-bad-cases}
 \end{align}
Now applying~\eqref{bound-for-n-alpha}, gives
 \begin{align}
 \sum_{\mathbf{x} \in \mathcal{X}'_{\C}}      \mathbb{P}( \eta_{\mathbf{x}} < n^{\alpha} \big | \mathbf{x} \in \mathcal{X}'_{\C}) \leq (q^{n(1-\tau)} -V_q(n,\rho))q^{-2\epsilon}\label{upperbound-results}
 \end{align}
and thus we see that 
 \begin{align}
 \sum_{\mathbf{x} \in \mathcal{X}'_{\C}}      \mathbb{P}( \eta_{\mathbf{x}} < n^{\alpha} \big | \mathbf{x} \in \mathcal{X}'_{\C}) & \overset{(a)}{=}   \sum_{\mathbf{x} \in \mathcal{X}'_{\C}}      \mathbb{E} [\mathbbm{1}( \eta_{\mathbf{x}} < n^{\alpha} \big | \mathbf{x} \in \mathcal{X}'_{\C}) ]\\
 & \overset{(b)}{=}  \mathbb{E}[ \sum_{\mathbf{x} \in \mathcal{X}'_{\C}}      \mathbbm{1}( \eta_{\mathbf{x}} < n^{\alpha} \big | \mathbf{x} \in \mathcal{X}'_{\C}) ]\\
 &\overset{(c)}{=}\mathbb{E} [|\mathcal{Q}_0(\mathcal{X}'_{\C})|]\label{relationship-to-bad-cases}
 \end{align}
where now the expectation in (a) is over the codes in the ensemble $\mathcal{C}_{k^{*},n}$, where (b) results from interchanging the expectation with the summation, and where (c) is by definition of $\mathcal{Q}_{0}$.

 Now combining \eqref{def-1-bad-cases}, \eqref{upperbound-results} and \eqref{relationship-to-bad-cases}, we have
 \begin{align}
     \mathbb{E}(q_0)  \leq q^{-2 \epsilon}
 \end{align}
which bounds the average (over the code ensemble) number of partial-remote points in $\X'_{\C}$. 
Then Markov's inequality directly tells us that 
 \begin{align}
     \mathbb{P}(q_0 > q^{\epsilon} \mathbb{E}(q_{0})) < q^{- \epsilon} \label{good-codes}
 \end{align}
which means that the expression $q_0 \leq  q^{\epsilon} \mathbb{E}(q_{0})$ holds for a proportion greater than $1-q^{-\epsilon}$ of all codes.

{\color{black} Now, in the footsteps of~\cite{cohen1997covering}, we apply a procedure that successively appends cosets to an initial code $\mathcal{C}' \in \mathcal{C}_{k^{*},n}$ that belongs in this above family of codes that indeed satisfies  $q_0 \leq  q^{\epsilon} \mathbb{E}(q_{0})$.} 
{\color{black} Let us quickly remember that the optimal successive appending linear greedy method resulting from Lemma~\ref{linear-greedy} and enclosed in Appendix~\ref{A_2}, allowed us to prove \eqref{bounding-union-begin}--\eqref{bounding-union} which yielded 
\begin{align}
     q(\mathcal{C}_{j+1}) \leq q(\mathcal{C}_{j})^2 \label{desending-2}
\end{align}
where $\mathcal{C}_{j+1} = <\mathcal{C}_{j}; \mathbf{x}>$, and where $q(\mathcal{C}_{j})$ represented the number of remote points of  the code $\mathcal{C}_j$, normalized by $q^{n(1 - \tau)} - V_q(n,\rho)$.

With the above in mind, let us now set this first initializing code $\mathcal{C}_{0}$ to be equal to $\mathcal{C}_{0} = \mathcal{C}'$, where $\mathcal{C}'$ is one of the aforementioned `good' codes that satisfy 
\begin{align}\label{good-codes-1}
q_0 \leq  q^{\epsilon} \mathbb{E}(q_{0}).
\end{align}
Then we will design $\mathcal{C}_{1} = <\mathcal{C}_{0}; \mathbf{x}>$ where $\mathbf{x}$ is a guaranteed-to-exist vector (cf.~\eqref{desending-2}) that increases the span of $\mathcal{C}_{0}$.

Now we calculate the same quantity we calculated in \eqref{def-1-bad-cases}, but we do so for $\mathcal{X}'_{\C_1}$. In other words, we calculate 
 \begin{align}
     q_{1}(\mathcal{X}'_{\C_1}) \triangleq \frac{|\mathcal{Q}_{0}(\mathcal{X}'_{\C_1})|}{q^{n(1 - \tau)} - V_q(n,\rho)} \label{def-1-bad-cases2}
 \end{align}
 where similar to before, now $q_1$ represents the average normalized remote points of the code $\mathcal{C}_1$ with respect to its associated $\mathcal{X}_{\C_1}$. We can now see that directly from~\eqref{desending-2}, we have that
\begin{align}
    \mathbb{E}(q_1) \leq q_0^2
\end{align}
where the above average is taken over all $\mathcal{C}_1$ codes, meaning over all codes that can take the role of our aforementioned $\mathcal{C}_1$, going over all possible initializing codes $\mathcal{C}_0 = \mathcal{C'}$, and over all possible base-expanding vectors $\mathbf{x}$. 
Now we apply again Markov's inequality, this time over the expanded codes $\mathcal{C}_1$, to get 
\begin{align}
    \mathbb{P}(q_1 > q^{\lambda} \mathbb{E}(q_1)) < q ^{-\lambda},\: \lambda \in \mathbb{R}
\end{align}
which tells us --- as before --- that the proportion of codes $\C_1$ that satisfy
\begin{align}
    q_1  \leq  q ^{\lambda- 2\epsilon} \label{good-codes-2}
\end{align}
is at least $1 - q ^{- \lambda}$. This proportion of codes that achieve~\eqref{good-codes-2}, is over all generated $\C_1$ that were built over all `good' $\C_0 = \C'$ that already satisfied \eqref{good-codes-1}. Thus we now know (cf.~\eqref{good-codes}) that the proportion of codes $\C_1$ --- among all codes in the entire ensemble $\C_{{k^{*}+1},n}$ --- that satisfy~\eqref{good-codes-2}, is equal to $(1-q^{- \epsilon})(1 - q^{-\lambda})$.

We now go from our step 1, to an arbitrary step $i$, and following the same logic as before, we conclude that the proportion of codes $\C_i$ --- where this proportion is among all codes in the entire ensemble $\C_{{k^{*}+i},n}$ --- that satisfy 
\begin{align}
    q_i < q^{-2^i(\epsilon - \lambda) - \lambda} \label{the-ultimate-pperbound-remote}
\end{align}
is equal to $(1 - q^{-\epsilon})(1 - q^{-\lambda})^{i}$.

Now let us go to some step $m$ which will allow us to terminate. We explain when will this termination happen.
Consider, for this step $i=m$, as before, the quantity
\begin{align}
     q_{m}(\mathcal{X}'_{\C_m}) \triangleq \frac{|\mathcal{Q}_{0}(\mathcal{X}'_{\C_m})|}{q^{n(1 - \tau)} - V_q(n,\rho)} \label{def-1-bad-casesm}
 \end{align}
 where similar to before, now $q_m$ represents the average normalized remote points of the code $\mathcal{C}_m$ with respect to its associated $\mathcal{X}_{\C_m}$.
 We want the corresponding  $\mathcal{Q}_{m}$ (cf.~\eqref{def-1-bad-cases}) to be empty. This will be guaranteed when $m$ is such that 
 \begin{align}
    q_m < q^{-n(1 - \tau)}\label{covering-condition}
\end{align}
where the above guarantee can be provided given that the cardinality of a set is a non-negative integer.  

{\color{black} This will be achieved by setting $m = \lceil \log_2(n(1-\tau))\rceil$ and $\lambda = \epsilon - 1$. This can be indeed verified by considering \eqref{the-ultimate-pperbound-remote} after setting $i=m, \lambda = \epsilon - 1$.} In conclusion, the proportion of `good' codes (among the entire ensemble) designed at this stage $m$, is no less than 
\begin{align}
    (1-q^{-\epsilon})(1- q^{- \epsilon +1})^{\lceil \log_2n(1 - \tau)\rceil}\label{good-proportion}
\end{align}
and for each such code $\C$, every point $\mathbf{x} \in \mathcal{X}_{\C}$ will be $\rho n$-covered by at least $n^{\alpha}$ codewords of that same code.  

With the above in place, let us return to \eqref{betaCondition} where we wish to guarantee that $\beta(\epsilon) >0 $. Toward this, let us consider~\eqref{good-proportion} and in this equation, let us set $\epsilon = 2 \log_q \log_2 (n(1 - \tau))$. 

Let us now prove that there exists an $\alpha>1$ such that 
\begin{align} \label{poly-substitution2}
    \mathbb{E}(\eta_{\mathbf{x}}| \mathbf{x} \in \X'_\C) \geq( n (1-\tau) )^{\alpha} =  (q^{k^{*}}-1) q^{n H_q(\rho) - o(n)} q^{-(n-n\tau)}.
\end{align}
This can be readily shown (cf.~\eqref{lower-bound}) by noting that, for large $n$, then the expression
\begin{align}
    \exists\:\alpha>1 : (n-n\tau)^{\alpha} =  (q^{k^{*}}-1) q^{n H_q(\rho) - o(n)} q^{-(n-n\tau)}\label{poly-substitution}
\end{align}
holds. With~\eqref{poly-substitution2} in place, we take the logarithm on both sides of the above, and after dividing by $n$, we get
\begin{align}
    \frac{k^{*}}{n} = 1 - \tau-H_q(\rho) + \frac{\alpha \log_q(n-n\tau) + o(n)}{n} \label{result-1}.
\end{align}
Let us now recall that we had conditionally accepted~\eqref{lower-bound-on-E}, by saying that~\eqref{lower-bound-on-E} holds for some properly chosen $\epsilon$ and $k^{*}$. We will use the aforementioned $\epsilon = 2 \log_q \log_2 (n(1 - \tau))$ and the $k^{*}$ from \eqref{result-1}. Let us apply these values in the LHS of~\eqref{lower-bound-on-E}, and note, after employing \eqref{poly-substitution}, that for these values in place, it holds that 
\[     (q^{k^{*}} - 1) q^{n H_q(\rho) - o(n) - n(1-\tau) } {\color{black}\zeta(n)} = (n - n \tau)^{\alpha} \zeta(n).
\]
Let us now note that for sufficiently large $n$ then
\begin{align}
    (n - n \tau)^{\alpha} \zeta(n)& \geq  q^2 \log_2(n(1-\tau))^4\label{Validation-ends}
\end{align}
simply because the RHS is logarithmic in $n$. At the same time though, we also note that 
$ q^2 \log_2(n(1-\tau))^4 = q^{2(2\log_q(\log_2(n(1-\tau)))) + 2}$ and then, by applying the chosen $\epsilon$, we get that 
$ q^2 \log_2(n(1-\tau))^4 = q^{2 \epsilon +2}$. Thus we now know that $(n - n \tau)^{\alpha} \zeta(n) \geq q^{2 \epsilon +2}$ which, after applying \eqref{Validation-ends}, gives that
\[     (q^{k^{*}} - 1) q^{n H_q(\rho) - o(n) - n(1-\tau) } {\color{black}\zeta(n)} \geq q^{2 \epsilon + 2}
\]
which is exactly~\eqref{lower-bound-on-E}. Thus~\eqref{lower-bound-on-E} is validated, and consequently, directly, we can also guarantee~\eqref{being-positive}, which in turn guarantees $\beta(\epsilon) \geq 0$, which in turn proves that~\eqref{Reverse-Chebyshev-2} indeed holds.

Following the logic immediately before~\eqref{Reverse-Chebyshev-2}, and with~\eqref{Reverse-Chebyshev-2} now in place, we can conclude that for any $\mathbf{x} \in \mathcal{X}'_{\C}$,  --- given our chosen $k^*$ and $\epsilon$ --- $\eta_{\mathbf{x}}$ will be, with high probability, bigger than $0$ which in turn implies that any $\mathbf{x} \in \mathcal{X}'_{\mathcal{C}}$ will, with high probability, be covered by $\mathcal{C}$. As it can be seen, the proportion  of partial-covering codes \eqref{good-proportion} approaches 1, as $n$ increases.

The only thing that remains now to be verified is the rate $k_n/n$ of these codes, which was declared in the theorem to be as in~\eqref{partial-covering-almost}. 
To verify that indeed this is our rate, we recall that we have started with a code from $\C_{k^{*},n}$ and that we have performed the appending procedure $m$ times, where we chose $m = \log_2(n-n\tau)$. This means that the current message length becomes 
\begin{align*}
    {k_n} = {k^{*} +\log_2(n(1-\tau))}.
\end{align*}
Now directly adding $m/n$ on both sides of the equation in~\eqref{result-1}, we have that 
\begin{align}
    \frac{k_n}{n}=\frac{k^{*} +\log_2(n-n\tau) }{n} = 1 - \tau -H(\rho) + \frac{(\alpha +1)\log_2(n-n\tau) + o(n)}{n}\label{result-2}
\end{align}
which simply says that $ \frac{k_n}{n} \leq 1 - \tau - H_q(\rho) + O(n^{-1} \log_q(n)) 
$ as in~\eqref{partial-covering-almost}. This concludes the proof of the theorem. \qed

}

\section{Proof of Proposition \ref{Better-Achievability-Sparse}}\label{F}
Referring to the proof of Theorem~\ref{Sparse-Theorem} in Appendix~\ref{E_1}, let us suppose that $L \leq m<q^{K}$ and that $\mathcal{X} \supseteq \mathcal{X}_{\mathbf{F}, \mathbf{D}}$. From \eqref{C_D_n_covered_set} we see that
\begin{align}
    |\mathcal{X}|&\overset{(a)}{=}\Pi^{m_n}_{i=1}|\mathcal{X}_i|\\
    &\overset{(b)}{\geq}  q^{n m_n(1-\tau)}\\
      &\overset{(c)}{=} q^{N(1-\tau)}\label{bounding-cardinality}
\end{align}
where (a) comes from the definition of $\mathcal{X}$ (cf. \eqref{C_D_n_covered_set}), where (b) holds due to \eqref{Lower_Bound}, and where (c) holds since $N= n m_n$. Then from~\eqref{simple-result} and~\eqref{bounding-cardinality}, we can conclude that
\begin{align}
    \rho &= H_q^{-1}(\frac{K}{N} - \tau + \epsilon(N))\\
    &\leq H_q^{-1}(\frac{K}{N} - (1 -\frac{\log_q(|\mathcal{X}|)}{N}) + \epsilon(N)).
\end{align}
Setting $m=L$ gives $\rho =H^{-1}_q( \frac{\log_q(L)}{N} + \epsilon(N))$ (cf.~\eqref{Sparse-Computational-Cost}). The communication cost is as described in~\eqref{Sparse-Communication-Cost}. \qed

\section{Various Proofs}{}

\subsection{Proof of Lemma \ref{uniformity}}\label{H}
{\color{black}  
For a fixed index $i \neq 0$, there is a fixed information vector $\mathbf{d}_i \in \F^{k^{*}}\backslash {\mathbf{0}}$ that generates --- as we move across the generator matrices $\mathbf{G}$ in the ensemble of codes $\mathcal{C}_{k,n}$ --- the codewords $\mathbf{c}_i$ that take the form
\begin{align}
\mathbf{c}_i= \mathbf{d}_i \mathbf{G} = \sum^{n}_{j=1} d_i(j,1) \mathbf{G}(j,:).  
 \end{align}
Given that the elements of $\mathbf{G} \in \F^{k^{*} \times n}$ are chosen uniformly and independently at random from $\F$, directly implies that the same holds for the elements of $\mathbf{c}_i$, since in the above linear combination, the elements of $\mathbf{d}_i$ are fixed, and naturally because the operations are over a \emph{finite} field. 
 
\qed
}

\subsection{Proof of Lemma \ref{Covered-set-point}}{}\label{D_1}
We can first see that whenever $\omega(\x) \leq \rho n $, then~\eqref{Claim-of-Lemma} automatically holds simply because such $\mathbf{x}$ must belong in $\mathcal{B}(0,\rho)$ which in turn is a subset of $\mathcal{X}_{\C}$.

{\color{black}
Let us now consider the case of $\omega(\x) > \rho n, \x \in \F^{n}$. Let us also recall the element selection process\footnote{We recall that the procedure for designing $\mathcal{X}_{\C}$, simply starts by taking the union $\mathcal{C} \cup \mathcal{B}(\mathbf{0},\rho)$, and then appending on this union, a sufficiently large number of vectors, chosen uniformly and independently at random from $\F^{n} \backslash{{ (\mathcal{C}  \cup \mathcal{B}(0,\rho))} }$.} that was described right underneath equation~\eqref{Sufficient-Condition-For-Choice}. Let $\mathcal{S} \subset \F^n$ be the set of all vectors $\x$ that are not codewords but are selected randomly in the aforementioned process. At this point, we can see that 
\begin{align}
    \mathbb{P}(\x \in \mathcal{X}_{\C}) = \mathbb{P(\x \in \mathcal{C})} + \mathbb{P}{(\x \notin {\C})}\mathbb{P}(\x\in \mathcal{S} \ | \ \mathbf{x} \notin \mathcal{C})
\end{align}
where $\C$ is the code that has been chosen uniformly at random from $\C_{k^{*},n}$. Consider a vector $\mathbf{y} \in \F^{n}$ with $\omega(\mathbf{y}) \geq \rho n$. We clearly see that $\mathbb{P}{(\mathbf{x} \in {\C})} = \mathbb{P}{(\mathbf{y} \in {\C})}$, and we also see that 
$\mathbb{P}(\x\in \mathcal{S} \ | \ \mathbf{x} \notin \mathcal{C}) =\mathbb{P}(\y\in \mathcal{S} \ | \ \mathbf{y} \notin \mathcal{C})$ as a direct outcome of the aforementioned vector selection process, and of the fact that $\omega(\x)> \rho n, \omega(\y) > \rho n $, which yields
\begin{align}
 \mathbb{P}{(\x \in \mathcal{X}_{\C})} = \mathbb{P}{(\y \in \mathcal{X}_{\C})}. \label{equiprobable}
\end{align}
}
Now let us note that 
\begin{align}
   \sum_{\mathbf{y} \in \F^{n} \backslash{\mathcal{B}(0,\rho)}}
   \mathbb{P}(\mathbf{y} \in \mathcal{X}_{\C}) & =
   \sum_{\mathbf{y} \in \F^{n} \backslash{\mathcal{B}(0,\rho)}}\mathbb{E}[{\mathbbm{1}(\mathbf{y} \in \mathcal{X}_{\C})}] \label{results1}\\
 &\overset{(a)}{=}  \mathbb{E}[{\sum_{\mathbf{y} \in \F^{n} \backslash{\mathcal{B}(0,\rho)}}\mathbbm{1}(\mathbf{y} \in \mathcal{X}_{\C})}] \\
   &\overset{(b)}{=} \mathbb{E}[q^{n(1-\tau)} - V_q(n,\rho)]
   =q^{n(1-\tau)} - V_q(n,\rho)\label{summation-on-probabilities}
\end{align}
where the average is over the codes in the ensemble $\mathcal{C}_{k^{*},n}$ and over the randomness in constructing $\mathcal{X}_{\C}$ once $\C \in \mathcal{C}_{k^{*},n}$ is picked. In the above, (a) follows by interchanging summation and expectation, and (b) holds since for every occurrence of $\C \in \C_{k^{*},n}$, there exist $q^{n(1-\tau)} - V_q(n,\rho)$ elements of $\F \backslash{\mathcal{B}(0,\rho)}$ that are in $\X_\C$. Finally, \eqref{equiprobable} and~\eqref{summation-on-probabilities} jointly imply that 
\begin{align}
(q^{n} - V_q(n,\rho)) \mathbb{P}(\mathbf{x} \in \mathcal{X}_{\C})=\sum_{\mathbf{y} \in \F^{n} \backslash{\mathcal{B}(0,\rho)}}
   \mathbb{P}(\mathbf{y} \in \mathcal{X}_{\C})=q^{n(1-\tau)} - V_q(n,\rho)\label{results2}
   \end{align}
which completes the proof.
\qed

\subsection{Proof of Lemma \ref{Point-Wise-Conditional_Probability}}{\label{D_2}}
For any $i\neq 1$ and $\x\in \F^n, \x\neq \mathbf{0}$, then
\begin{align}
  \mathbb{P}[\mathbf{c}_i &= \mathbf{x} | \mathbf{x} \in \mathcal{X}_{\C}]   \mathbb{P}[\mathbf{x} \in \mathcal{X}_{\C}] \\&\overset{(a)}{=} \mathbb{P}[[\mathbf{c}_i = \mathbf{x}]  \cap [\mathbf{x} \in \mathcal{X}_{\C}]]\\& \overset{(b)}{=}\mathbb{P}[[\mathbf{c}_i = \mathbf{x}] \:\cap\: [\mathbf{c}_i \in \mathcal{X}_{\C}]] 
  \\& \overset{(c)}{=}\mathbb{P}[\mathbf{c}_i = \mathbf{x}] \\&\overset{(d)}{=}q^{-n}
\end{align}
where (a) is directly from the definition of conditional probability \cite{kolmogoroff1956foundations}, (b) is true since the LHS requirement that $\mathbf{x} = \mathbf{c}_i$ is maintained in the RHS, (c) is true since $\mathcal{C} \subset \mathcal{X}_{\C}$, and (d) is from~\eqref{claim-1} in Lemma~\ref{uniformity}. 

Thus for $i \neq 1$ and $\mathbf{x}\neq  0$, we have that
\[\mathbb{P}(\mathbf{c}_i = \mathbf{x}|\mathbf{x} \in \mathcal{X}_{\C}) = q^{-n(1-\tau)}\frac{1 - q^{-n}V(\rho,n)}{1 -  q^{-(n-n\tau)}V(\rho,n)} = q^{-(n-n\tau)} {\color{black}\zeta(n)}, \ i \in [2: K^{*}],  \ \omega(\mathbf{x}) > \rho n\]
and the proof is concluded by noting that in the limit of large $n$, the expression \[\zeta(n)\triangleq \frac{1 - q^{-n}V(\rho,n)}{1 -  q^{-(n-n\tau)}V(\rho,n)}\] converges to $1$.
\qed

\subsection{Proof of Lemma \ref{Average-eta-x}}{\label{D_3}}
From the definition in~\eqref{Definition-eta-x}, let us recall that
 \begin{align}
     \eta_{\mathbf{x}} \triangleq \sum^{q^{k^{*}}}_{i=1} \eta_{\mathbf{x},i}
 \end{align}
 describes the number of codewords that cover $\mathbf{x} \in \mathcal{X}'_{\C}$.
 Consider a Hamming ball of radius $\rho$ centered around $\mathbf{x} \in \mathcal{X}'_{\C}$.
 \begin{itemize}
 \item Considering \eqref{probability-1}, we know that if $|\{0\} \cap \mathcal{B}(\mathbf{x},\rho)|=1$, then with probability one, $\mathbf{0}$ covers $\mathbf{x}$. Hence now the assumption that $\mathbf{x} \in \mathcal{X}'_{\C}$, contradicts the above, and thus we can conclude that $|\{0\} \cap \mathcal{B}(\mathbf{x},\rho)|=0$.
 \item Now consider some vector $\mathbf{x}'$ in $(\mathcal{B}(0,\rho)\backslash{\{0\}}) \cap \mathcal{B}(\mathbf{x},\rho)$, i.e., some vector that covers our aforementioned $\mathbf{x} \in \mathcal{X}'_{\C}$. We are interested in the probability $ \mathbb{P}(\mathbf{c}_i = \mathbf{x}'|\mathbf{x}' \in \mathcal{X}'_{\C})$ which is the probability that $\mathbf{x}'$ is equal to $\mathbf{c}_i$, for our fixed $i, i \neq 1$. From~\eqref{probability-1}, we know that this probability is equal to $q^{-n}$. Now going over all $q^{k^{*}}-1$ codewords of interest, we can conclude that our current case of interest, contributes to the sum $ \eta_{\mathbf{x}}$, by an amount equal to  $(q^{k^{*}}-1)[|(\mathcal{B}(0,\rho)\backslash{\{0\}}) \cap \mathcal{B}(\mathbf{x},\rho)| q^{-n}$.
 \item Now let us consider the dominant case where $\mathbf{x}'$ in $\mathcal{B}(\mathbf{x},\rho) \backslash{ \mathcal{B}(0,\rho)}$. In this case, directly from~\eqref{probability-1}, the aforementioned probability $ \mathbb{P}(\mathbf{c}_i = \mathbf{x}'|\mathbf{x}' \in \mathcal{X}'_{\C})$ takes the form $q^{-n(1-\tau)}{ \zeta(n)}$, and thus --- similarly to above --- yields a contribution to the sum $ \eta_{\mathbf{x}}$ by an amount equal to \[(q^{k^{*}}-1)[| \mathcal{B}(\mathbf{x},\rho) \backslash{ \mathcal{B}(0,\rho)}  | q^{-n(1-\tau)}{ \zeta(n)}. \] \qed
\end{itemize}

\subsection{Proof of Lemma \ref{Bounding_l_lemma}}{\label{D_5}}
We prove the lemma in two steps.

In the first step, after defining $\overline{\eta}_{} \triangleq \mathbb{E}(\eta_{\mathbf{x},i}|\mathbf{x} \in \mathcal{X}'_{\C}) $ and 
 $ \overline{\eta^{(2)}} \triangleq \mathbb{E}(\eta_{\mathbf{x},i} \eta_{\mathbf{x},j}|\mathbf{x} \in \mathcal{X}'_{\C})  $ for any $i,j\in [q^{k^{*}}]$, we can see that 
\begin{align}
        Var(\eta_{\mathbf{x}} |\mathbf{x} \in \mathcal{X}_{\C}) \leq (q^{k^*} -1)(q-1) \overline{\eta} (1 - \frac{q-2}{q-1} \overline{\eta_{}}) \label{Bounding-Variance-1}
    \end{align}
    since
    \begin{align}
 Var(\eta_{\mathbf{x}} |\mathbf{x} \in \mathcal{X}_{\C})  &= \mathbb{E}((\sum^{q^{k^{*}}}_{i=1} \eta_{\mathbf{x},i})^{2}|\mathbf{x} \in \mathcal{X}'_{\C} ) - \mathbb{E}^2(\sum^{q^{k^{*}}}_{i=1} \eta_{\mathbf{x},i}|\mathbf{x} \in \mathcal{X}'_{\C}) \\
&= \mathbb{E}( \sum^{q^{k^{*}}}_{i=1} \eta^2_{\mathbf{x},i} + \sum^{q^{k^{*}}}_{p.i=1,i \neq p} \eta_{\mathbf{x},i} \eta_{\mathbf{x},p}|\x \in \mathcal{X}'_{\C}) - \mathbb{E}^2(\sum^{q^{k^{*}}}_{i=1} \eta_{\x,i}|\x \in \mathcal{X}'_{\C}) \\
&= \sum^{q^{k^{*}}}_{i=1} \mathbb{E}(\eta_{\x,i} | \x \in \mathcal{X}'_{\C}) + \sum^{q^{k^{*}}}_{p,i=1, i\neq p} \mathbb{E}(\eta_{\x,i} \eta_{\x,p}| \x \in \X'_{\C}) - \mathbb{E}^2(\sum^{q^{k^{*}}}_{i=1} \eta_{\x,i}|\x \in \X'_{\C})\\ & =
q^{k^{*}} \overline{\eta} + q^{k^{*}} (q^{k^{*}} -1) \overline{\eta^{(2)}} - q^{2k^{*}} \overline{\eta}.
    \end{align}
Combining now the above with~\eqref{Bounding-Variance-1}, allows us to modify the {main claim} as 
\begin{align}
    q^{k^{*}} \overline{\eta} + q^{k^{*}} (q^{k^{*}} -1) \overline{\eta^{(2)}} - q^{2k^{*}} \overline{\eta} - (q^{k^*} -1)(q-1) \overline{\eta} (1 - \frac{q-2}{q-1} \overline{\eta_{}}) \leq 0 \label{Refocus}
\end{align}
so we now need to prove \eqref{Refocus}. To prove this, we focus on the LHS and show that 
\begin{align}
      &q^{k^{*}} \overline{\eta} + q^{k^{*}} (q^{k^{*}} -1) \overline{\eta^{(2)}} - q^{2k^{*}} \overline{\eta} - (q^{k^*} -1)(q-1) \overline{\eta} (1 - \frac{q-2}{q-1} \overline{\eta_{}})
      \\
      &\overset{(a)}{\leq }
       q^{k^{*}} \overline{\eta} + q^{k^{*}} (q^{k^{*}} -1) \overline{\eta}^{2} - q^{2k^{*}} \overline{\eta} - (q^{k^*} -1)(q-1) \overline{\eta} (1 - \frac{q-2}{q-1})
       \overline{\eta_{}}
       \\ &\overset{(b)}{\leq }
       -(q^{k^{*}+1} - 2 q^{k^{*}} -q +1) \overline{\eta} + (q^{k^{*}+1} -3 q^{k^{*}} - q +2) \overline{\eta}^2
       \\
        &\overset{(c)}{\leq } 0
\end{align}
where (a) holds because $\overline{\eta^{(2)}} \leq \overline{\eta}^{2}$, (b) follows after simple rearranging of terms, and (c) holds because $ 0 \leq\overline{\eta} \leq 1$ and because $q^{k^{*}+1} - 2 q^{k^{*}} -q +1 > q^{k^{*}+1} -3 q^{k^{*}} - q +2$ for any $q^{k^{*}}\geq 1$.

In the second step, we will prove that 
    \begin{align}
         \frac{(q^{k^*} -1)(q-1) \overline{\eta} (1 - \frac{q-2}{q-1} \overline{\eta_{}})}{\mathbb{E}(\eta_{\x}|\x \in \mathcal{X}'_{\C}) q^2} \leq 1.\label{bounding-variance-2}
    \end{align}
To see this, we note that 
\begin{align}
    \frac{(q^{k^*} -1)(q-1) \overline{\eta} (1 - \frac{q-2}{q-1} \overline{\eta_{}})}{\mathbb{E}(\eta_{\x}|\x \in \mathcal{X}'_{\C}) q^2} &\overset{(a)}{=}
    \frac{(q^{k^*} -1)(q-1)  (1 - \frac{q-2}{q-1} \overline{\eta}) }{q^{k^{*}}  q^2}\\
    &\overset{(b)}{=} (\frac{q^{k^{*}} -1}{q^{k^{*}}})  (\frac{q-1}{q^2})  (1 -\frac{q-2}{q-1} \overline{\eta})
   \\ & \overset{(c)}{\leq} 1
\end{align}
where (a) holds because $\mathbb{E}(\eta_{\x}) = \sum^{q^{k^{*}}}_{i=1} \mathbb{E}(\eta_{\x,i}) = q^{k^{*}} \overline{\eta}$, where (b) holds by rearranging terms, and where (c) holds because each multiplicative element in the RHS of (b) is non-negative and less than $1$.

Now combining the two steps by bringing together~\eqref{Bounding-Variance-1} with~\eqref{bounding-variance-2}, yields the desired~\eqref{bounding-variance}.
\qed

\section{Proof of Lemma \ref{lower-bound-lemma}}{\label{M}}
Let us define
\begin{align} \label{definitionI}
  \mathcal{I}_{(\omega(\mathbf{x}),\rho)}\triangleq |\mathcal{B}(\mathbf{x},\rho) \cap{ \mathcal{B}(0,\rho)}|
\end{align}
and let us note that
\begin{align}
    \rho n =\underset{\omega(\mathbf{x}): \mathbf{x} \in \mathcal{X}_{\C} \backslash{\mathcal{B}(0,\rho)}}{\arg\max} |\mathcal{I}_{(\omega(\mathbf{x}),\rho)}|
\end{align}
 since the distance between $\mathbf{x} \in \mathcal{X}_{\C} \backslash{\mathcal{B}(0,\rho)}$ and $\mathbf{0}$ is minimized when $\omega(\mathbf{x}) = \rho n$. We also know that
 \begin{align}
     |\mathcal{B}(\mathbf{x},\rho) \backslash{ \mathcal{B}(0,\rho)}| &= Vol(n,\rho) - |\mathcal{I}(\rho,\rho)|. \label{eq-diference}
 \end{align}
Let us focus on the case where $q=2$ and $ 0\leq \rho \leq \frac{1}{2}$, where from~\cite{danoyan2013some} we have that
 \begin{align}
|\mathcal{I}(\rho,\rho)| = \sum^{\lfloor\frac{n\rho}{2} \rfloor}_{i=0} \sum^{i}_{j=0}
{n \rho \choose i}{n-n\rho \choose j}  + \sum^{n\rho}_{i = \lfloor \frac{n \rho}{2}\rfloor +1} \sum^{n \rho-i}_{j=0} {n \rho \choose i}{ n-n \rho \choose j}.\label{eq-intersection}
 \end{align}
We know that $n\rho \leq n - n \rho$ and that
\begin{align}
    V_q(n,\rho) = \sum^{ \rho n  }_{i=0} {n \choose i}  &=\sum^{\lfloor \frac{n \rho}{2} \rfloor}_{i=0} \sum^{n \rho-i}_{j=0} { n \rho \choose i}{ n-n \rho \choose j} \\&+ \sum^{n \rho}_{i=\lfloor \frac{n\rho}{2} \rfloor +1} \sum^{n \rho-i}_{j=0} {n \rho \choose i} {n-n \rho \choose j}\label{eq-volume}
\end{align}
and thus after substituting \eqref{eq-intersection} and \eqref{eq-volume} into \eqref{eq-diference}, we conclude that
 \begin{align}
        |\mathcal{B}(\mathbf{x},\rho) \backslash{ \mathcal{B}(0,\rho)}| = \sum^{\lfloor\frac{n \rho}{2} \rfloor}_{i=0} \sum^{n \rho-i}_{j=i+1}{ n \rho \choose i}{ n-n \rho \choose j}. \label{Difference}
 \end{align}
Now considering that $0<\rho\leq \frac{1}{2}$, we have
\begin{align}
    {n \choose n \rho} &=  \sum^{\lfloor \frac{n\rho}{2} \rfloor}_{i=0} { n \rho \choose i}{ n- n \rho  \choose n\rho -i} +\sum^{ {n \rho}}_{i= \lfloor \frac{n \rho}{2} \rfloor +1} { n\rho \choose i}{ n -  n\rho \choose n\rho -i} \\&=
    \sum^{\lfloor \frac{n\rho}{2} \rfloor}_{i=0} { n \rho \choose i}{ n- n \rho  \choose n\rho -i} +\sum^{ {n \rho}}_{i= \lfloor \frac{n \rho}{2} \rfloor +1} { n\rho \choose n\rho -i}{ n -  n\rho \choose n\rho -i} \\ &=
    \sum^{\lfloor \frac{n\rho}{2} \rfloor}_{i=0} { n \rho \choose i}{ n- n \rho  \choose n\rho -i} +\sum^{  n\rho  - \lfloor \frac{n \rho}{2}\rfloor  -1}_{i= 0} { n\rho \choose i}{ n -  n\rho \choose i}.
    \label{the-little-one}
\end{align}
Now let us note that
\begin{align}
     |\mathcal{B}(\mathbf{x},\rho) \backslash{ \mathcal{B}(0,\rho)}| &\overset{(a)}{=} \sum^{\lfloor\frac{n \rho}{2} \rfloor}_{i=0} \sum^{n \rho-i}_{j=i+1}{ n \rho \choose i}{ n-n \rho \choose j}\label{f-1}\\
     &\overset{(b)}{=}\sum^{\lfloor\frac{n \rho}{2} \rfloor}_{i=0} \sum^{n \rho-i-1}_{j=i+2}  { n \rho \choose i}{ n-n \rho \choose j}\label{f-2}\\&+
     \sum^{\lfloor\frac{n \rho}{2} \rfloor}_{i=0}
     { n \rho \choose i}{ n-n \rho \choose n\rho -i} (q-1)^{n \rho} \label{f-3}\\&+   \sum^{\lfloor\frac{n \rho}{2} \rfloor  -\mathbbm{1}[{\lfloor \frac{n\rho}{2}\rfloor =\frac{n\rho}{2} ]}}_{i=0}
     { n \rho \choose i}{ n-n \rho \choose i+1}\label{f-4}\\&\overset{(c)}{\geq}
      \sum^{\lfloor \frac{n\rho}{2} \rfloor}_{i=0} { n \rho \choose i}{ n- n \rho  \choose n\rho -i} +\sum^{  n\rho  - \lfloor \frac{n \rho}{2}\rfloor  -1}_{i= 0} { n\rho \choose i}{ n -  n\rho \choose i}\label{f-5}\\&\overset{(d)}{=} {n \choose n \rho} \overset{(e)}{\geq} 2^{n H(\rho) - o(n)}
\end{align}
where (a) holds from~\eqref{Difference}, (b) follows after expanding the inner summation, (c) holds because the first summation of the RHS in (b) is non-negative, the second summation is present on the RHS of (c) after considering that $q=2$, and because the third summation of the RHS in (b) is present on the RHS of (c) after considering Lemma~\ref{nonInequalityNessecarry} found in Appendix~\ref{Op}. Furthermore (d) follows from~\eqref{the-little-one}, and (e) follows from the Stirling inequality that applies because $0<\rho \leq \frac{1}{2}$. The proof is concluded for the binary case of $q=2$.

Let us now consider the more involved case of $q>2$. For this we will need the following lemma, whose proof is found in Appendix~\ref{Opp}. \begin{Lemma}\label{Forq-intersection}
For $\mathcal{I}_{(\omega(\mathbf{x}),\rho)}\triangleq |\mathcal{B}(\mathbf{x},\rho) \cap{ \mathcal{B}(0,\rho)}|$ (cf.~\eqref{definitionI}), then 
\begin{align}
    |\mathcal{I}(\rho,\rho)| = \sum_{i+j = \rho n,\; i \leq j}{ {n-\rho n} \choose i}{ \rho n \choose j}(q-1)^{\rho n} + \theta(q)\label{intersection-argument}
\end{align}
{\color{black}where we can guarantee that $\theta(q) \leq V_q(n,\rho-1/n)$.}
\end{Lemma}
\begin{proof}
See Appendix~\ref{Opp}.
\end{proof}

With this lemma in place, we now note that
\begin{align}
     |\mathcal{B}(\mathbf{x},\rho) \backslash{ \mathcal{B}(0,\rho)}| &\overset{(a)}{=}V_q(n,\rho) - |\mathcal{I}(\rho,\rho)|
     \\ &\overset{(b)}{=}V_q(n,\rho - 1/n) + { {n-\rho n} \choose \rho n} \\&- \sum_{i+j = \rho n,\; i \leq j}{ {n-\rho n} \choose i}{ \rho n \choose j}(q-1)^{\rho n} - \theta(q)\\& \overset{(c)}{\geq}
     \sum_{i+j = \rho n,\; i > j}{ {n-\rho n} \choose i}{ \rho n \choose j}(q-1)^{\rho n}
     \\ & \overset{(d)}{=}    \sum^{\lfloor \frac{\rho n }{2} \rfloor -1 }_{j = \min\{0, 2 \rho n -n\}}{ {n-\rho n} \choose \rho n -j}{ \rho n \choose j}(q-1)^{\rho n}
\end{align}
where (a) follows from a basic set cardinality rule, (b) follows from~\eqref{intersection-argument}, (c) follows from Lemma~\eqref{Forq-intersection} which tells us that $\theta(q) \leq V_q(n, \rho - 1/n)$, and from the fact that ${ {n-\rho n} \choose \rho n}  =   \sum_{i+j = \rho n}{ {n-\rho n} \choose i}{ \rho n \choose j}(q-1)^{\rho n}$, and (d) holds since $ j\geq 0$ and at the same time $i \leq n - \rho n$ which gives $j \geq 2\rho n -n$.

Now from Stirling's bound, we know that for any $j \in [\min\{0,2 \rho n - n \} , \lfloor\frac{\rho n}{2}\rfloor - 1]$, the following holds 
\begin{align}
    { {n-\rho n} \choose \rho n -j}{ \rho n \choose j}& \geq \sqrt{\frac{n-\rho n}{8 (\rho n - j) (n-2 \rho n + j)}} 2^{n{H((\rho n - j)/ (n - \rho n))}}\\& \times
    \sqrt{\frac{\rho n}{8  (j) (\rho n-j)}} 2^{n{H(j/\rho n)}} \\&=2^{n [H((\rho n -j)/\rho n) + H((\rho n -j)/(n - \rho n))]- o(n)}.\label{exponent-bound}
\end{align}
After defining 
\begin{align}
    \kappa \triangleq \frac{\rho n - j}{\rho n}
\end{align}
the exponent in~\eqref{exponent-bound} takes the form
\begin{align}
    n[H(\kappa) + H(\kappa \frac{\rho }{1 - \rho})] - o(n). \label{exponent}
\end{align}
With this exponent in place, let us consider a large $n$ and let us assume without loss of generality\footnote{This assumption, along with the assumption that $n\rho$ is an integer, has no impact on the result, because any non-integer residual will vanish in importance as $n$ increases.} that $n \rho^{2} \in \mathbb{N}$. For the case where $ 0<\rho \leq \frac{1}{2} $, let us set $j$ such that $ j=n \rho^{2} \leq \lfloor\frac{n \rho }{2}\rfloor - 1$, in which case we get $\kappa = 1-\rho$. On the other hand, for the case where $ \frac{1}{2}<\rho \leq -1/2 + \sqrt{5}/2$, let us set $j$ such that $ 2 \rho n - n \leq j=n \rho^{}(1 - \rho) \leq \lfloor\frac{n \rho }{2}\rfloor - 1$, in which case $\kappa = \rho$. In each case, $\kappa$ is plugged in~\eqref{exponent}, and after utilizing~\eqref{exponent-bound}, we see that
\begin{align}
       \sum^{\lfloor \frac{\rho n }{2} \rfloor -1 }_{j = \min\{0, 2 \rho n -n\}}{ {n-\rho n} \choose \rho n -j}{ \rho n \choose j}(q-1)^{\rho n}  \geq q^{n H_q(\rho) - o(n)}
\end{align}
which holds for the range $0<\rho \leq -1/2 + \sqrt{5}/2$ which is the union of the above two regions in $\rho$.
\qed

\subsection{Proof of Proposition \ref{Proposition-T-1}}{\label{N}}
For $h(x) \triangleq -x \log_{q}(x)$, we know that
\begin{align}
 h(x)\leq H_q(x),\: 0\leq x \leq 1-1/q
\end{align}
and thus that
\begin{align}
 H_q^{-1}(x)\leq     h^{-1}(x)\label{inverse-bound} .
\end{align}
We also know that if $y =x \ln(x)$, then $x = e^{W(y)}$ where $W(.)$ is the Lambert function. Also note that since $h(x) = - \log_{q}(e) x \ln(x)$, we have that
\begin{align}
    h^{-1}(x) = e^{W(- \ln(q) x)}.\label{inverse-function}
\end{align}
Furthermore, for $c>0$ being a positive real number, we have that
\begin{align}
    \lim_{T \rightarrow{ \infty}} T e^{W(-c/T)} & \overset{(a)}{=} \lim_{T \rightarrow{\infty}} \frac{e^{W(-c/T)}}{1/T} \\
    &\overset{(b)}{=}\lim_{T \rightarrow{\infty}}\frac{e^{W(-c/T)} \frac{1}{-c/T + e^{W(-c/T)}} c T^{-2}}{- T^{-2}}\\
    &\overset{(c)}{=} \lim_{T \rightarrow{\infty}} \frac{c e^{W(-c/T)}}{c/T - e^{W(-c/T)}} \\
    &\overset{(d)}{=} \lim_{T \rightarrow{\infty}} \frac{c T e^{W(-c/T)}}{c - T e^{W(-c/T)}}\\
    & \overset{(e)}{=} \frac{\lim_{T \rightarrow{\infty}} c T e^{W(-c/T)}}{\lim_{T \rightarrow{\infty}} c - T e^{W(-c/T)}}\label{limit-eqiation}
\end{align}
where (a), (c) and (d) follow by basic algebraic rearranging, (b) follows from L’Hopital’s rule, and {\color{black}(e) follows from the Algebraic Limit Theorem.} 
{\color{black}The above implies that 
\begin{align}
     \lim_{T \rightarrow{ \infty}} T e^{W(-c/T)} = 0 \label{zero-limit}
\end{align}
}which allows us to conclude that
\begin{align}
    \lim_{T \rightarrow{ \infty}}
  T H_q^{-1}((c/T) & \overset{(a)}{\leq}    \lim_{T \rightarrow{ \infty}}
  Th^{-1}(c/T) \\
  &\overset{(b)}{\leq} \lim_{T \rightarrow{ \infty}} Te^{W(- \ln(q) c/T)} \\ &
  \overset{(c)}{=}0
\end{align}
where (a) follows from~\eqref{inverse-bound}, (b) from~\eqref{inverse-function}, and (c) from~\eqref{zero-limit}. Thus the proof is concluded.
\qed

\subsection{Proof of Proposition \ref{Derivative}}\label{O}
Starting from
\begin{align}
  c/T = H_q(f/T)  \label{reverse-formula}
\end{align}
we take the derivative with respect to $T$ on both sides, to get 
{\color{black}
\begin{align}
    c = \log_q(\frac{f/T}{1 - f/T} (q-1)) ({\frac{\partial f}{\partial T}} T - f). \label{derivated}
\end{align}
}After applying~\eqref{reverse-formula} into \eqref{derivated}, and after some basic algebraic rearranging, the proof is concluded.
\qed

{\color{black}\subsection{Proof of Lemma~\ref{Forq-intersection}} \label{Opp}
Let us consider two vectors $\mathbf{a}, \mathbf{b} \in \F^{n}$, where $\omega(\mathbf{a}) = \rho n$ and where $\mathbf{b} \in \mathcal{B}(\mathbf{0},\rho) \cap \mathcal{B}(\mathbf{a},\rho)$. Let 
\begin{align}
    \mathcal{B} &\triangleq\{\mathbf{b}(i), \in [n] \ : \ \mathbf{a}(i) =0 \ \& \   \mathbf{b}(i) \neq 0 \}\label{infront-zero-element}\\
    \mathcal{C} &\triangleq \{\mathbf{b}(i), \in [n] \ : \ \mathbf{a}(i) \neq 0 \ \& \  \mathbf{b}(i)=  \mathbf{a}(i) \} \label{infront-non-zero-element}
\end{align}
and let $x \triangleq |\mathcal{B}|, y \triangleq |\mathcal{C}|$. 
Let $b_j, j\in[x]$ denote the $j$th element of $\mathcal{B}$, and let $c_j, j\in[y]$ denote the $j$th element of $\mathcal{C}$. Note that ordering does not matter. 
Let $a_j, j\in[\rho n]$ be the non-zero elements of $\mathbf{a}$, and without loss of generality, let 
\begin{align}
    \mathbf{a} & = [0,\hdots,0,a_1,a_2,\hdots, a_{y},\hdots,a_{\rho n}]
\end{align}
as well as let 
\begin{align}
    \mathbf{b}  = [b_1,b_2,\hdots,b_x,0,0,\hdots,0, c_1,c_2,\hdots, c_{y},0,\hdots,0].
\end{align}
Since $\omega(\mathbf{b}) \leq \rho n$ and $d(\mathbf{a},\mathbf{b}) \leq \rho n$, we have that
\begin{align}
    x+ y \leq \rho n\label{constraint-1}\\
    x  \leq y\label{constraint-2}
\end{align}
and we have that for any fixed pair $(x,y)\in[\rho n]^{2}$, there are $ {{n - \rho n} \choose x}{ \rho n \choose y} (q-1)^{x+y} $ such points (that satisfy that given $(x,y)$) in $\mathcal{B}(\mathbf{0},\rho) \cap \mathcal{B}(\mathbf{a},\rho)$. Therefore, accumulating over all $(x,y)$ for which $x+y = \rho n$, the overall intersection gains cardinality $|\mathcal{B}(\mathbf{0},\rho) \cap \mathcal{B}(\mathbf{a},\rho)| = \sum_{i+j = \rho n,\; i \leq j}{ {n-\rho n} \choose i}{ \rho n \choose j}(q-1)^{\rho n}$, while when $x+y < \rho n$ --- due to~\eqref{constraint-2} --- this same accumulated intersection has at most $V_q(n,\rho-1/n)$ points. This concludes the proof.

\qed
}

\subsection{Statement and Proof of Lemma~\ref{nonInequalityNessecarry}}\label{Op}
\begin{Lemma} \label{nonInequalityNessecarry}
If $0<\rho \leq 1/2 $ and $\rho n \in \mathbb{N}$, then
\begin{align}
    \sum^{\lfloor\frac{n \rho}{2} \rfloor  -\mathbbm{1}[{\lfloor \frac{n\rho}{2}\rfloor =\frac{n\rho}{2} ]}}_{i=0}
     { n \rho \choose i}{ n-n \rho \choose i+1} \geq  \sum^{  n\rho  - \lfloor \frac{n \rho}{2}\rfloor  -1}_{i= 0} { n\rho \choose i}{ n -  n\rho \choose i}.
\end{align}
\end{Lemma}
\begin{proof}
We solve the problem by considering the following cases. 

\emph{Case 1: ($2 \ | \ \rho n$).} We note that if $2|\rho n$, then $\lfloor \frac{n\rho}{2}\rfloor = \frac{n \rho}{2}$, which gives 
    \begin{align}
        \sum^{\frac{n \rho}{2}   -1}_{i=0}
     { n \rho \choose i}{ n-n \rho \choose i+1} \geq  \sum^{
    \frac{n\rho}{2}   -1}_{i= 0} { n\rho \choose i}{ n -  n\rho \choose i}.
    \end{align}
    Since $0<\rho  \leq \frac{1}{2}$, we have that $i+1 \leq( n - n\rho )/2, \ \forall i \in [0: n \rho /2-1]$, which gives 
    \begin{align}
          { n-n \rho \choose i+1}  \geq  { n-n \rho \choose i}
      \end{align}
      to conclude the proof for this case.

\emph{Case 2: ($2\nmid\rho n$).} {\color{black}Here we note that having $2\nmid\rho n$, implies $\lfloor \frac{n\rho}{2}\rfloor = \frac{n \rho}{2} - 0.5$} which gives 
    \begin{align}
        \sum^{\frac{n \rho}{2} - 0.5}_{i=0}
     { n \rho \choose i}{ n-n \rho \choose i+1} \geq  \sum^{\frac{n \rho}{2} - 0.5}_{i= 0} { n\rho \choose i}{ n -  n\rho \choose i}.
    \end{align}
    For this case, we consider the following two subcases. 

{\color{black}    \emph{Case 2a: ($n$ is odd).} We first note that if $n$ is an odd number, then $n -n\rho$ is an even number. Also since $\rho n  \in \mathbb{N}$, $0<\rho<1/2$ and $n \rho \leq \frac{n-1}{2}$, then $2(\frac{n\rho}{2} + 0.5) \leq n -n \rho$. Thus as before we can say that $\forall i \in [0 : n\rho /2 - 0.5], i +1 \leq (n - n\rho)/2$, which gives
         \begin{align}
          { n-n \rho \choose i+1}  \geq  { n-n \rho \choose i}
      \end{align}
      which in turn concludes the proof for this case. 
    
\emph{Case 2b: ($n$ is even).}
        If $n$ is even then $n -n\rho$ is odd, and thus we can again say that having $0< \rho \leq 1/2 $ gives $i +1 \leq (n - n\rho +1)/2, \ \forall i \in [0 : n\rho /2 - 0.5]$, which gives 
        \begin{align}
          { n-n \rho \choose i+1}  \geq  { n-n \rho \choose i}
      \end{align}
which in turn completes the proof for this final case also.
}
\end{proof}
\bibliographystyle{ieeetr}
\bibliography{ref}
\end{document}